\pdfoutput=1
\RequirePackage{ifpdf}
\ifpdf 
\documentclass[pdftex]{sigma}
\else
\documentclass{sigma}
\fi

\def\wh{\widehat}

\def\le{\left}
\def\ri{\right}

\def\1{{\mathbf{1}}}

\def \pa{\partial}
\def\C{{\mathbb C}}
\def\R{{\mathbb R}}
\def\N{{\mathbb N}}

\renewcommand{\Re}{\operatorname{Re}}
\renewcommand{\Im}{\operatorname{Im}}
\newcommand{\inte}{\operatornamewithlimits{Int}}
\newcommand{\exte}{\operatornamewithlimits{Ext}}

\numberwithin{equation}{section}

\newtheorem{Theorem}{Theorem}[section]

\newtheorem{Lemma}[Theorem]{Lemma}
\newtheorem{Proposition}[Theorem]{Proposition}
\newtheorem{problem}[Theorem]{Riemann--Hilbert Problem}
 { \theoremstyle{definition}
\newtheorem{Definition}[Theorem]{Definition}

\newtheorem{Remark}[Theorem]{Remark} }

\def\Xint#1{\mathchoice
{\XXint\displaystyle\textstyle{#1}}%
{\XXint\textstyle\scriptstyle{#1}}%
{\XXint\scriptstyle\scriptscriptstyle{#1}}%
{\XXint\scriptscriptstyle\scriptscriptstyle{#1}}%
 \int}
\def\XXint#1#2#3{{\setbox0=\hbox{$#1{#2#3}{\int}$ }
\vcenter{\hbox{$#2#3$ }}\kern-.6\wd0}}

\def\dashint{\Xint-}

\newcommand{\opnorm}[1]{{\left\vert\kern-0.25ex\left\vert\kern-0.25ex\left\vert #1
 \right\vert\kern-0.25ex\right\vert\kern-0.25ex\right\vert}}

\def\d{{\rm d}}
\def \lb{\lambda}

\def \pa{\partial}

\def\l{\lambda}

\begin{document}
\allowdisplaybreaks

\newcommand{\arXivNumber}{1802.01153}

\renewcommand{\thefootnote}{}

\renewcommand{\PaperNumber}{091}

\FirstPageHeading

\ShortArticleName{Painlev\'e IV Critical Asymptotics for Orthogonal Polynomials in the Complex Plane}

\ArticleName{Painlev\'e IV Critical Asymptotics\\ for Orthogonal Polynomials in the Complex Plane\footnote{This paper is a~contribution to the Special Issue on Painlev\'e Equations and Applications in Memory of Andrei Kapaev. The full collection is available at \href{https://www.emis.de/journals/SIGMA/Kapaev.html}{https://www.emis.de/journals/SIGMA/Kapaev.html}}}

\Author{Marco BERTOLA~$^{\dag\ddag}$, Jos\'e Gustavo ELIAS REBELO~$^{\dag}$ and Tamara GRAVA~$^{\dag\S}$ }

\AuthorNameForHeading{M.~Bertola, J.G.~Elias Rebelo and T.~Grava}

\Address{$^{\dag}$~Area of Mathematics, SISSA, via Bonomea 265 - 34136, Trieste, Italy}
\EmailD{\href{mailto:bertola@sissa.it}{bertola@sissa.it}, \href{mailto:jgerebelo@gmail.com}{jgerebelo@gmail.com}, \href{mailto:grava@sissa.it}{grava@sissa.it}}
\Address{$^{\ddag}$~Department of Mathematics and Statistics, Concordia University,\\
\hphantom{$^{\ddag}$}~1455 de Maisonneuve W., Montr\'eal, Qu\'ebec, Canada H3G 1M8}
\Address{$^\S$~School of Mathematics, University of Bristol, UK}

\ArticleDates{Received February 06, 2018, in final form August 14, 2018; Published online August 30, 2018}

\Abstract{We study the asymptotic behaviour of orthogonal polynomials in the complex plane that are associated to a certain normal matrix model. The model depends on a~parameter and the asymptotic distribution of the eigenvalues undergoes a transition for a~special value of the parameter, where it develops a corner-type singularity. In the double scaling limit near the transition we determine the asymptotic behaviour of the orthogonal polynomials in terms of a solution of the Painlev\'e~IV equation. We determine the Fredholm determinant associated to such solution and we compute it numerically on the real line, showing also that the corresponding Painlev\'e\ transcendent is pole-free on a semiaxis.}


\Classification{34M55; 34M56; 33C15}

\begin{flushright}
\begin{minipage}{80mm}
{\it To the memory of Andrei Kapaev, a master of Painlev\'e equation, a very generous colleague and a special friend.}
\end{minipage}
\end{flushright}

\renewcommand{\thefootnote}{\arabic{footnote}}
\setcounter{footnote}{0}

\section{Introduction}\label{Section: Introduction}

In this work we consider the critical asymptotic behaviour for the orthogonal polynomials that are associated to the matrix model described here. Over the set of normal $n\times n$ matrices $\mathcal N_n := \{M \colon [ M, M^{\star}] = 0 \} \subset \operatorname{Mat}_{n \times n} (\mathbb{C})$ we consider the measure of the form
\begin{gather}\label{random_matrix_density}
M \to \frac{1}{\mathcal{Z}_{n,N}}{\rm e}^{-N \operatorname{Tr} ( W (M))} \mathrm{d} M, \qquad \mathcal{Z}_{n,N} = \int_{\mathcal{N}_{n}} {\rm e}^{-N \operatorname{Tr}( W (M) )} \mathrm{d} M,
\end{gather}
where $\d M$ stands for the induced volume form on $\mathcal N_n$, invariant under conjugation by unitary matrices.

Since normal matrices are diagonalizable by unitary transformations, the probability density~\eqref{random_matrix_density} can be reduced to the form~\cite{Mehta}
\begin{gather}\label{prob}
\dfrac{1}{Z_{n,N}} \prod_{i<j} | \lambda_i - \lambda_j|^2 {\rm e}^{-N\sum\limits_{j=1}^{n}{W (\lambda_i)}} \mathrm{d} A (\lambda_1) \cdots \mathrm{d} A( \lambda_n),
\end{gather}
where $\lambda_j$ are the complex eigenvalues of the normal matrix $M$ \cite{ChauZaboronsky} and $\d A(\lb)=\d\Re(\lb)\d\Im(\lb)$ is the area element in the complex plane.

Associated to the above data we consider the sequence of (monic) \textit{orthogonal polynomials} $p_n (z) = z^n + \cdots$,
\begin{gather}\label{eq: def orth pol}
\int_{\mathbb{C}} p_n (\lambda) \overline{p_m (\lambda)} {\rm e}^{- N W (\lambda)} \mathrm{d} A (\lambda) = h_{n, N} \delta_{n, m}, \qquad h_{n, N} > 0 , \qquad n, m = 0, 1, \dots,
\end{gather}
where $W \colon \mathbb{C} \rightarrow \mathbb{R}$ is the external potential.

Likewise in the standard GUE, the orthogonal polynomial $p_n(z)$ is also the expectation value of the characteristic polynomial of the random normal matrix $M$: $p_n(z) = \langle \det (z\1 - M) \rangle$, where the expectation is with respect to the measure~\eqref{random_matrix_density} (see, e.g.,~\cite{Elbau}). The statistical quantities related to the eigenvalues of matrix models can be expressed in terms of the associated orthogonal polynomials, $p_{n} (\lambda)$. In particular, the probability measure (\ref{prob}) can be written as a~determinantal point field (sometimes also called ``process'') with Christoffel--Darboux kernel
\begin{gather*}
K(\lambda,\eta) = {\rm e}^{- N ( W (\lambda) + W (\eta)) / 2} \sum_{j=0}^{n-1} \frac{1}{h_{j, N}} p_{j} (\lambda) \overline{p_{j} (\eta)}.
\end{gather*}
The average density of eigenvalues is
\begin{gather*}
\rho_{n,N} (\lambda) = \frac{1}{n} K(\lambda,\lambda)
\end{gather*}
and in the limit
\begin{gather*}
n \to \infty ,\qquad N\to \infty ,\qquad \frac{N}{n} \to \frac{1}{T} ,
\end{gather*}
it converges to the unique probability measure, $\mu^*$, in the plane which minimizes the functional~\cite{ElbauFelder,HedMak}
\begin{gather*}
I (\mu) = \iint \log |\lambda - \eta|^{-1} \mathrm{d} \mu (\lambda) \mathrm{d} \mu (\eta) + \frac{1}{T} \int W (\lambda) \mathrm{d} \mu (\lambda).
\end{gather*}
The functional $I (\mu)$ is the Coulomb energy functional in two dimensions and the existence of a~unique minimizer is a well-established fact under mild assumptions on the potential~$W (\lambda)$~\cite{SaffTotik}. If $W (\lambda)$ is twice continuously differentiable then the equilibrium measure is given by
\begin{gather*}
\mathrm{d} \mu^* (\lambda) =\frac 1{4\pi T} \Delta W (\lambda) \chi_{D} (\lambda) \mathrm{d} A (\lambda) ,
\end{gather*}
where $\chi_{D}$ is the characteristic function of the compact support set $D = \operatorname{supp} ( \mu^{*} )$ and $\Delta=4\partial_{ \lb}\bar{\partial}_{ \lb}$. For example, if $W (\lambda) = |\lambda|^2$, one has the Ginibre ensemble~\cite{Ginibre} and the measure $\mathrm{d} \mu^* (\lambda)$ is the uniform measure on the disk of radius $\sqrt{T}$. Following~\cite{BGM}, we consider the external potential to be of the form
\begin{gather}\label{eq: init pot s}
W (\lambda) = |\lambda|^{2d} - t \lambda^d - \bar{t} \bar{\lambda}^d, \qquad \lambda \in \mathbb{C}, \qquad d\in \N, \qquad t \in \C^*.
\end{gather}
The case $d=1$ is reducible to Hermite polynomials and we will consider only $d\geq 2$. Up to a~rotation $\lambda \mapsto {\rm e}^{{\rm i}\theta} \lambda$ we can assume without loss of generality that $t \in \R_+$, leading to the external potential and orthogonality measure
\begin{gather}\label{eq: red extern pot}
W (\lambda) = |\lambda|^{2d} - t \big( \lambda^d + \bar{\lambda}^d \big) , \qquad \lambda \in \mathbb{C}, \qquad {\rm e}^{- N ( |\lambda|^{2d} - t ( \lambda^d + \bar{\lambda}^d ) )} \mathrm{d} A (\lambda).
\end{gather}
There remains a residual discrete rotational $\mathbb{Z}_{d}$-symmetry which allows us to further reduce the problem as observed in~\cite{BM} and~\cite{Etingof_Ma}. The potential~$W$ in~\eqref{eq: red extern pot} can be written as $W (\lambda) = \frac{1}{d} Q \big( \lambda^{d} \big)$, and hence the equilibrium measure for~$W$ can be obtained from the equilibrium measure associated to~$Q$ by an unfolding procedure. Considering the particular case of~\eqref{eq: init pot s}, the potential~$Q$ can be rewritten as
\begin{gather*}
Q (u) = d | u |^2 - d t ( u + \bar{u} ) = d | u - t |^2 - t^2 d ,
\end{gather*}
which corresponds to the Ginibre ensemble whose equilibrium measure is the normalized area measure of the disk centered at $u=t$ and of radius $t_c$
\begin{gather*}
| u - t| \leq t_c , \qquad t_c = \sqrt{\frac{T}{d}}.
\end{gather*}
Pulling back to the $\lambda$-plane, we obtain the equilibrium measure of $W$
\begin{gather*}
\mathrm{d} \mu_{W} = \frac{d}{\pi t^2_c} |\lambda|^{2 ( d - 1 )} \chi_{D} \mathrm{d} A (\lambda) ,
\end{gather*}
where $\mathrm{d} A$ is the area measure and $\chi_{D}$ is the characteristic function of the domain $D$ whose boundary is a {\it lemniscate}
\begin{gather}\label{eq: domain support set}
D := \big\{ \lambda \in \mathbb{C},\, \big| \lambda^d - t \big| \leq t_c \big\}.
\end{gather}
\begin{figure}[t]\centering
\includegraphics[width=4cm]{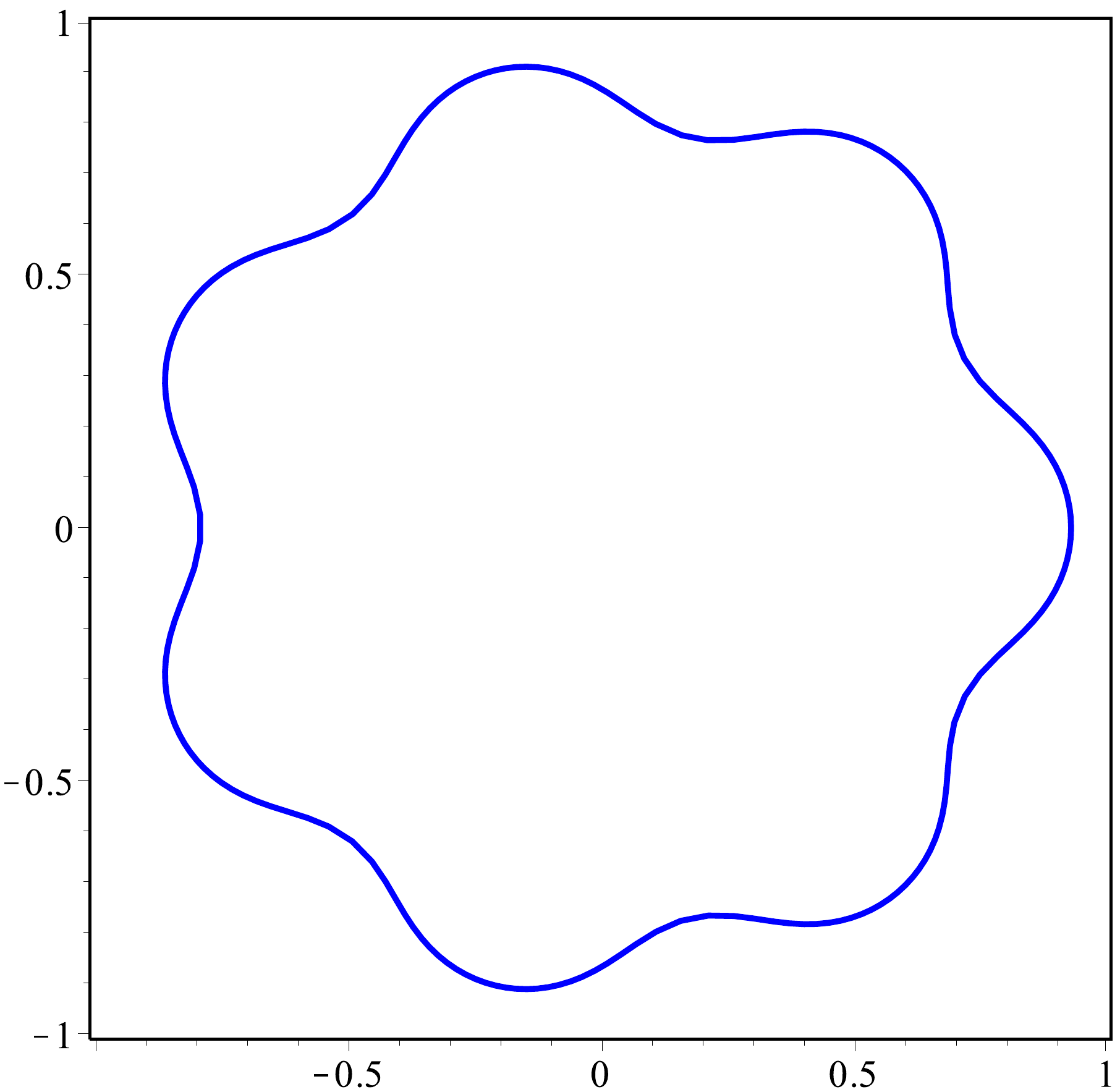}\quad
\includegraphics[width=4cm]{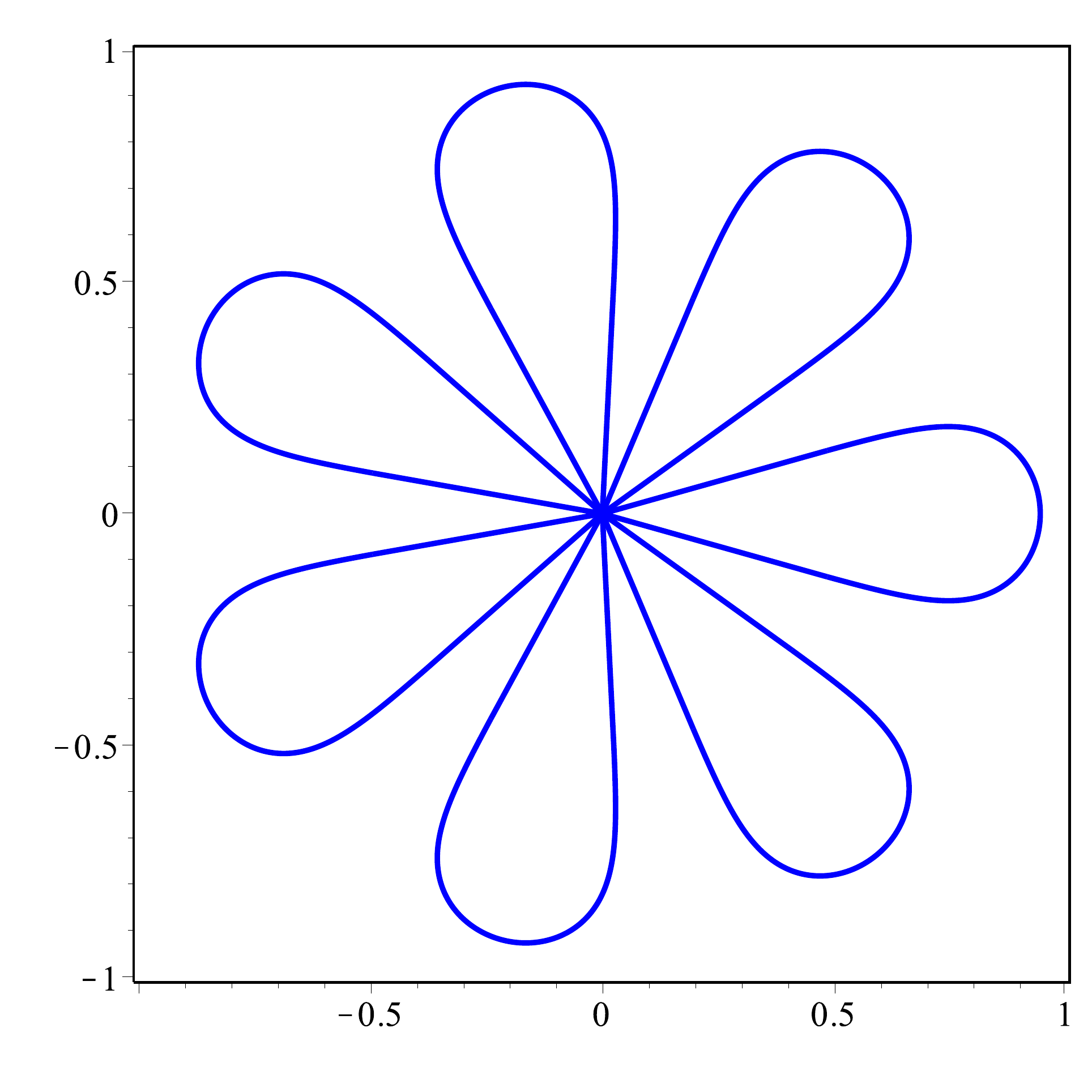}\quad
\includegraphics[width=4cm]{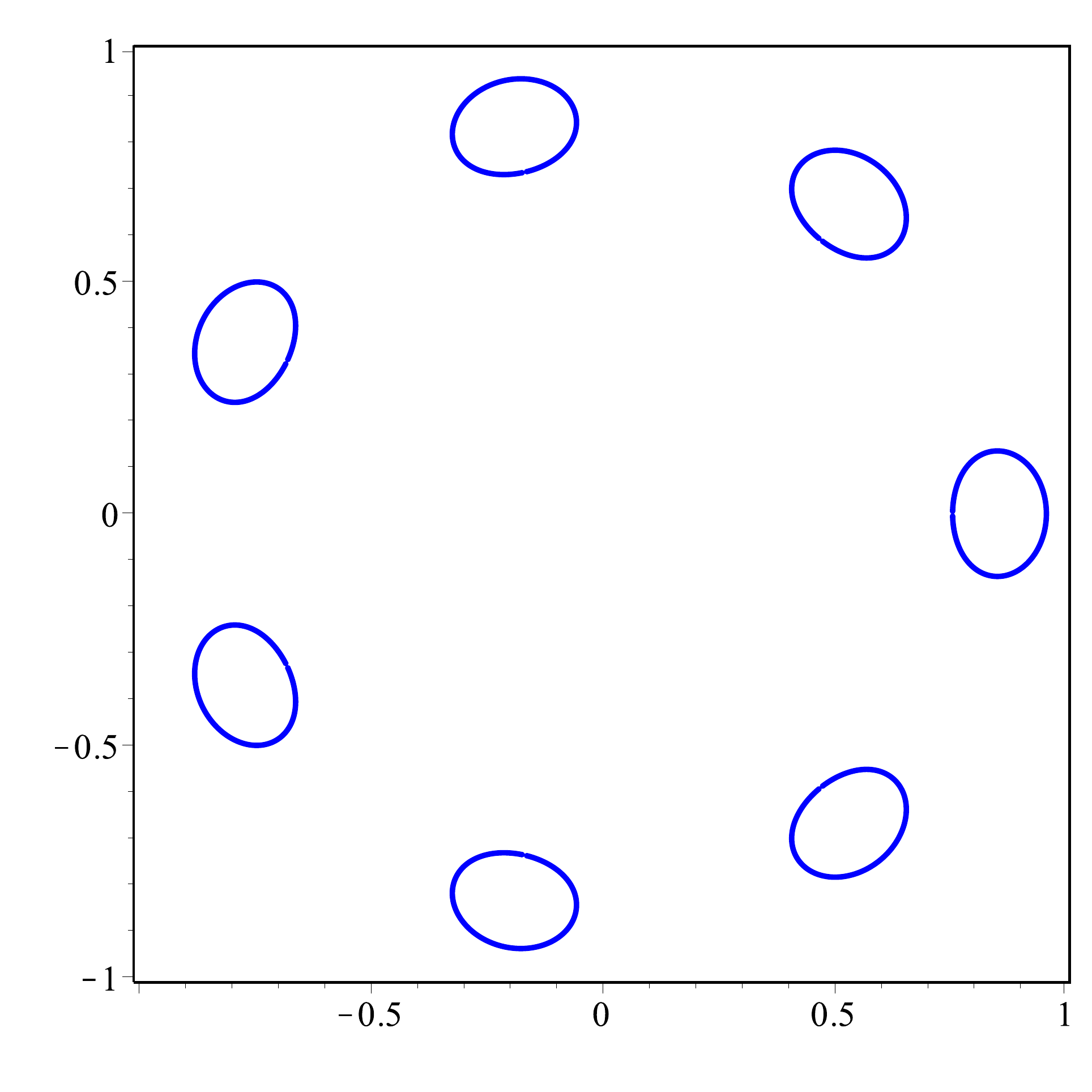}
\caption{Domain $D$ for different values of $t$: on the left $t < t_{c}$, in the center $t=t_c$ and on the right $ t > t_{c}$.}\label{fig:domain_D}
\end{figure}
A plot of the domain $D$ for different values of $t$ is shown in Fig.~\ref{fig:domain_D}. Depending on the radius $t_{c}$, the domain $D$ is either connected for $t \leq t_{c}$ or has $d$ connected components for $ t > t_{c}$. The cases $t > t_{c}$ and $t < t_{c}$ were analysed in \cite{BGM}. In this manuscript we study the critical regime $t = t_{c}$. Boundary points where $\Delta W>0$ have been classified by Sakai \cite{Sakai} and fall in the category of regular points, cusps or double points. Details about the application of Sakai's theory to the setting when $\Delta W > 0$ near a singular point can be found in several papers, e.g.,~\cite{LeeMakarov}.
The classification of boundary points when $\Delta W=0$ is still an open problem to the best of our knowledge. In this manuscript we take a first step in considering a specific example when $\Delta W=0$. In particular we will study the asymptotic behaviour of the orthogonal polynomials associated to this model.

The asymptotic analysis of orthogonal polynomials when the measure is of planar type, is much harder than the orthogonality on contours. When the orthogonality planar measure is not varying there are several studies like in \cite{SS} or \cite{GPSS}. In a very recent work \cite{HW} the asymptotic of orthogonal polynomials for a varying planar measure was obtained in the regular case. Critical regimes have been considered in \cite{Lee}, in \cite{BS, KT} studying a normal matrix model with cut-off, and in~\cite{AKMW} where new types of determinantal point fields, have emerged. The equilibrium problem for a class of potentials to which our case belongs, was also recently considered in~\cite{AmSeo}.

In this paper we derive the asymptotic behaviour for large $n$ of the orthogonal polynomials associated to the exponential weight \eqref{eq: red extern pot} in the critical regime $t_{c}$. The zeros of the polynomials accumulate on the Szeg\"o curve that was first observed in relation to the zeros of the Taylor polynomials of the exponential function~\cite{Szego}. For large $n$ we obtain the expansion of the orthogonal polynomials in terms of a special solution of a Painlev\'e~IV transcendent. We show that such solution is pole free on the negative semi-axis. We obtain the Fredholm determinant associated to such solution and compute it numerically on the real line. Painlev\'e IV critical asymptotic was observed for orthogonal polynomials with deformed Laguerre weights \cite{DM}. Finally we remark that normal matrix models are strictly related to the 2-dimensional Toda lattice, see, e.g., \cite{Teo,Zabrodin}. It would be interesting to study the critical case considered in this paper in this perspective as done in \cite{BaMa} or \cite{CGM, CV, DM} for the one-dimensional Toda lattice.

\subsection{Zeros of orthogonal polynomials: statement of the result}

The zeros of the orthogonal polynomials \eqref{eq: def orth pol} in the critical case $t = t_{c}$ concentrate on the curve~$\hat{\mathcal{C}}$ that is described below. Let us introduce the function
\begin{gather}\label{hat_phi_r}
\hat{\varphi} (\lambda) = \log \big( t_c - \lb^d \big) + \dfrac{\lb^d}{t_c} - \log t_c,
\end{gather}
and consider the level curve $\hat{\mathcal{C}}$
\begin{gather}\label{Gamma}
\hat{\mathcal{C}}: = \left\{ \lb\in\mathbb{C},\, \Re \hat{\varphi} (\lambda) = \log\dfrac{\big| t_c - \lambda^d \big|}{t_c} +\dfrac{\Re\lambda^d}{t_c} =0, \, \big| \lambda^d - t_c \big| \leq t_c \right\}.
\end{gather}
Observe that the level curve $\hat{\mathcal{C}}$ consists of a closed contour contained in the set $D$ defined in~(\ref{eq: domain support set}) with $t=t_c$ because it is the pullback under the map $z = 1- \frac {\l^d}{t_c}$ of the celebrated (closed) Szeg\"o curve $\big|z{\rm e}^{1-z}\big|=1$~\cite{Szego} (see Fig.~\ref{Gammar_s}).
\begin{figure}[t]\centering
\includegraphics[width=5.5cm]{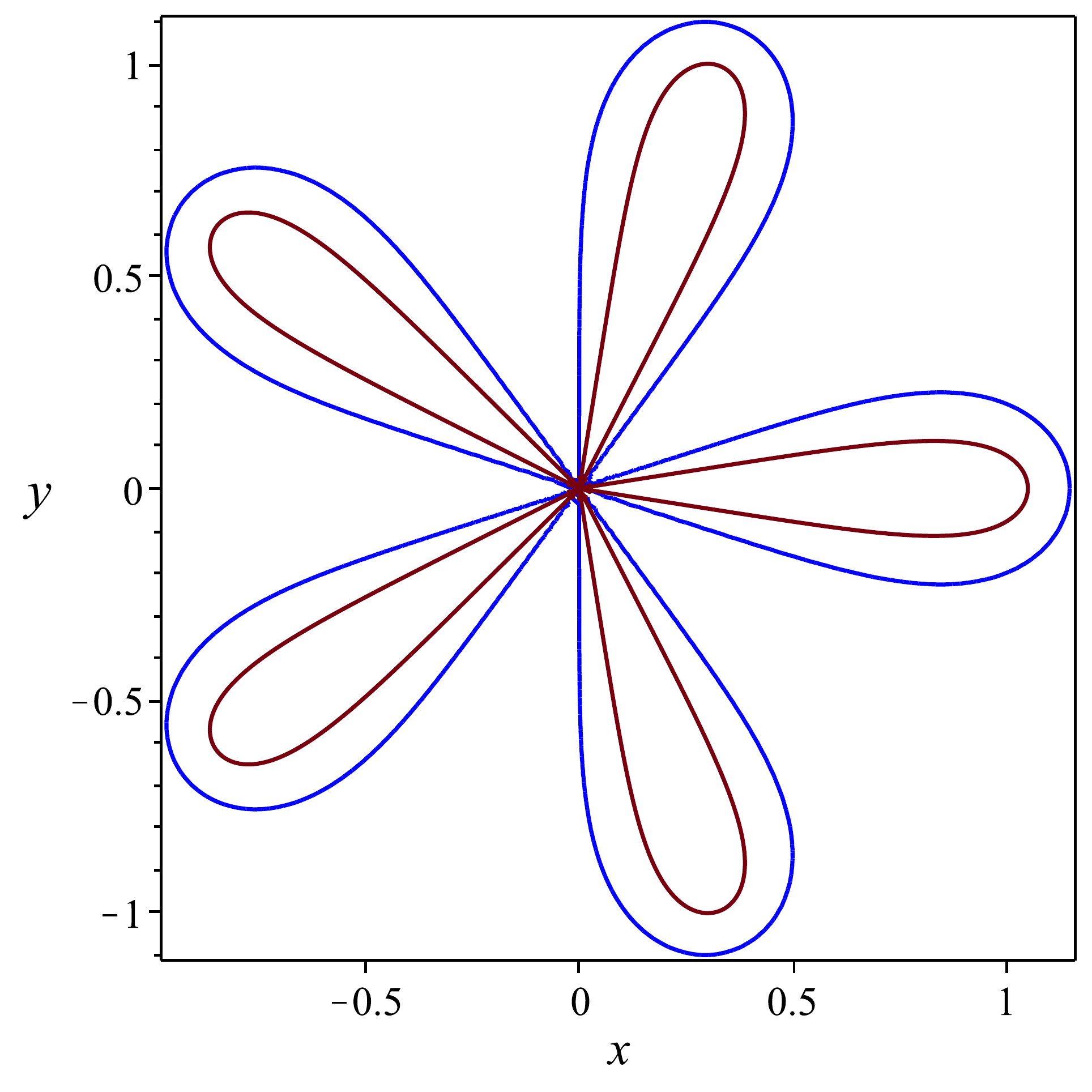}
\caption{The boundary of the domain $D$ that describes the support of the eigenvalues distribution in black and the level curve $\hat{\mathcal{C}}$ in red at the critical time $t_c$ and for $d=5$.}\label{Gammar_s}
\end{figure}

Let us define the measure $\hat{\nu}$, associated with this family of curves, given by
\begin{gather}\label{nur}
\mathrm{d} \hat{\nu} (\lambda) = \dfrac{1}{2\pi {\rm i} d} \mathrm{d} \hat{\varphi} (\lambda),
\end{gather}
and supported on $\hat{\mathcal{C}}$.
\begin{Lemma}\label{lemma1}
The a-priori complex measure $\mathrm{d} \hat{\nu}$ in \eqref{nur} is a probability measure on the contour~$\hat{\mathcal{C}}$ defined in~\eqref{Gamma}.
\end{Lemma}
When $t = t_{c}$, the zeros of the orthogonal polynomials~\eqref{eq: def orth pol} behave in the following way.

\begin{Theorem}\label{theorem1}The zeros of the polynomials $p_n (\lambda)$ defined in~\eqref{eq: def orth pol} for $t = t_c = \sqrt{T/d}$, behave as follows
\begin{itemize}\itemsep=0pt
\item for $n=kd+d-1$, let $\omega={\rm e}^{\frac{2\pi {\rm i}}{d}}$. Then $t^{\frac{1}{d}},\omega t^{\frac{1}{d}},\dots,\omega^{d-1} t^{\frac{1}{d}}$ are zeros of the polynomials $p_{kd+d-1}$ with multiplicity $k$ and $\lb=0$ is a zero with multiplicity $d-1$.
\item for $n = kd + \ell$, $\ell = 0,\dots,d-2$ the polynomial $p_n (\lambda)$ has a zero in $\lb=0$ with multiplicity~$\ell$ and the remaining zeros in the limit $n$, $N \to \infty$ such that
\begin{gather*}
N=\dfrac{n-\ell}{T},
\end{gather*} accumulate on the level curve $\hat{\mathcal{C}}$ defined in \eqref{Gamma}. The measure $\hat{\nu}$ in \eqref{nur} is the weak-star limit of the normalized zero counting measure $\nu_n$ of the polynomials~$p_n$ for $n = k d + \ell$, $\ell = 0, \dots, d-2$.
\end{itemize}
\end{Theorem}

Our next result gives strong uniform asymptotics as $n \to \infty $ for the polynomials $p_n (\lambda)$ in the complex plane. We consider the double scaling limit
\begin{gather*}
t\to t_c,\qquad k\to \infty,\qquad k=\dfrac{n - \ell}{d},
\end{gather*}
in such a way that
\begin{gather*}
\lim_{k\to\infty,\, t \to t_c} \sqrt{k} \left( \dfrac{t^2}{t_c^2} - 1 \right) = \mathcal{S} ,
\end{gather*}
with $\mathcal{S}$ in compact subsets of the real line. In the description of the asymptotic behaviour of the orthogonal polynomials, $p_n (\lambda)$, in this double-scaling limit, the Painlev\'e IV transcendent (see Section~\ref{Section: Detour: Painleve IV}) with $\Theta_0=\Theta_{\infty}=\frac{\gamma}{2}$ and $\gamma = \frac{d - \ell - 1}{d} \in (0,1)$ plays a major role.

\begin{Theorem}[double scaling limit] \label{theorem2}The polynomials $p_n (\lambda)$ with $n = k d + \ell$, $\ell = 0,\dots, d-2$, $\gamma = \frac{d - \ell - 1}{d} \in (0,1)$, have the following asymptotic behaviour when $n$, $N \to \infty$ in such a way that $N T = n - \ell$ and
\begin{gather*}
\lim_{k \to \infty,\, t \to t_c}\sqrt{k}\left(\dfrac{t^2}{t_c^2}-1\right)=\mathcal{S} ,
\end{gather*}
with $\mathcal{S}$ in compact subsets of the real line so that the solution $Y (\mathcal{S})$ of the Painlev\'e IV equation~\eqref{P4} does not have poles. Below, the function $Z = Z(\mathcal{S})$, $U = U(\mathcal{S})$ and the Hamiltonian $H = H (\mathcal{S})$ are related to the Painlev\'e IV equation~\eqref{P4} by the relations~\eqref{isomono} and~\eqref{def_H}, respectively.
\begin{itemize}\itemsep=0pt
\item[$(1)$] For $\lb$ in compact subsets of the exterior of $\hat{\mathcal{C}}$ one has
\begin{gather*}
p_n (\lambda) = \lb^{d-1} \big( \lb^d - t_c \big)^{k-\gamma} \left(1 - \dfrac{H (\mathcal{S}) t_c}{\sqrt{k} \lambda^d} + \mathcal{O} \left( \frac{1}{k} \right) \right),
\end{gather*}
with $H (\mathcal{S})$ the Hamiltonian \eqref{def_H} of the Painlev\'e IV equation~\eqref{P4}.
\item[$(2)$] For $\lb$ in the region near $\hat{\mathcal{C}}$ and away from the point $\lb = 0$ one has
\begin{gather}
p_n (\lambda) = \lb^{d-1} \big( \lb^d - t_c \big)^{k-\gamma} \nonumber\\
\hphantom{p_n (\lambda) =}{} \times \left(1 - \dfrac{H (\mathcal{S}) t_c}{\sqrt{k} \lambda^d} + \dfrac{Z (\mathcal{S})}{U (\mathcal{S})} \dfrac{t_c {\rm e}^{-k \hat{\varphi}(\lambda) }}{\lambda^d k^{\frac{1+\gamma}{2}} }
\left(\dfrac{\lambda^d - t_c}{\lambda^d}\right)^{{\gamma}} + \mathcal{O} \left(\frac{1}{k}\right) \right),\label{pn_post}
\end{gather}
with $\hat{\varphi} (\lambda)$ defined in \eqref{hat_phi_r}.
\item[$(3)$] For $\lb$ in compact subsets of the interior region of $\hat{\mathcal{C}}$ and away from the point $\lb = 0$ one has
\begin{gather*}
p_n (\lambda) = \lb^{\ell} \big(\lambda^d - t_c\big)^k \dfrac{{\rm e}^{-k\hat{\varphi} (\lambda)}}{k^{\frac{1}{2} + \gamma}} \left( \dfrac{Z (\mathcal{S})}{U (\mathcal{S})}\dfrac{t_c}{\lambda^d} + \mathcal{O} \left( \frac{1}{k} \right)\right).
\end{gather*}
\item[$(4)$] In the neighbourhood of the point $\lb=0$ one has
\begin{gather*}
p_n (\lambda) = \lb^{\ell} \big( \lb^d - t_c \big)^{k-\gamma} \dfrac{{\rm e}^{-\frac{k}{2} \hat{\varphi}(\lb)}\lb^{d\gamma}}{(-2k\hat{\varphi}(\lambda))^{\frac{\gamma}{4}}}\\
\hphantom{p_n (\lambda) =}{}\times
\begin{cases}
\Psi_{11} \big( \sqrt{-2k \hat{\varphi} (\lambda)} ; \mathcal{S} \big)+ \mathcal{O} \left( \frac{1}{k} \right),& \sqrt{- \hat{\varphi}(\lb)}\in\Omega_{\infty}, \\
-\Psi_{12} \big( \sqrt{-2k \hat{\varphi} (\lambda)} ; \mathcal{S} \big)+ \mathcal{O} \left( \frac{1}{k} \right),& \sqrt{- \hat{\varphi}(\lb)}\in\Omega_{0}, \\
\Psi_{1 1} \big( \sqrt{-2k \hat{\varphi} (\lambda)} ; \mathcal{S} \big)-\Psi_{1 2} \big( \sqrt{-2k \hat{\varphi} (\lambda)} ; \mathcal{S} \big)+ \mathcal{O} \left( \frac{1}{k} \right), & \sqrt{- \hat{\varphi}(\lb)}\in\Omega_{2},
\end{cases}
\end{gather*}
where $\Psi$ is the solution of the Painlev\'e IV Riemann--Hilbert Problem~{\rm \ref{PIV_0}} with Stokes multipliers specified in \eqref{choice1} and \eqref{choice2} and the regions $\Omega_0$, $\Omega_{\infty}$ and $\Omega_2$ are specified in Fig.~{\rm \ref{fig: RHP Psi}}. Here $\Psi_{11}$ and $\Psi_{12}$ are the entries of the matrix $\Psi$.
\end{itemize}
\end{Theorem}
We observe that, in compact subsets of the exterior of $\hat{\mathcal{C}}$, there are no zeros of the polyno\-mials~$p_n (\lambda)$. The only possible zeros are located in $\lb = 0$ and in the region where the second term in parenthesis in the expression~\eqref{pn_post} is of order one. Since $\Re \hat{\varphi} (\lambda)$ is negative inside $\hat{\mathcal{C}}$ and positive outside $\hat{\mathcal{C}}$, it follows that the possible zeros of $p_n (\lambda)$ lie inside $\hat{\mathcal{C}}$ and are determined by the condition
\begin{gather*}
\Re \hat{\varphi} (\lambda) = -\dfrac{1+\gamma}{2} \dfrac{\log k}{k} + \dfrac{1}{k} \log \left( \left| \dfrac{\lambda^d-t_k}{\lambda^d} \right|^{\gamma} \left| \dfrac{t_c Z (\mathcal{S})}{\lambda^d U (\mathcal{S})} \right| \right) + \dfrac{1}{k^{\frac{3}{2}}}\Re \left( \dfrac{t_c H (\mathcal{S})}{\lambda^d} \right) + \mathcal O\big(k^{-2}\big).
\end{gather*}

We conclude with the following proposition.
\begin{Proposition}\label{propo_zeros_pre_0}The support of the counting measure of the zeros of the polynomials $p_n (\lambda)$ in the limit $n\to\infty$ outside an arbitrary small disk surrounding the point $\lambda=0$ tends uniformly to the curve $\hat{\mathcal{C}}$ defined in~\eqref{Gamma}. The zeros are within a distance~$\mathcal{O}( 1 / k)$ from the curve defined by~\eqref{Gamma}.
\end{Proposition}

This manuscript is organised as follows. In Section~\ref{Section: Detour: Painleve IV} we present the Lax pair and Riemann--Hilbert problem for the Painlev\'e IV equation and we show how to associate for a~particular value of the Stokes matrices the solution of the Painlev\'e IV equation in terms of a~Fredholm determinant. We compute such solution numerically for several values of the parameter $\gamma$ that parametrizes the asymptotic behaviour of the orthogonal polynomials. In Section~\ref{Section: Orthogonal Polynomials and the Riemann--Hilbert Problem} we show how to associate to the orthogonal polynomials on the planar domain a Riemann--Hilbert problem and in Section~\ref{section4} we perform the asymptotic analysis of this Riemann--Hilbert problem thus deriving the asymptotic behaviour of the orthogonal polynomials $p_n(\lb)$ as $n\to \infty$ in the complex plane.

\section{Painlev\'{e} IV and Fredholm determinant}\label{Section: Detour: Painleve IV}

The purpose of this section is to introduce the Painlev\'{e} IV equation and study various aspects related to it that are connected with the work that we will be doing in the following sections.

We will begin by introducing and deriving the general Painlev\'{e} IV equation following the ideas outlined by Miwa, Jimbo and Ueno \cite{JMU1, JMU2}. In the sequence of this, we will also discuss the associated Stokes' phenomenon and the monodromy problem (see Wasow~\cite{Wasow}).

The second part of this section will be devoted to applying the general setting of Painlev\'{e}~IV in a particular case, in order to obtain a special solution for it. This will correspond to a~special Riemann--Hilbert problem whose properties will be outlined. The reason and motivation for introducing this solution is the fact that this Riemann--Hilbert problem will be seen to be useful in the following section.

The third and final part of this section will be devoted to the study of the Fredholm determinant, as it has been established by Its~\cite{ItsIzerginKorepinSlavnov}, Harnad~\cite{ItsHarnad}, Bertola and Cafasso~\cite{BertolaCafasso}, among others. This will be used in order to study the $\tau$-function of the special solution of Painlev\'{e}~IV that we introduced before and analyze its validity with respect to the value of the parameter~$s$.

\subsection{The general Painlev\'{e} IV}\label{Subsection: The general Painleve IV}

We recall the Lax pair formulation of the Painlev\'{e} IV equation as in the original work of Miwa, Jimbo and Ueno \cite{JMU1, JMU2}, with notation adapted to our purposes.

Consider the matrix--valued function $\Psi$, associated with Painlev\'{e} IV, as the fundamental joint solution of the system of linear differential equations, the Lax pair, given by
\begin{gather}\label{eq: Lax pair PIV}
 \frac \pa{\pa \lambda} \Psi = A ( \lambda ; s ) \Psi, \qquad \frac \pa{ \pa s} \Psi = B ( \lambda ; s) \Psi,
\end{gather}
where
\begin{gather*}
 A ( \lambda ; s ) = - \frac{1}{2} \left( \lambda + s \right) \sigma_{3} + \begin{pmatrix} 0 & \dfrac{Z}{U} \vspace{1mm}\\ -U & 0 \end{pmatrix} + \frac{1}{\lambda} \begin{pmatrix} \Theta_{\infty} - Z & \dfrac{( \Theta_{\infty} - Z)^{2} - \Theta^{2}_{0}}{Y U} \vspace{1mm}\\ -U Y & Z - \Theta_{\infty} \end{pmatrix} ,
\\
 B ( \lambda ; s ) = - \frac{1}{2} \lambda \sigma_{3} + \begin{pmatrix} 0 & \dfrac{Z}{U} \vspace{1mm}\\ -U & 0 \end{pmatrix} ,
\end{gather*}
where $\sigma_{3}$ is the third Pauli matrix, $Y$, $Z$ and $U$ are functions of $s$ whose dependence we now derive, while $\Theta_0$ and $\Theta_\infty$ are arbitrarily chosen constants. The compatibility of \eqref{eq: Lax pair PIV} yields the zero-curvature equation
\begin{gather*}
 [ \partial_{\lambda} - A, \partial_{s} - B ] \equiv 0,
\end{gather*}
which, in turn, yields the following system of nonlinear ODEs
\begin{gather}
 U' = U( Y - s ),\nonumber \\
 Z' = Z Y + \dfrac{- ( Z - \Theta_{\infty} )^{2} + \Theta^{2}_{0}}{Y},\nonumber\\
 Y' = - Y^{2} + s Y - 2 Z + 2 \Theta_{\infty}.\label{isomono}
\end{gather}
In the equations above the prime means differentiation with respect to $s$. By eliminating $U$, $Z$ from the above equations we obtain the classical form of the fourth Painlev\'e equation
\begin{gather}\label{P4}
Y'' = \frac{1}{2} \frac{(Y')^{2}}{Y} + \frac{3}{2} Y^{3} - 2 s Y^{2} + \left( 1 + \frac{s^{2}}{2} - 2 \Theta_{\infty} \right) Y - \frac{2 \Theta^{2}_{0}}{Y}.
\end{gather}

\subsubsection{Stokes' phenomenon and the Monodromy problem}\label{Subsection: Stokes phenomenon and the Monodromy problem}
The review material of this section is mostly based on \cite{ItsbookPainleve} and \cite{JMU1, JMU2}. The solution of the first equation in \eqref{eq: Lax pair PIV} admits a formal sectorial expansion at $\lambda=\infty$ of the form
\begin{gather}
\Psi _{\rm formal} (\lambda) = \left( \mathbf{1} + \frac{\Psi_1(s)}{\lambda} +\frac {\Psi_2(s)}{\lambda^2} + \mathcal{O} \left( \frac{1}{\lambda^{3}} \right) \right) \lambda^{\Theta_{\infty} \sigma_{3}} {\rm e}^{- \theta (\lambda) \sigma_{3}} ,\nonumber\\
\label{Psi1}
\Psi_{1} (s) = \begin{bmatrix} H (s) & \dfrac{Z (s)}{U (s)} \vspace{1mm}\\ U (s) & -H (s) \end{bmatrix} ,
\\
\label{Psi2}
\Psi_{2} (s) =
\left[ \begin{matrix}
\dfrac{1}{2} \big({H (s)^{2} + Z (s) - s H (s)}\big)
 &{\dfrac {Z (s) ( Z (s) - \gamma - Y (s) s - H (s) Y (s) ) }{U (s) Y (s) }}
\vspace{1mm}\\ U (s) H (s) + U (s) Y (s) - s U (s) & \dfrac{1}{2} \big( {H (s)^{2} + Z (s) + s H (s)}\big)
\end{matrix} \right] ,
\\
\label{eq: theta def}
\theta (\lambda) = \frac{\lambda^{2}}{4} + \frac{s}{2} \lambda,
\end{gather}
where
\begin{gather}\label{def_H}
H = \left( s + \frac{2 \Theta_{\infty}}{Y} - Y \right) Z - \frac{\Theta_{\infty}^{2} - \Theta_0^{2}}{Y} -\frac{Z^{2}}{Y}.
\end{gather}
In each of the sectors $q = ( 0, {\rm I}, {\rm II}, {\rm III}, {\rm IV}, {\rm V})$ of Fig.~\ref{fig: Stokes phenom} there exists a unique solution $\Psi_{(q)} (\lambda)$ which is Poincar\'e\ asymptotic to $\Psi_{\textrm {formal}} (\lambda)$, even though the latter is only a non-convergent (in general) formal series. The matrices in Fig.~\ref{fig: Stokes phenom} on the rays $\Gamma_q$ are the Stokes' matrices~$\mathbf S_q$ (and the formal monodromy on the negative axis) and relate the different solutions as follows: $\Psi_{(q+1)} = \Psi_{q} {\mathbf S}_q$ ($q$ modulo~$6$). The rays $\Gamma_{1,5}$ can be chosen arbitrarily within the respective quadrants, or more generally any smooth infinite arc within the quadrants that admits an asymptotic direction.

\begin{figure}[t]\centering
\includegraphics[width=0.6\textwidth]{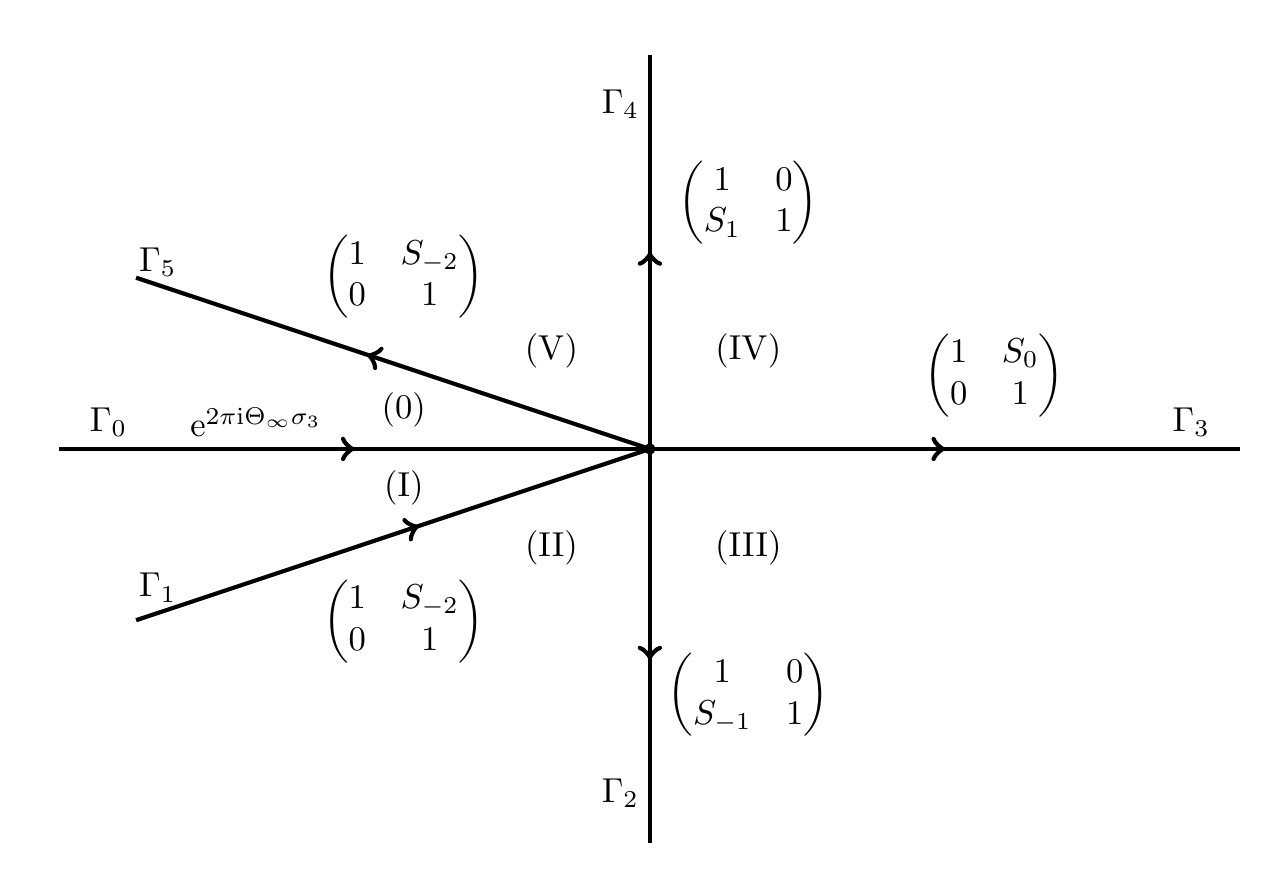}
\caption{Stokes' phenomenon and the corresponding jumps in the $\Psi (\lambda)$ Riemann--Hilbert problem.}\label{fig: Stokes phenom}
\end{figure}

Each of these solutions of Painlev\'{e} IV for each sector $q$, $\Psi_{(q)} (\lambda)$, has an asymptotic expansion near $\lambda = 0$ with the following behaviour
\begin{gather*}
\Psi_{(q)} (\lambda) = G ( \mathbf{1} + \mathcal{O} (\lambda) ) \lambda^{\Theta_{0} \sigma_{3}} C_{(q)} ,
\end{gather*}
where the matrices $C_{(q)}$ defined by this equation are called ``connection matrices'' and~$G$ is a~diagonalizing matrix for the $\frac{1}{\lambda}$ coefficient of $A ( \lambda ; s )$. From this equation and using an argument of analytic continuation, we obtain the following {\it monodromy relation}
\begin{gather}
C^{-1}_{(0)} {\rm e}^{2 \pi {\rm i} \Theta_{0} \sigma_{3}} C_{(0)} = \begin{pmatrix} 1 & - S_{-2} \\ 0 & 1 \end{pmatrix} \begin{pmatrix} 1 & 0 \\ - S_{1} & 1 \end{pmatrix} \nonumber\\
\hphantom{C^{-1}_{(0)} {\rm e}^{2 \pi {\rm i} \Theta_{0} \sigma_{3}} C_{(0)} =}{}
\times \begin{pmatrix} 1 & - S_{0} \\ 0 & 1 \end{pmatrix} \begin{pmatrix} 1 & 0 \\ - S_{-1} & 1 \end{pmatrix}
\begin{pmatrix} 1 & -S_{-2} \\ 0 & 1 \end{pmatrix} {\rm e}^{ 2 \pi {\rm i} \Theta_{\infty} \sigma_{3}}.\label{monrel}
\end{gather}
The remaining connection matrices $C_{(q)}$ are obtained from $C_{(0)}$ by multiplying it by an appropriate sequence of the Stokes' matrices; for example $C_{\left( \mathrm{V} \right)} = C_{(0)} \left(\begin{smallmatrix} 1 & S_{-2}\\ 0& 1\end{smallmatrix}\right)$.

The general solution to the isomonodromic equations \eqref{isomono} for given $\Theta_{0}$, $\Theta_{\infty}$ is parametrized by the choices of parameters $S_{-1},\dots$ and connection matrix $C_{(0)}$ that satisfy~\eqref{monrel}.

\subsection[Special Riemann--Hilbert problem $\Psi (\lambda)$ and Painlev\'{e} IV]{Special Riemann--Hilbert problem $\boldsymbol{\Psi (\lambda)}$ and Painlev\'{e} IV}\label{Subsection: Special Riemann--Hilbert Problem Psi ( lambda ) and Painleve IV}
Several particular solutions of the Painlev\'e IV have been considered in the literature, see, e.g., \cite{Clarkson,DK,Kapaev}. In our work we only need special values for the Stokes' parameters and $\Theta_{0,\infty}$:
\begin{gather}\label{choice1}
S_{1} = 1, \qquad S_{0} = 0, \qquad S_{-1} = -1, \qquad S_{-2} = -1 ,\\
\label{choice2} C_{( \textrm{III})} = C_{( \textrm{IV} )} = \mathbf{1},\qquad \Theta_{0} = \frac{\gamma}{2} , \qquad \Theta_{\infty} = \frac{\gamma}{2}.
\end{gather}
The resulting Riemann--Hilbert problem is the one depicted in Fig.~\ref{fig: RHP Psi}, where we have more appropriately renamed $\Gamma_{1}$ and $\Gamma_{\infty}$. It is formalized below:
\begin{figure}\centering
\includegraphics{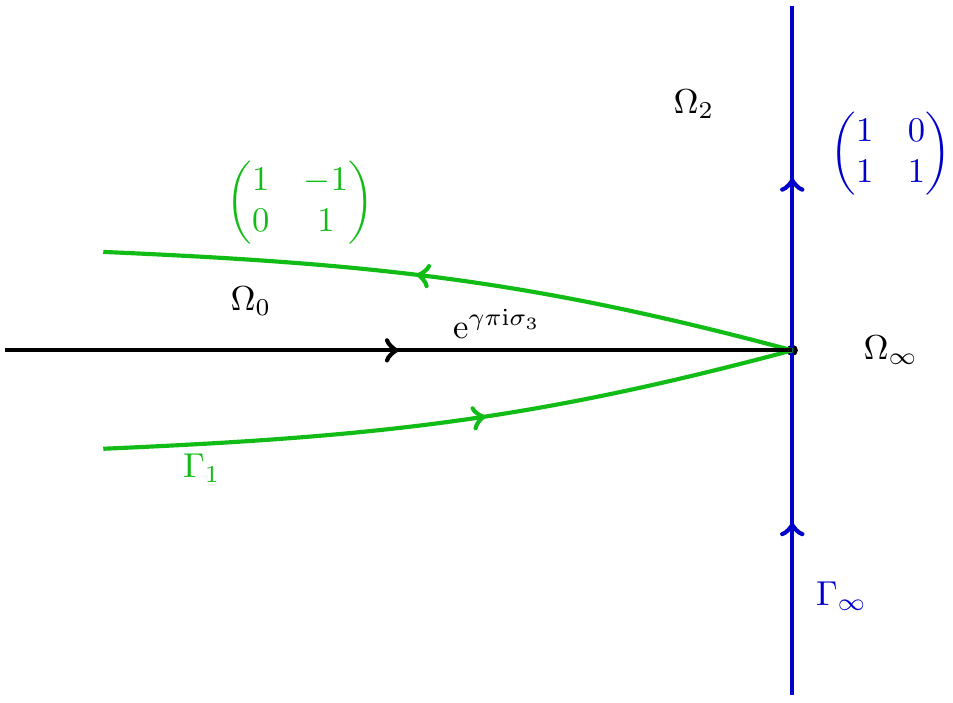}
\caption{Jumps in the $\Psi (\lambda)$ Riemann--Hilbert problem.}\label{fig: RHP Psi}
\end{figure}

\begin{problem}\label{PIV_0} The matrix $\Psi (\lambda ;s)$ is analytic in $\mathbb{C} \setminus ( \Gamma_{1} \cup \Gamma_{\infty} \cup \mathbb{R}_{-} )$ and admits non-tangential boundary values. Moreover
\begin{enumerate}\itemsep=0pt
\item[$1.$] The boundary values are related by $($see Fig.~{\rm \ref{fig: RHP Psi})}:
\begin{gather*}
\Psi_{+} (\lambda;s) = \Psi_{-} (\lambda ;s) v_{\Psi}, \qquad \lambda \in \Sigma_{\Psi} ,\nonumber\\
v_{\Psi} =
\begin{cases}
 \begin{pmatrix} 1 & -1 \\ 0 & 1 \end{pmatrix} , & \lambda \in \Gamma_{1}, \vspace{1mm}\\
 \begin{pmatrix} 1 & 0 \\ 1 & 1 \end{pmatrix} , & \lambda \in \Gamma_{\infty}, \\
 {\rm e}^{\gamma \pi {\rm i} \sigma_3}, \quad & \lambda \in \mathbb{R}_{-},
\end{cases}
\end{gather*}
where $\Psi_{\pm}$ are the boundary values on the left and right of the oriented contour $\Sigma_{\Psi}$.
\item[$2.$] Near $\lambda=\infty$ the matrix has the following sectorial behaviour
\begin{gather*}
\Psi (\lambda;s) = \left( \mathbf{1} + \mathcal{O} \left( \frac{1}{\lambda} \right) \right) \lambda^{\frac{\gamma}{2} \sigma_{3}} {\rm e}^{- \theta (\lambda) \sigma_{3}} , \qquad \lambda \rightarrow \infty.
\end{gather*}

\item[$3.$] Near $\l=0$ in the region $\Omega_\infty$, $\Psi (\lambda;s) $ has the behaviour
\begin{gather*}
\Psi (\lambda;s) = \mathcal{O} (1) \lambda^{\frac{\gamma}{2} \sigma_{3}}.
\end{gather*}
\end{enumerate}
\end{problem}
 We remind the reader that subcripts $_\pm$ denote the boundary values from the left~($+$) or right~($-$) of an oriented countour. We will henceforth omit explicit notation for the~$s$ dependence of~$\Psi$.

We now proceed to do a further modification of this $\Psi (\lambda)$-Riemann--Hilbert problem, that will be used in the next section for the study of the Fredholm determinant:
\begin{itemize}\itemsep=0pt
\item move the jump contour $\Gamma_{\infty} = {\rm i} \mathbb{R}$ in the following way ${\rm i} \mathbb{R} \rightarrow {\rm i} \mathbb{R} + \varepsilon$, with $\epsilon>0$;
\item join together (collapse) the two parts of the contour $\Gamma_{1}$ and $\mathbb{R}_{-}$.
\end{itemize}
The use of the verb ``move'' means that we define a new matrix-valued function in the region between the old and new position of the contours that analytically continues $\Psi$ across the original contours. This is possible because the jump matrices are all constants in $\lambda$. We still use the symbol $\Psi$ for the resulting function. Finally, we introduce $M$ as follows
\begin{gather*}
M (\lambda) : = \Psi (\lambda) {\rm e}^{ \theta (\lambda) \sigma_{3}} \lambda^{- \frac{\gamma}{2} \sigma_{3}}.
\end{gather*}
The powers of $\lambda$ are intended as the principal determination with $\arg (\lambda) \in ( -\pi, \pi)$. The mat\-rix~$M$ solves a Riemann--Hilbert problem with jumps indicated in Fig.~\ref{fig: RHP H} and formalized below.
\begin{figure}[t]\centering
\includegraphics[width=0.5\textwidth]{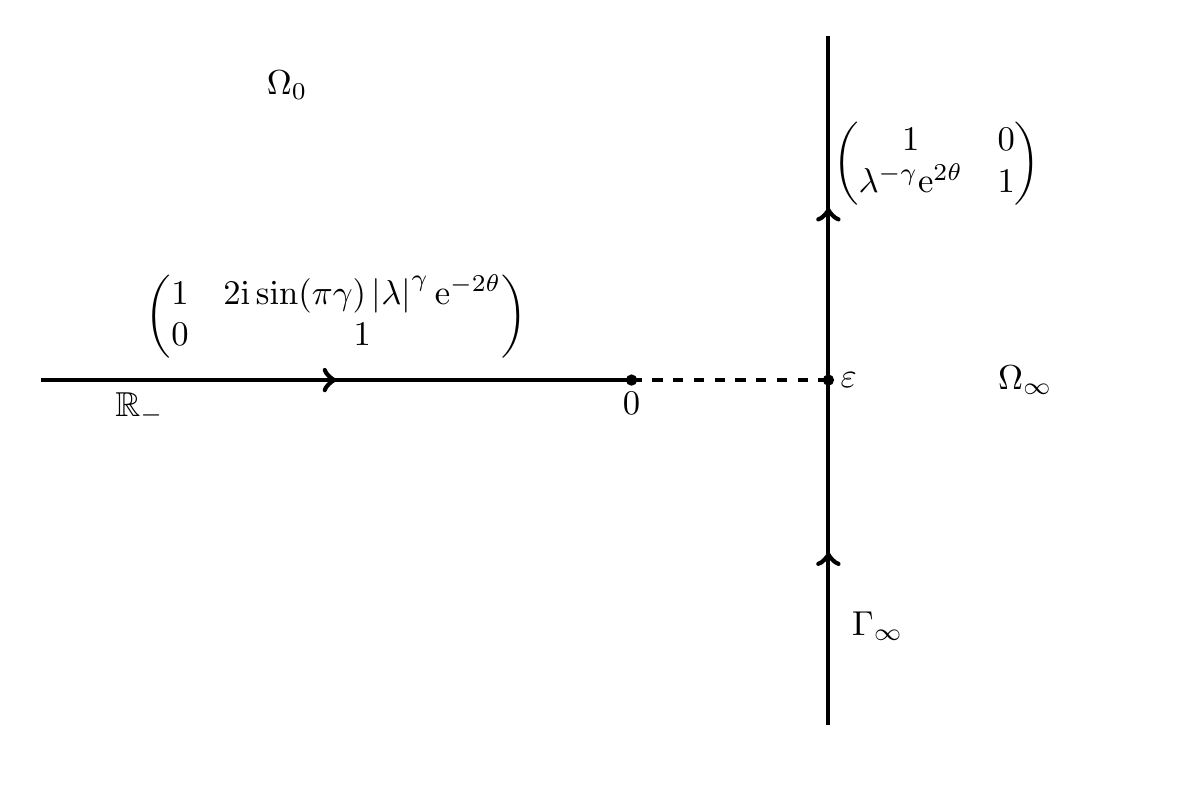}
\caption{Jumps in the $M (\lambda)$ Riemann--Hilbert problem.}\label{fig: RHP H}
\end{figure}

\begin{problem}\label{HRHP}The matrix $M (\lambda)$ is analytic in $\mathbb{C} \setminus ( \Gamma_{\infty} \cup \mathbb{R}_{-} )$ and admits non-tangential boundary values. Moreover
\begin{enumerate}\itemsep=0pt
 \item[$1.$] The boundary values $M_{\pm}$ are related by $($see Fig.~{\rm \ref{fig: RHP H})}:
\begin{gather*}
M_{+} (\lambda) = M_{-} (\lambda) v_{M}, \qquad \lambda \in \Sigma_{M} ,\nonumber\\
v_{M} =
\begin{cases}
\begin{pmatrix} 1 & 2 {\rm i} \sin{ (\pi \gamma)} |\lambda|^{\gamma} {\rm e}^{- 2 \theta (\lambda) } \\ 0 & 1 \end{pmatrix} , & \lambda \in \mathbb{R}_{-},\vspace{1mm} \\
\begin{pmatrix} 1 & 0 \\ \lambda^{- \gamma} {\rm e}^{2 \theta (\lambda)} & 1 \end{pmatrix} , & \lambda \in \Gamma_{\infty}.
\end{cases}
\end{gather*}

\item[$2.$] Near $\lambda=\infty$ the matrix has the following sectorial behaviour:
\begin{gather*}
M (\lambda) = \left( \mathbf{1} + \mathcal{O} \left( \frac{1}{\lambda} \right) \right), \qquad \lambda \rightarrow \infty.
\end{gather*}

\item[$3.$] The matrix $M (\lambda)$ is bounded in a full neighbourhood of $\lambda=0$.
\end{enumerate}
\end{problem}

\subsection[Fredholm determinant and the $\tau$-function of Painlev\'{e} IV]{Fredholm determinant and the $\boldsymbol{\tau}$-function of Painlev\'{e} IV}\label{Subsection: Fredholm determinant and the tau-function of Painleve IV}

The Riemann--Hilbert Problem~\ref{HRHP} is associated to an integral operator on $L^{2} ( \mathbb{R}_{-} \cup \Gamma_\infty)$, where $\Gamma_\infty={\rm i}\R+\epsilon$, and will be used in the next section in order to compute the Fredholm determinant. We first introduce some general definitions. Let $( X, \mathrm{d} \nu )$ be a measurable space and $K (x,y) \colon X \times X \to \C$ a function. We define the operator $\mathcal{\hat{K}}\colon L^{2} ( X , \mathrm{d} \nu ) \rightarrow L^{2} ( X , \mathrm{d} \nu )$ as
\begin{gather*}
\big( \mathcal{\hat{K}} f \big)(x)= \int_{X} K (x,y) f (y) \mathrm{d} \nu (y).
\end{gather*}
If $\iint_{X^2} | K (x,y) |^2 \d x \d y < +\infty$, the operator $\hat{\mathcal K}$ is a Hilbert--Schmidt operator and its spectrum is purely discrete, the eigenvalues can accumulate only at~$0$ and all the multiplicities of the nonzero eigenvalues are finite (see, e.g.,~\cite{Simon}).
\begin{Definition}\label{def: Fred Det}The $\tau$-function is defined as the following Fredholm determinant
\begin{gather*}
\tau (\rho) = \det{\big( \operatorname{Id} - \rho \mathcal{\hat{K}}\big)},
\end{gather*}
where $\rho$ is a nonzero complex parameter.
\end{Definition}

From the properties of the Fredholm determinant we know that $\tau (\rho) = 0$ if and only if $\frac{1}{\rho} $ is an eigenvalue of $\mathcal{\hat{K}}$. With this being said, since in our case we are interested in a different parametric dependence of $\mathcal {\hat K}$ and not its spectral properties, we shall set $\rho=1$ (see~\eqref{hatK}).

\subsubsection{Riemann--Hilbert problems and Fredholm determinants}\label{Subsection: Riemann--Hilbert problems and Fredholm determinants}
Following \cite{ItsIzerginKorepinSlavnov} let us consider a set of contours $\Sigma \subset \mathbb{C}$ and $f, g\colon \Sigma \to \mathbb{C}^p$ matrix--valued (smooth) functions that satisfy the condition $f^{\rm T} (\lambda) g (\lambda) = 0$. Then we define the following scalar kernel
\begin{gather*}
K ( \lambda, \mu) : = \frac{f^{\rm T} (\lambda) g (\mu)}{\lambda - \mu}.
\end{gather*}
We denote by $\hat{\mathcal{K}}\colon L^{2} ( \Sigma, \mathbb{C}^{p} ) \rightarrow L^{2} ( \Sigma, \mathbb{C}^{p} )$ the integral operator with kernel $K ( \lambda , \mu)$. Let $\mathcal{R}$ be the resolvent operator
$\mathcal R : = - \hat{\mathcal{K}} \circ ( \operatorname{Id} - \hat{\mathcal{K}} )^{-1}.$
\begin{Theorem}[see \cite{ItsIzerginKorepinSlavnov}]\label{thm: rhp fred det}The kernel, $R( \lambda , \mu )$, of the resolvent operator $\mathcal R$ is given by
\begin{gather*}
R ( \lambda, \mu ) := \frac{f^{\rm T} (\lambda) A^{\rm T} (\lambda) A^{- {\rm T}} (\mu) g (\mu)}{\lambda - \mu} ,
\end{gather*}
where the $p \times p$ matrix $A (\lambda)$ solves the Riemann--Hilbert problem
\begin{gather*}
A_{+} (\lambda) = A_{-} (\lambda) \big( \mathbf{1} - 2 \pi {\rm i} f (\lambda) g^{\rm T} (\lambda) \big) , \qquad \lambda \in \Sigma ,\\
A (\lambda) = \mathbf{1} + \mathcal{O} \big( \lambda^{-1} \big) , \qquad \lambda \rightarrow \infty.
\end{gather*}
Moreover, the above Riemann--Hilbert problem admits a unique solution if and only if the Fredholm determinant $\det \big( {\mathrm {Id}} - \hat{\mathcal K} \big)$ is non-zero.
\end{Theorem}
The case of interest to us, below, is with $p=2$.

\subsubsection[Fredholm determinant and the Riemann--Hilbert problem $M$]{Fredholm determinant and the Riemann--Hilbert problem $\boldsymbol{M}$}\label{Subsection: Fredholm determinant and the Riemann--Hilbert problem mathcalH}
The Riemann--Hilbert Problem \ref{HRHP} is a special example of the previous setup, as we now see. We first define the following vectors
\begin{gather}\label{eq: f spec P}
f (z) = \begin{bmatrix} 2 {\rm i} \sin (\pi \gamma) |z|^\frac{{\gamma}}{2} {\rm e}^{- \theta (z)} \chi_{\mathbb{R}_{-}} (z) \\ z^{- \frac{{\gamma}}{2}} {\rm e}^{\theta (z)} \chi_\infty
 (z) \end{bmatrix} ,\\
\label{eq: g spec P}
g (z) = - \frac{1}{2 \pi {\rm i}} \begin{bmatrix} z^{- \frac{{\gamma}}{2}} {\rm e}^{\theta (z)} \chi_{\infty} (z) \\ | z |^\frac{{\gamma}}{2} {\rm e}^{- \theta (z)} \chi_{\mathbb{R}_{-}} (z) \end{bmatrix} ,
\end{gather}
where $\chi_{\infty}$ and $\chi_{\mathbb{R}_{-}}$ are the characteristic functions of $\Gamma_\infty$ and $\mathbb{R}_{-}$, respectively. This leads to the following jump matrices for Theorem~\ref{thm: rhp fred det}
\begin{gather*}
\mathbf{1} - 2 \pi {\rm i} f (z) g^{\rm T} (z) =
\begin{cases}
\begin{pmatrix} 1 & 2 {\rm i} \sin{ (\pi \gamma)} |z|^{\gamma} {\rm e}^{- 2 \theta (z)} \\ 0 & 1 \end{pmatrix} , & z \in \mathbb{R}_{-}, \vspace{1mm}\\
\begin{pmatrix} 1 & 0 \\ z^{- \gamma} {\rm e}^{2 \theta (z)} & 1 \end{pmatrix} , & z \in\Gamma_\infty,
\end{cases}
\end{gather*}
which match those in the Riemann--Hilbert Problem~\ref{HRHP} ($\mathbf 1$ denotes the $2\times 2$ identity matrix).

Let $\Sigma = \Gamma_\infty\cup \mathbb{R}_{-}$, as a consequence of this, we can now write the Riemann--Hilbert Problem~(\ref{HRHP}) for the matrix $M$ in the following way
\begin{gather*}
M_{+} (z) = M_{-} (z) \big( \mathbf{1} - 2 \pi {\rm i} f (z) g^{\rm T} (z) \big) , \qquad z \in \Sigma ,\\
M = \mathbf{1} + \mathcal{O} \left( \frac{1}{z} \right) , \qquad z \rightarrow \infty.
\end{gather*}
 It should also be noted that, due to our construction, we have $f^{\rm T} (z) g (z) \equiv 0$.

Let us consider the operator $\hat{\mathcal{K}} \colon L^{2} (\Sigma) \rightarrow L^{2} (\Sigma)$ associated to the above Riemann--Hilbert problem.
Due to the decomposition of $ L^{2} (\Sigma) = L^{2} ( \mathbb{R_{-}}\cup \Gamma_\infty) =L^{2} (\mathbb{R_{-}}) \oplus L^{2} (\Gamma_\infty) $ and
 $\varphi=\varphi_0\oplus \varphi_1$ with $\varphi_0\in L^{2} (\mathbb{R}_{-})$ and $\varphi_1\in L^{2} (\Gamma_\infty)$, we can write the operator $\hat{{\cal K}}$ in the form
\begin{gather}
\big(\hat{\mathcal K }\varphi\big) (z) = \int_{\Sigma} \frac{f^{\rm T} (z) g (w)}{z - w} \varphi (w) \mathrm{d} w \nonumber\\
\hphantom{\big(\hat{\mathcal K }\varphi\big) (z)}{} = - \frac{2 {\rm i} \sin{ (\pi \gamma)}}{ 2 \pi {\rm i} } \left( \int_{\Gamma_\infty} \frac{w^{- \frac{\gamma}{2}} {\rm e}^{\theta (w)} \varphi_{1} (w)}{z - w} \mathrm{d} w \right) |z|^{\frac{\gamma}{2}} {\rm e}^{- \theta (z)} \chi_{\mathbb{R}_{-}} (z) \nonumber\\
\hphantom{\big(\hat{\mathcal K }\varphi\big) (z) = }{}
- \frac{1}{2 \pi {\rm i}} \left( \int_{\mathbb{R}_{-}} \frac{|w|^{\frac{\gamma}{2}} {\rm e}^{- \theta (w)} \varphi_{0} (w)}{z - w} \mathrm{d} w \right) z^{- \frac{\gamma}{2}} {\rm e}^{\theta (z)} \chi_{\Gamma_\infty} (z).\label{hatK}
\end{gather}

Due to this fact, the operator $\hat{\mathcal{K}}$ can now be written in block form $\hat{\mathcal{K}}_{i j}$ with
\begin{gather*}
\hat{\mathcal{K}}_{0 0}\colon \ L^{2} (\mathbb{R}_{-}) \rightarrow L^{2} (\mathbb{R}_{-}) ,\\
\hat{\mathcal{K}}_{0 1}\colon \ L^{2} (\Gamma_\infty) \rightarrow L^{2} (\mathbb{R}_{-}) ,\\
\hat{\mathcal{K}}_{1 0}\colon \ L^{2} (\mathbb{R}_{-}) \rightarrow L^{2} (\Gamma_\infty) ,\\
\hat{\mathcal{K}}_{1 1}\colon \ L^{2} (\Gamma_\infty) \rightarrow L^{2} (\Gamma_\infty).
\end{gather*}
In our case, due to the form of the vectors $f$ and $g$ in \eqref{eq: f spec P}, \eqref{eq: g spec P}, the diagonal blocks are null operators: $\hat{\mathcal{K}}_{0 0} = \hat{\mathcal{K}}_{1 1} = 0$, while
\begin{gather*}
 \hat{{\cal K}}_{0 1} \colon \ L^{2} (\Gamma_\infty) \rightarrow L^{2} (\mathbb{R}_{-}),\qquad \big( \hat{{\cal K}}_{0 1} \varphi_1 \big) (z) = \int_{\Gamma_\infty}K_{0,1}(z,w)\varphi_1(w)\mathrm{d}w,\nonumber\\
K_{0 1} (z,w) = - \frac{2 {\rm i} \sin{ (\pi \gamma)}}{(2\pi {\rm i})} \frac{w^{- \frac{\gamma}{2}} |z|^{\frac{\gamma}{2}} {\rm e}^{\theta(w)- \theta (z)}}{z - w},\\ 
 \hat{{\cal K}}_{1 0} \colon \ L^{2} (\mathbb{R}_{-}) \rightarrow L^{2} (\Gamma_\infty), \qquad \big( \hat{{\cal K}}_{1 0} \varphi_0 \big) (z) = \int_{R_-}K_{1,0}(z,w)\varphi_0(w)\mathrm{d}w,\nonumber\\
 K_{1 0} (z,w) = - \dfrac{K_{0 1} (w,z)}{ 2 {\rm i} \sin( \pi \gamma)} = - \frac{1}{(2\pi {\rm i})} \frac{|w|^{\frac{\gamma}{2}} z^{- \frac{\gamma}{2}} {\rm e}^{\theta(z)- \theta (w)}}{z - w}.
\end{gather*}

Following the Definition \ref{def: Fred Det} of the tau function, we can now compute it for the case when $\hat{\mathcal K}$ is given by this operator
\begin{gather*}
\tau = \det{\big( \operatorname{Id} - \mathcal{\hat{K}} \big)} = \det{\left(\begin{bmatrix} \mathbf{1}_{0} & 0 \\ 0 & \mathbf{1}_{1} \end{bmatrix} - \begin{bmatrix} 0 &\hat{{\cal K}}_{0 1} \\\hat{{\cal K}}_{1 0} & 0 \end{bmatrix} \right)} = \det{\left( \begin{bmatrix} \mathbf{1}_{0} & \hat{{\cal K}}_{0 1} \\ 0 & \mathbf{1}_{1} \end{bmatrix} \begin{bmatrix} \mathbf{1}_{0} & - \hat{{\cal K}}_{0 1} \\ - \hat{{\cal K}}_{1 0} & \mathbf{1}_{1} \end{bmatrix} \right)} \\
\hphantom{\tau = \det{\big( \operatorname{Id} - \mathcal{\hat{K}} \big)}}{} = \det{ \begin{bmatrix} \mathbf{1}_{0} - \hat{{\cal K}}_{0 1} \hat{{\cal K}}_{1 0} & 0 \\ - \hat{{\cal K}}_{1 0} & \mathbf{1}_{1} \end{bmatrix}} ,
\end{gather*}
where $\mathbf{1}_{0} = \operatorname{Id}_{L^{2} (\mathbb{R}_{-})}$, $\mathbf{1}_{1} = \operatorname{Id}_{L^{2} \left(\Gamma_{\infty}\right)}$. We thus see that the determinant of $\mathcal {\hat K}$ can be written as the Fredholm determinant of the operator on $L^2(\R_-)$ given by $\hat {\cal K}_{01}\hat {\cal K}_{10}$ and we are thus left with the equality
\begin{gather}\label{eq: det simplf}
\det{( \operatorname{Id} - \mathcal{\hat{K}} )} = \det \big[ \operatorname{Id}_{L^{2} ( \mathbb{R}_{-})} - \hat{{\cal K}}_{0 1} \hat{{\cal K}}_{1 0} \big].
\end{gather}
We define $\hat{\cal {F}}:= \hat{{\cal K}}_{0 1} \hat{{\cal K}}_{1 0}\colon L^2 (\R_-) \to L^2 (\R_-)$. The kernel of this operator is
is finally written as
\begin{gather}\label{eq: Fkernel int}
\big( \hat{\cal {F}} \varphi_0 \big) (z) := \big(\hat{{\cal K}}_{0 1} \hat{{\cal K}}_{1 0} \varphi_{0} \big) (z) = \int_{\mathbb{R}_{-}} F (z,w) \varphi_{0} (w) \mathrm{d} w ,\\
F (z,w) = \int_{\Gamma_{\infty}} K_{0,1} (z,x) K_{1,0} (x,w) \d x\nonumber\\
\hphantom{F (z,w)}{} = \frac{2 {\rm i} \sin{ (\pi \gamma)}}{4 \pi^{2}} |z w|^{\frac{\gamma}{2}} {\rm e}^{- \theta (w) - \theta (z)} \int_{\Gamma_{\infty}} \frac{x^{- \gamma} {\rm e}^{2 \theta (x)}}{ (z - x) (w - x)} \mathrm{d} x,\nonumber
\end{gather}
where we recall that $\theta (x) = x^2/4+sx/2$ so that the operator $\hat{\cal {F}}=\hat{\cal {F}} (\gamma, s)$ depends on the parameters~$s$ and~$\gamma$. We observe that $\det{\big( \operatorname{Id} - \mathcal{\hat{K}} \big)} $ is non zero as long as the sup-norm of the operator~$\hat{\cal {F}}$ is less than one. The values of~$s$ which satisfy this condition correspond to a regular solution of the Painlev\'e IV equation. We also observe that the kernel $F (z,w)$ is real-valued (and symmetric) and hence the tau function is real-valued.

\subsubsection[Norm estimate for small $s$]{Norm estimate for small $\boldsymbol{s}$}\label{Subsection: Norm estimate for small s}
For $s$ negative and sufficiently large in absolute value, we can estimate the norm of the operator $\hat{{\cal F}} = \hat{{\cal K}}_{0 1} \hat{{\cal K}}_{1 0}$ and guarantee that the determinant is non-zero.
Introducing the quantities
\begin{gather}
h (z) : = \int_{\Gamma_{\infty}} {\rm i} \frac{x^{- \gamma} {\rm e}^{\frac{x^{2}}{2} + { s } x}}{z - x} \mathrm{d} x ,\qquad
a (z) = |z|^{\frac{\gamma}{2}} {\rm e}^{- \frac{z^{2}}{4} - \frac{ s }{2} z} ,\qquad c_0= \frac{\sin{ (\pi \gamma)}}{2 \pi},\label{a(u)}
\end{gather}
we can write the operator $\hat{\cal {F}} $ in the form
\begin{gather*}
\big( \hat{\cal {F}} \varphi_0 \big) (z) = \int_{\mathbb{R}_{-}} F (z,w) \varphi_{0} (w) \mathrm{d} w = c_0 a (z) \int_{\mathbb{R}_{-}} a (w) \varphi_{0} (w)\frac{h (z) - h (w)}{w - z} \frac{\mathrm{d} w}\pi.
\end{gather*}

Next, we introduce the Hilbert transform $( \mathcal{P} \phi) (z) = \dashint_{\R_-} \frac{\phi (w)}{z - w} \frac {{\rm d} w}{\pi}$ which is a bounded operator from $L^2 (\R_-)$ to itself with norm one: $\opnorm{ \mathcal{P} } = 1$. Here and below, we denote by
\begin{gather*}
\opnorm{A} = \sup_{f \in L^2 ( \R_- ),\, f\not\equiv 0}\dfrac{\| A f\|_{L^2 (\R_-)}}{\|f\|_{L^2 (\R_-)}}.
\end{gather*}
Given an essentially bounded function $f \in L^\infty (\R_-)$, we denote by $M_f$ the corresponding multiplication operator and recall that $\opnorm{ M_{f}} = \|f\|_{\infty}$. This leads to
\begin{gather*}
\big\| \big( \hat{\cal{F}} \varphi_0 \big) (z) \big\|_{L^2} = \left\| \int_{\mathbb{R}_{-}} F (z,w) \varphi_{0} (w) \mathrm{d} w \right\|_{L^2} = | c_0 | \| ( M_{a h} \mathcal{P} M_{a} - M_{a} \mathcal{P} M_{a h} ) (\varphi_0) \|_{L^{2}} \\
\hphantom{\big\| \big( \hat{\cal{F}} \varphi_0 \big) (z) \big\|_{L^2} = \left\| \int_{\mathbb{R}_{-}} F (z,w) \varphi_{0} (w) \mathrm{d} w \right\|_{L^2} }{} \leq | c_0 | \big( \| M_{a h} \mathcal{P} M_{a} \varphi_0 \|_{L^{2}} + \| M_{a} \mathcal{P} M_{a h} \varphi_0 \|_{L^{2}} \big) \\
\hphantom{\big\| \big( \hat{\cal{F}} \varphi_0 \big) (z) \big\|_{L^2} = \left\| \int_{\mathbb{R}_{-}} F (z,w) \varphi_{0} (w) \mathrm{d} w \right\|_{L^2} }{} \leq | c_0 | \|\varphi_0\|_{L^2}\big( \opnorm{ M_{a h} \mathcal{P} M_{a} } + \opnorm {M_{a} \mathcal{P} M_{a h} } \big).
\end{gather*}

As a result of this procedure, we obtain the following estimate for the norm of $\hat{\cal{F}}$
\begin{gather*}
\big|\big|\big|\hat{\cal{F}}\big|\big|\big| \leq \frac{\sin{ (\pi \gamma)}}{\pi} \| a\|^{2}_{\infty} \| h\|_{\infty}.
\end{gather*}
A calculus exercise shows that the function $a (z)$ in \eqref{a(u)}, satisfies $\| a\|_{\infty} \leq 1$, for $z \in \R_-$, $\gamma \in (0,1)$, ${s} < 0$. We now estimate the values of $\|h\|_\infty$. By setting $x = \varepsilon + {\rm i} t$, $t \in \mathbb{R}$ in the formula for $h (z)$ in~\eqref{a(u)} and $s < 0$ and $\varepsilon = \max ( 0, -{s})$ we have
\begin{gather*}
| h (z) | \leq \int_{{\rm i} \mathbb{R} - { s }} | \mathrm{d} x| \frac{|x|^{- \gamma} {\rm e}^{\Re \big( \frac{x^{2}}{2} + { s } x \big)}}{ |z - x|} \qquad (|z - x| \geq -{ s } ) \\
\hphantom{| h (z) |}{} \leq \frac{1}{|s|} \int_{{\rm i} \mathbb{R} - { s }} | \mathrm{d} x | |x|^{- \gamma} {\rm e}^{\Re \big( \frac{x^{2}}{2} +{ s } x \big)} \leq \frac{1}{|s|} \int_{\mathbb{R}} \mathrm{d} t | {-}{s} + {\rm i} t |^{- \gamma} {\rm e}^{- \frac{t^{2}}{2} - \frac{{ s }^{2}}{2}} \\
\hphantom{| h (z) |}{} \leq \frac{{\rm e}^{- \frac{{ s }^{2}}{2}}}{|s|} |s|^{- \gamma} \int_{\mathbb{R}} \mathrm{d} t {\rm e}^{- \frac{t^{2}}{2}} = \frac{{\rm e}^{- \frac{{ s }^{2}}{2}}}{|s|^{1 + \gamma}} \sqrt{2 \pi}.
\end{gather*}
The estimate for $\big|\big|\big|\hat{\cal{F}}\big|\big|\big|$ becomes
\begin{gather}
\big|\big|\big|\hat{\cal{F}}\big|\big|\big| \leq \frac{\sin{ (\pi \gamma)}}{\pi} \frac{{\rm e}^{- \frac{{ s }^{2}}{2}}}{|s|^{1 + \gamma}} \sqrt{2 \pi}.\label{174}
\end{gather}
At this point, we need to see for which values of ${s}<0$ the norm of $\hat{\cal{F}}$ is smaller than~$1$, which guarantees that the determinant \eqref{eq: det simplf} is not zero and hence that our Riemann--Hilbert Problem~\ref{HRHP} admits a solution. From~\eqref{174} we have (recall $\gamma \in (0,1) $)
\begin{gather*}
\big|\big|\big|\hat{\cal{F}}\big|\big|\big| \leq \frac{\sin{ ( \pi \gamma )}}{\pi} \frac{\rm{e}^{- \frac{{ s }^{2}}{2}}}{|s|^{1 + \gamma}} \sqrt{2 \pi} \leq
\begin{cases}
\dfrac{{\sqrt{2}}}{\sqrt{\pi}} \dfrac{ {\rm e}^{-\frac{ s^2}2}}{ s^2}, & {s} \in [ -1, 0 ),\vspace{1mm}\\
\dfrac{{\sqrt{2}}}{\sqrt{\pi}} \dfrac{ {\rm e}^{-\frac{ s^2}2}}{ | s |}, & {s} <- 1,
\end{cases}
\end{gather*}
and we can easily see that the norm $\hat{\cal{F}}$ is less than one for ${s} < {s}_0$ where ${s}_0 \simeq -0.7701449782$. In summary, we have proven
\begin{Theorem}\label{thmexist}The Riemann--Hilbert Problem~{\rm \ref{HRHP}} admits a solution for $s \in ( - \infty, {s}_0)$, $s_0\simeq -0.7701449782$. In particular, the solution of the fourth Painlev\'e equation~\eqref{P4} for our choice of monodromy data \eqref{choice1}, \eqref{choice2} is pole-free within that range.
\end{Theorem}

It would be desirable to show that the solution of the Painlev\'e\ equation \eqref{P4} for our choice of monodromy data (\ref{choice1}), (\ref{choice2}) is pole-free on the whole real line. However, the numerical analysis which will be performed in the following section, shows that there is a discrete (in principle infinite) number of values of $s$ for which the Fredholm determinant~\eqref{eq: det simplf} vanishes and hence the solution of~\eqref{P4} has poles.

\subsection{Gaussian quadrature and numerics}\label{Subsection: Gaussian Quadrature and Numerics}

The following is an exposition, for the benefit of the reader, of the main idea of numerical eva\-luation of Fredholm determinants based on Nystrom's method, as explained by Bornemann~\cite{Bornemann}. The simple idea is to suitably ``discretize'' the integral operator using appropriate Gaussian quadrature (see~\cite{Szego}).

More specifically, the type of quadrature formul\ae\ that we will use are the so-called ``Gauss--Hermite" quadratures, which are defined in the following way
\begin{Definition}\label{defquad}Considering an integral of the form
\begin{gather*}
\int^{\infty}_{- \infty} f (x) {\rm e}^{- \Lambda x^{2}} \mathrm{d} x ,
\end{gather*}
the \textit{Gauss--Hermite quadrature} is an approximation of the value of this integral and is defined as
\begin{gather}\label{quadrule}
\int^{\infty}_{- \infty} f (x) {\rm e}^{- \Lambda x^{2}} \mathrm{d} x \simeq \sum^{m}_{i = 1} f \big( x^{(m)}_{i} \big) w^{(m)}_{i} ,
\end{gather}
where $m$ is the number of points used in this approximation: the points $\big\{x_j^{(m)}\big\}_{j=1}^m$ are called the {\it nodes} and the coefficients $\big\{w_j^{(m)}\big\}_{j=1}^m$ are called the {\it weights} of the quadrature rule. The nodes~$x^{(m)}_{j}$ are the roots of the $m$-th Hermite polynomial $H_{m} \big( \sqrt{\Lambda} x \big)$ and the weights~$ w^{(m)}_{i}$ are given by
\begin{gather*}
 w^{(m)}_{i} = \frac{2^{m - 1} m! \sqrt{\pi}}{\sqrt{\Lambda} m^{2} \big( H_{m - 1} \big( \sqrt{\Lambda} x^{(m)}_{i} \big) \big)^{2}}.
\end{gather*}
\end{Definition}
We will have to adapt this definition so that it is suitable to be applied to our case.

In order to use this technique for the computation of $F (x,y)$ in \eqref{eq: Fkernel int} we rewrite it as
\begin{gather*}
F (x,y) = {\rm e}^{\frac{- x^{2} - y^{2}}{4}} H (x,y) ,\\
H (x,y) := \frac{ \sin{ (\pi \gamma)}}{2 \pi^{2}} |x|^{\frac{\gamma}{2}} |y|^{\frac{\gamma}{2}} {\rm e}^{- \frac{s}{2} (x + y)} \int_{\mathbb{R}} \mathrm{d} t \frac{({\rm i} t + \varepsilon)^{- \gamma} {\rm e}^{- \frac{t^{2}}{2}} {\rm e}^{{\rm i} t (s+\varepsilon)} }{ (x - \varepsilon - {\rm i} t) (y - \varepsilon - {\rm i} t)} {\rm e}^{\frac{\varepsilon^{2}}{2} + s \varepsilon}.
\end{gather*}

Then Nystrom's method states that the Fredholm determinant is approximated by
\begin{gather*}
\det{}_{L^{2} (\mathbb{R}_{-})} \big[ \operatorname{Id}_{L^{2} (\mathbb{R}_{-})} -\hat{\cal{F}} \big] \simeq \det{}_{n \times n} \left[ \operatorname{Id}_{n} - \begin{bmatrix} H \big( x^{(2n)}_{i} , x^{(2n)}_{j} \big) \sqrt{w^{(2n)}_{i} w^{(2n)}_{j}} \end{bmatrix}^{n}_{i, j = 1} \right] ,
\end{gather*}
where the nodes $x^{(2n)}_{i}$ and the weights $w^{(2n)}_{i}$ are chosen as in~(\ref{quadrule}) with $\Lambda=\frac{1}{2}$ and where we are selecting only half of the nodes, namely~$n$ out of~$2n$ that lie on the negative real axis. Therefore, the matrix $\mathbb{H} = [ H_{jk}]_{j,k=1}^n$ that we are interested in computing is given by
\begin{gather*}
\begin{bmatrix} H_{j k} \end{bmatrix}^{n}_{j , k = 1}:=
 \begin{bmatrix} H \big( x^{(2n)}_{j} , x^{(2n)}_{k} \big) \sqrt{w^{(2n)}_{j} w^{(2n)}_{k}} \end{bmatrix}^{n}_{j, k = 1},
\end{gather*}
where
\begin{gather*}
H_{j k} = \frac{{\rm e}^{\frac{\varepsilon^{2}}{2} + s \varepsilon} \sin{ (\pi \gamma)}}{2 \pi^{2}} \big| x^{(2n)}_{j} x^{(2n)}_{k} \big|^{\frac{\gamma}{2}} {\rm e}^{- \frac{s}{2} \big( x^{(2n)}_{j} + x^{(2n)}_{k} \big)}\\
\hphantom{H_{j k} =}{}\times \sqrt{w^{(2n)}_{j} w^{(2n)}_{k}} \int_{\mathbb{R}} \mathrm{d} t \frac{{\rm e}^{- \frac{t^{2}}{2}} ({\rm i} t + \varepsilon)^{- \gamma} {\rm e}^{{\rm i} t (s+\varepsilon)}}{ \big( x^{(2n)}_{j} - \varepsilon - {\rm i} t \big) \big( x^{(2n)}_{k} - \varepsilon - {\rm i} t \big)}.
\end{gather*}
In order to compute the above integral, we once again apply the Gauss--Hermite quadrature
\begin{gather*}
\int_{\mathbb{R}} \mathrm{d} t \frac{{\rm e}^{- \frac{t^{2}}{2}} ({\rm i} t + \varepsilon)^{- \gamma} {\rm e}^{{\rm i} t (s+\varepsilon)}}{ \big( x^{(2n)}_{j} - \varepsilon - {\rm i} t \big) \big( x^{(2n)}_{k} - \varepsilon - {\rm i} t \big)} = \sum^{2 n}_{\ell = 1} \frac{w^{(2n)}_{\ell} \big( {\rm i} x^{(2n)}_{\ell} + \varepsilon \big)^{- \gamma} {\rm e}^{{\rm i} x^{(2n)}_{\ell} (s+\varepsilon)} }{ \big( x^{(2n)}_{j} - \varepsilon - {\rm i} x^{(2n)}_{\ell} \big)\big( x^{(2n)}_{k} - \varepsilon - {\rm i} x^{(2n)}_{\ell} \big)}.
\end{gather*}
Note that, in the above, we use all the $2n$ nodes because the integral is over the whole $\R$. In principle, we could use a different number of nodes for this last quadrature, but this is the way the code was actually implemented for simplicity.

The matrix $\mathbb{H} = \left[ H_{jk} \right]_{j,k=1}^n$ can now be written as
\begin{gather*}
\mathbb{H}=\begin{bmatrix} H_{j k} \end{bmatrix}^{n}_{j , k = 1} = \frac{{\rm e}^{\frac{\varepsilon^{2}}{2} + s \varepsilon}}{2 \pi^{2}} \sin{ (\pi \gamma)} \mathbb{A}. \mathbb{A}^{\rm T} ,
\end{gather*}
where $\mathbb{A}$ is the $n\times 2n$ matrix given by
\begin{gather*}
\mathbb{A}_{j \ell} = \left| x_{j} \right|^{\frac{\gamma}{2}} \sqrt{w_{j}} {\rm e}^{- \frac{s}{2} x_{j}} \frac{({\rm i} x_{\ell} + \varepsilon )^{- \frac{\gamma}{2}} {\rm e}^{{\rm i}\frac{x_{\ell}}{2} (s+\varepsilon)} }{\left( x_{j} - \varepsilon - {\rm i} x_{\ell} \right)} \sqrt{w_{\ell}} ,
\end{gather*}
and $1 \leq j \leq n$ and $1 \leq \ell \leq 2 n$.

The result of this approximation is expressed in terms of the parameters $s\in\R$, $\gamma \in \left[ 0, 1 \right)$ and $n\in\N$, which is the number of nodes in the Gaussian quadrature. Therefore, the final approximation \ that we have for $\tau$ is
\begin{gather*}
\tau ( s, \gamma, n ) = \det{\left( \mathbf{1}_{n \times n} - \frac{{\rm e}^{\frac{\varepsilon^{2}}{2} + s \varepsilon}}{2 \pi^{2}} \sin{ (\pi \gamma)} \mathbb{A}. \mathbb{A}^{\rm T} \right)}.
\end{gather*}
We take $\varepsilon = \max{[ 0 , - s]}$ so as to optimize the distance from the saddle-point of the function $\theta (x)=x^2/4+sx/2$ while respecting the condition that the relative interior of the contours~$\R_-$ and~$\Gamma_\infty$ do not intersect. It is worth pointing out that the value of the Fredholm determinant (numerical considerations aside) is independent of $\varepsilon$ thanks to Cauchy's theorem.
Various numerical results are pictures in Figs.~\ref{fig: tau n=30}, \ref{fig: tau n=80}, \ref{fig: tau gamma 0,1 n=150 extended}.

\begin{figure}[t]\centering
\begin{minipage}[b]{75mm} \centering
 \includegraphics[width=0.95\linewidth]{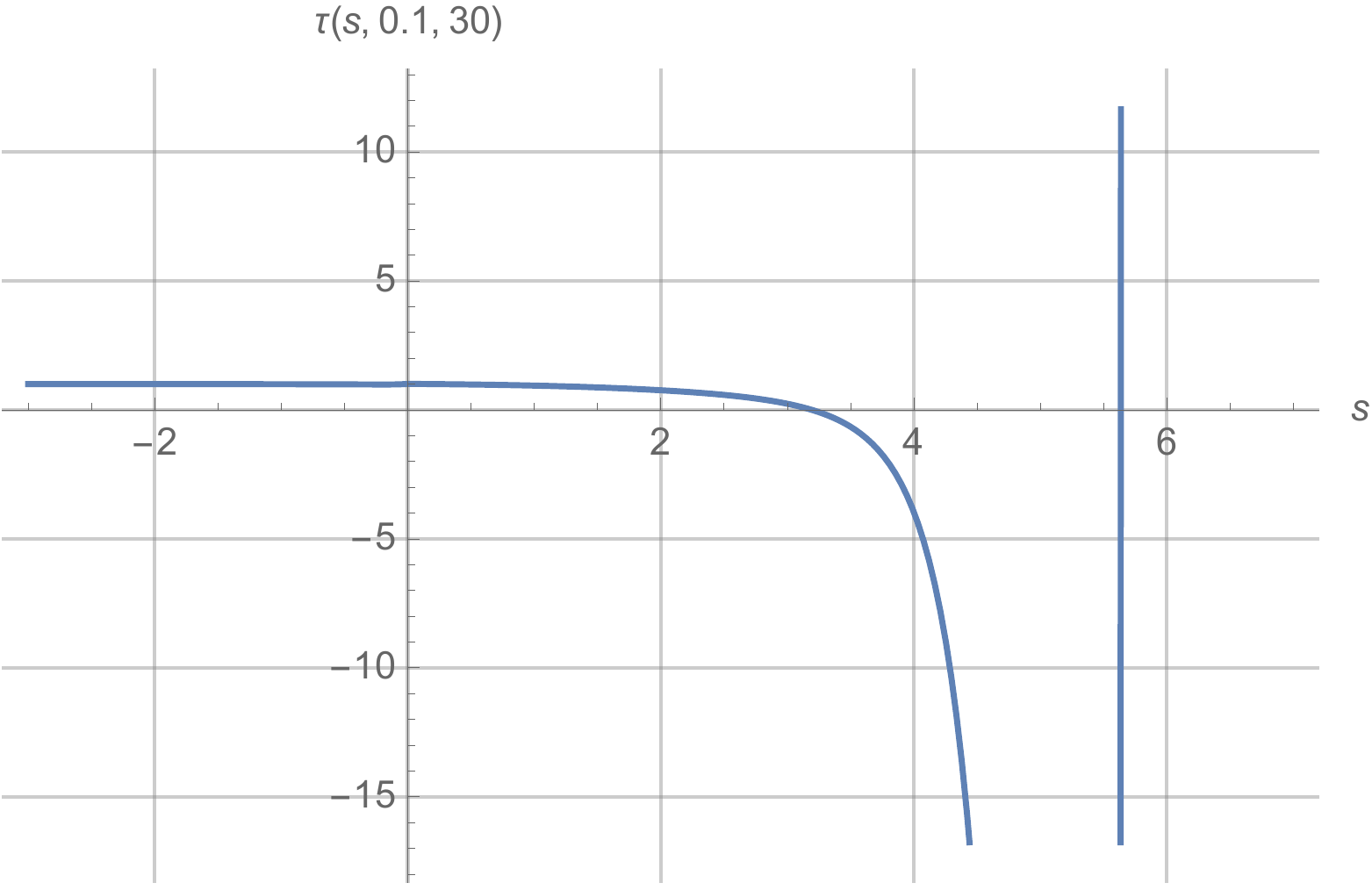}
{\footnotesize (a) $\gamma = 0.1$.} 
\end{minipage} \quad
\begin{minipage}[b]{75mm} \centering
 \includegraphics[width=0.95\linewidth]{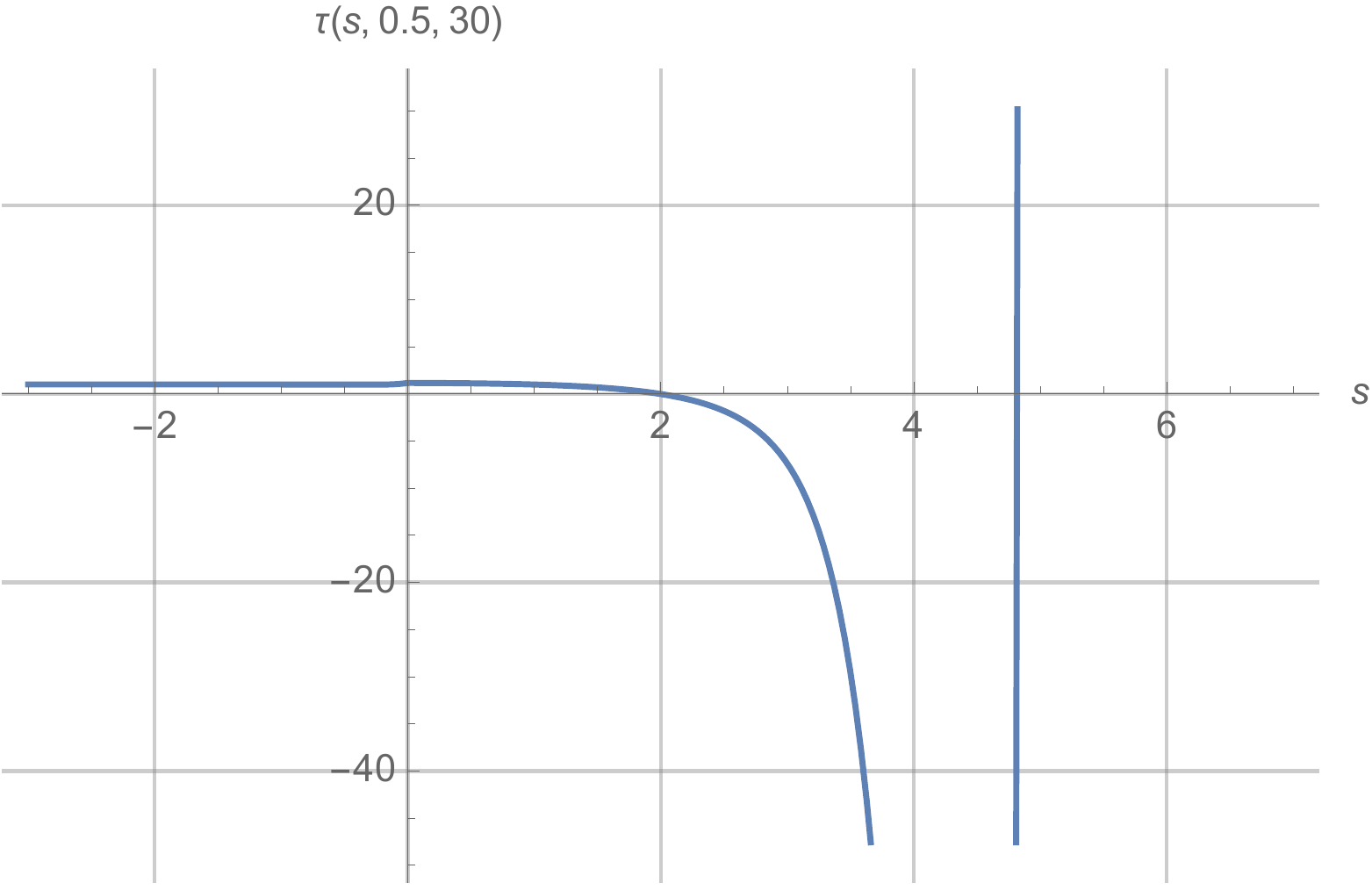}
{\footnotesize (b) $\gamma = 0.5$.} 
\end{minipage}
\caption{$\tau$-function with $n = 30$ points.}\label{fig: tau n=30}
\end{figure}

\begin{figure}[t]\centering
\begin{minipage}[b]{75mm} \centering
 \includegraphics[width=0.95\linewidth]{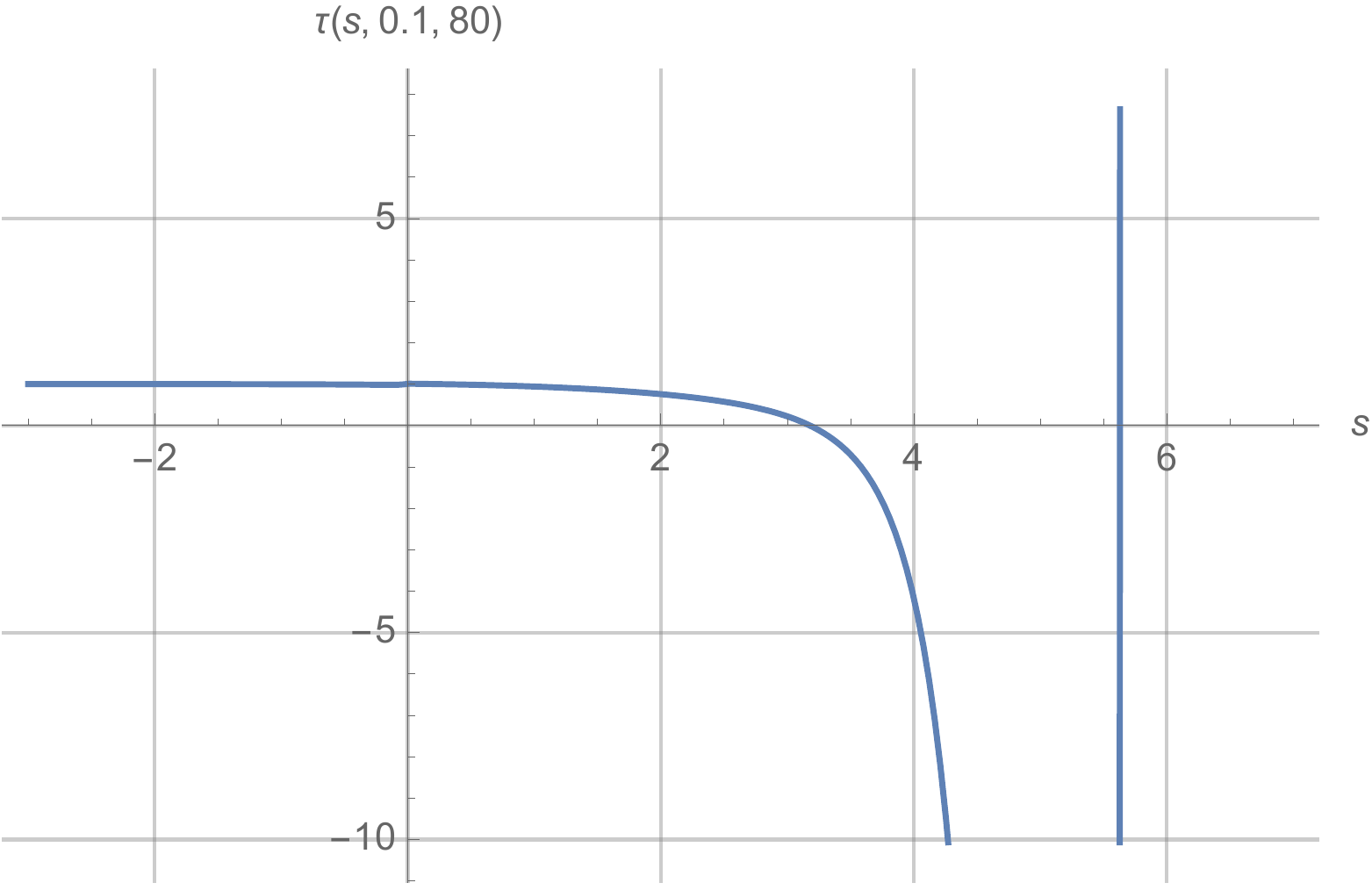}\\
{\footnotesize (a) $\gamma = 0.1$.} 
\end{minipage} \quad
\begin{minipage}[b]{75mm} \centering
 \includegraphics[width=0.95\linewidth]{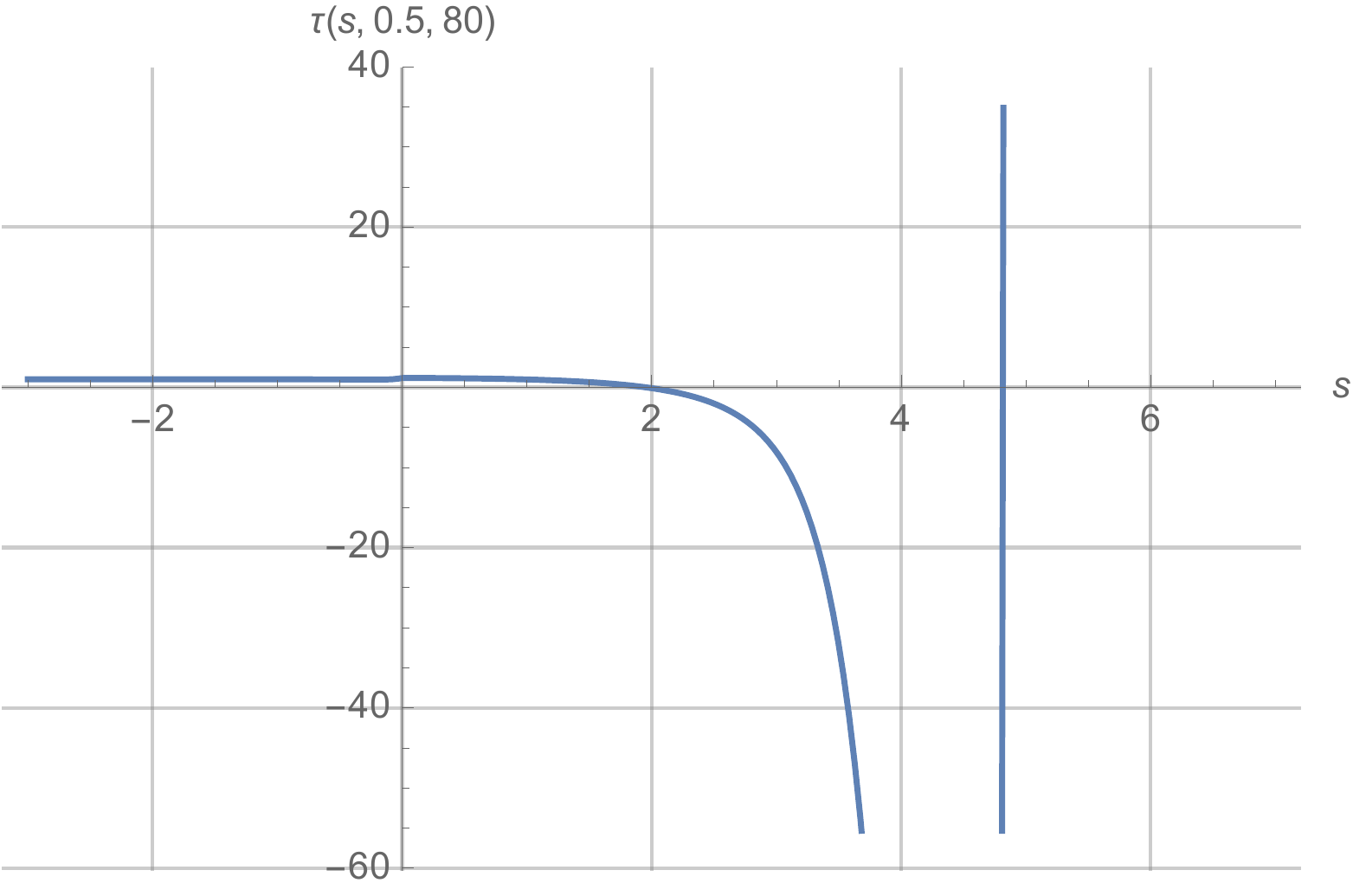}\\
{\footnotesize (b) $\gamma = 0.5$.} 
\end{minipage}
\caption{$\tau$-function with $n = 80$ points.}\label{fig: tau n=80}
\end{figure}

\begin{figure}[t]\centering
\includegraphics[width=80mm]{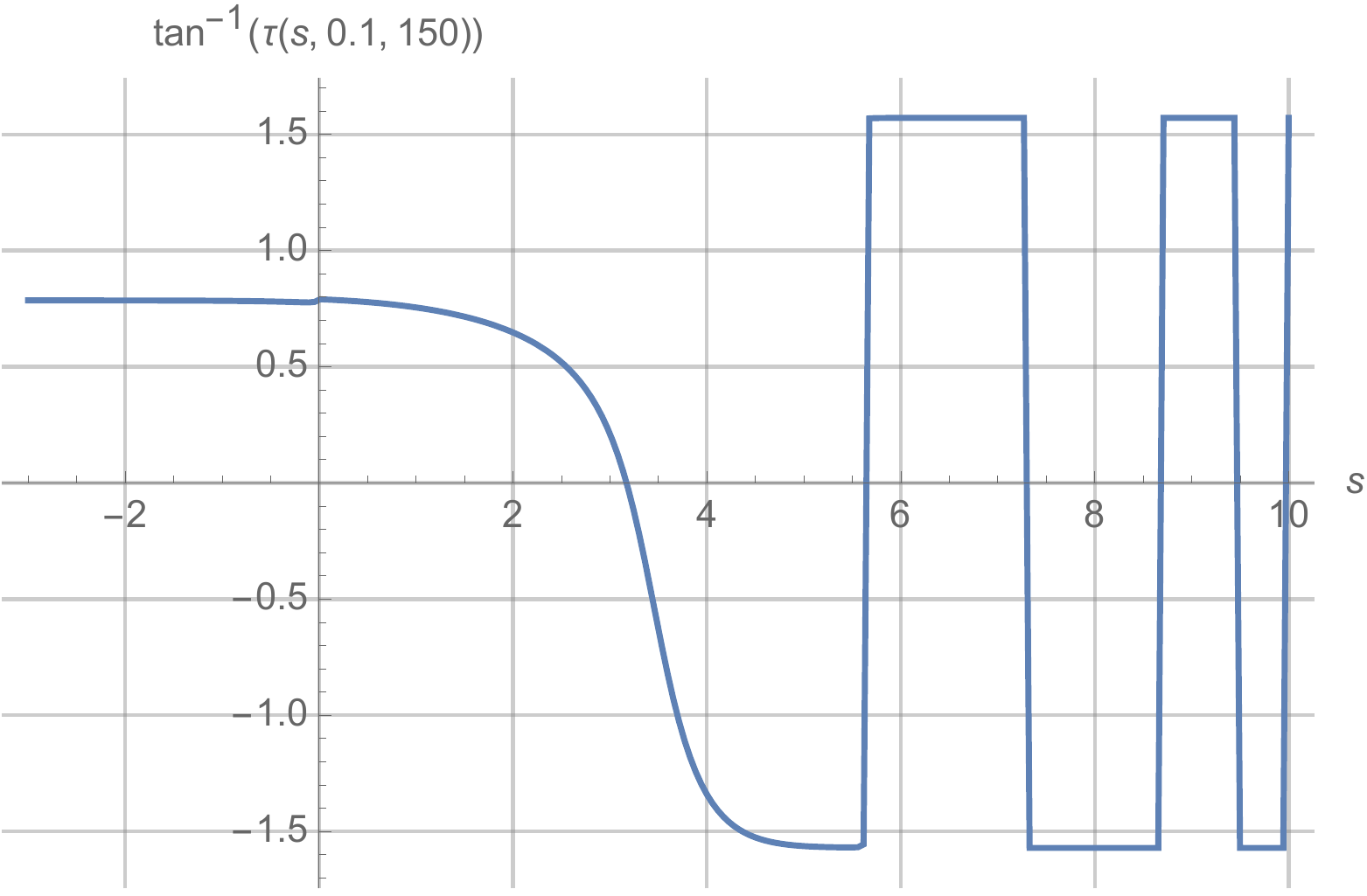}
\caption{Arctan of the $\tau$-function for $\gamma = 0.1$ and with $n=150$ with an extended range of values for~$s$ being considered.}\label{fig: tau gamma 0,1 n=150 extended}
\end{figure}

\section{Orthogonal polynomials and their Riemann--Hilbert problem}\label{Section: Orthogonal Polynomials and the Riemann--Hilbert Problem}

The orthogonality relations for the polynomials $p_n$ can be rewritten
\begin{gather*}
\int_{\mathbb{C}} p_n (\lambda) \bar{\lambda}^{j d + \ell} {\rm e}^{-N W (\lambda)} \mathrm{d} A (\lambda) = 0 , \qquad j = 0 , \dots , k - 1,\nonumber\\
n = k d + \ell , \qquad 0 \leq \ell \leq d-1.
\end{gather*}
As a consequence, the $n$-th monic orthogonal polynomial has a discrete symmetry
\begin{gather*}
p_n \big( {\rm e}^{\frac{2\pi {\rm i}}{d}} \lb \big) = {\rm e}^{\frac{2 \pi {\rm i} n}{d}} p_n (\lambda) \quad \Rightarrow \quad p_n (\lambda) = \lambda^{\ell} q^{(\ell)}_k \big( \lambda^d \big).
\end{gather*}
Using this definition, the initial sequence of orthogonal polynomials $\{ p_n (\lambda)\}^{\infty}_{n = 0}$ can be split in $d$ sub-sequences, each of which labelled by the remainder $\ell \equiv n$ mod~$d$. Through a change of coordinate $\lb^d=u$ these sequences of monic orthogonal polynomials are seen to satisfy the orthogonality relations
\begin{gather}\label{eq:ort_reduced}
\int_{\mathbb{C}} q_k^{( \ell)} (u) \bar{u}^{j} |u|^{-2\gamma} {\rm e}^{-N ( |u|^2 - t u- t \bar{u} )} \mathrm{d} A (u) = 0, \qquad j=0,\dots, k-1 ,\\
\gamma := \frac{d - \ell - 1}{d} \in [0,1).\nonumber
\end{gather}
With the further change of coordinate $u = - t (z-1)$, $z \in \mathbb{C}$, we introduce the transformed monic polynomial
\begin{gather}\label{pikdef}
\pi_k (z) := \frac{ ( - 1)^k}{t^k} q^{( \ell )}_k ( - t (z-1)) ,
\end{gather}

\begin{figure}[t]\centering
\includegraphics[width=65mm]{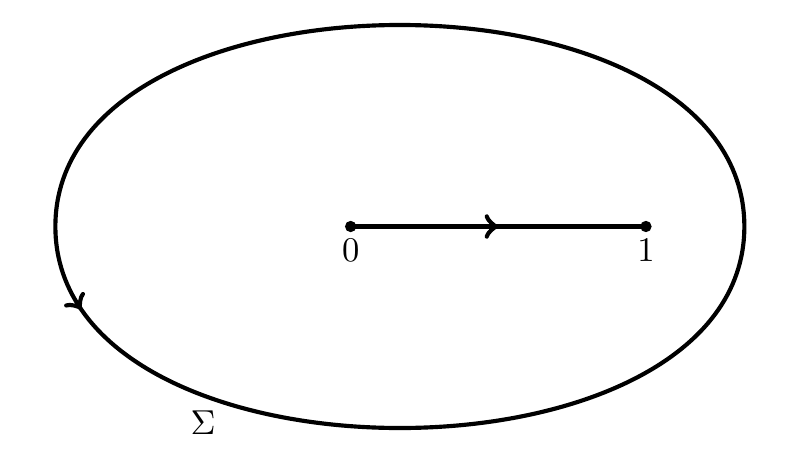}
\caption{The contour $\Sigma$.}\label{fig:cont_sigma}
\end{figure}
Next we transform the orthogonality on the plane to the orthogonality on a contour as in \cite{BGM} and \cite{BBMS}
\begin{Theorem}[{\cite[Theorem~2.1]{BGM}}]\label{Theorem: contour change}For any polynomial $q (u)$ the following identity is satisfied
\begin{gather*}
\int_{\mathbb{C}} \! q (u) \bar{u}^j |u|^{-2 \gamma} {\rm e}^{-N ( |u|^2 - t u - t \bar{u} )} \mathrm{d} A (u) = \frac{\pi \Gamma ( j - \gamma + 1)}{N^{j-\gamma+1}} \frac{1}{2\pi {\rm i}} \oint_{\tilde\Sigma} \! q (u) \frac{{\rm e}^{N t u}}{( u - t)^{j+1}} \left( 1 - \frac{t}{u} \right)^{\gamma} \!\mathrm{d} u ,
\end{gather*}
where $\gamma \in (0,1)$, $j$ is an arbitrary non-negative integer, and $\tilde\Sigma$ is a positively oriented simple closed loop enclosing $u=0$ and $u=t$.
\end{Theorem}

Consider now $q^{(\ell)}_k (u)$ to be the monic polynomial of degree $k$ with the orthogonality relations~\eqref{eq:ort_reduced}. Then according to Theorem~\ref{Theorem: contour change} the transformed monic polynomial $\pi_k(z)$ of deg\-ree~$k$ given by~\eqref{pikdef}, is characterised by the non-hermitian orthogonality relations
\begin{gather}\label{eq: pol line contour}
\oint_{\Sigma} \pi_k (z) z^j \frac{{\rm e}^{-N t^2 z}}{z^k} \left(\frac{z}{z-1}\right)^{\gamma} \mathrm{d} z = 0, \qquad j = 0 ,1,\dots ,k-1 ,
\end{gather}
where $\Sigma$ is a simple, positively oriented contour encircling $z=0$ and $z=1$, as it can be seen in Fig.~\ref{fig:cont_sigma}, and the function $\big(\frac{z}{z-1}\big)^{\gamma}$ is analytic in $\mathbb{C} \setminus [0,1]$ and tends to one for $|z| \to \infty$.

\subsection{The Riemann--Hilbert problem}\label{Subsection: The Riemann--Hilbert Problem}

We now consider the polynomials $\pi_k (z)$ in the limit $k \rightarrow \infty$ and $N \rightarrow \infty$ in such a way that, for $n = k d + \ell$, one has
\begin{gather*}
T = \frac{n - \ell}{N} > 0.
\end{gather*}
We set the notation
\begin{gather*}
V (z) = \frac{z}{z_0} + \log{z}, \qquad z_0 = z_{0} ( t, N ) = \frac{k}{N} \frac{1}{t^{2}} , \qquad z_0 = \frac{t^2_c}{t^2} , \qquad t^2_c = \frac{T}{d} ,\nonumber\\
w_k (z) := {\rm e}^{ - k V (z) } \left( \frac{z}{z-1} \right)^{\gamma}.
\end{gather*}
The orthogonality relations \eqref{eq: pol line contour} now read
\begin{gather*}
\oint_{\Sigma} \pi_k (z) z^j w_k (z) \mathrm{d} z = 0 , \qquad j=0,1,\dots ,k-1.
\end{gather*}
When the limit $k \rightarrow \infty$ is taken, three different regimes arise
\begin{itemize}\itemsep=0pt
\item \textit{pre-critical} case: $0 < t < t_c$, leading to $z_0 > 1$,
\item \textit{critical} case: $t = t_c$, leading to $z_0 = 1$,
\item \textit{post-critical} case: $t > t_c$, leading to $z_0 <1$.
\end{itemize}
The pre- and post-critical case were already analyzed in \cite{BGM} and we are now interested in analyzing the critical case. We begin by defining the so-called complex moments as
\begin{gather*}
\nu_j : = \oint_{\Sigma} z^j w_k (z) \mathrm{d} z,
\end{gather*}
where the dependency on $k$ is omitted in order to simplify notation, and use this to introduce the auxiliary polynomial
\begin{gather*}
\Pi_{k-1} (z) : = \frac 1{\det [\nu_{i+j} ]_{0\leq i,j\leq k-1}} \det \left[
\begin{matrix}
\nu_0 & \nu_1 & \dots & \nu_{k-1}\\
\nu_1 & \nu_2 & \dots & \nu_{k}\\
\vdots &&&\vdots \\
\nu_{k-2}&\dots&&\nu_{2k-3}\\
1 & z & \dots & z^{k-1}
\end{matrix}
\right].
\end{gather*}
It can be seen that this polynomial is not necessarily monic and its degree may be less than $k - 1$. In order to guarantee the existence of such a polynomial, the determinant in the denominator must not vanish.
\begin{Proposition}[{\cite[Proposition~2.2]{BGM}}] The determinant $\det [ \nu_{i+j} ]_{0\leq i,j\leq k-1}$ does not vanish and is given by
\begin{gather*}
\det [ \nu_{i+j}]_{0\leq i,j\leq k-1} = (-1)^{k (k - 1) / 2} ( 2 \mathrm{i} )^{k} \left( \prod_{j=0}^{k-1} t^{2j+2\gamma-2} \frac{N^{j-\gamma+1}}{\Gamma ( j - \gamma + 1 ) } \right) \\
\hphantom{\det [ \nu_{i+j}]_{0\leq i,j\leq k-1} =}{}\times \det \left[\iint_{\mathbb{C}} z^i \bar{z}^j |z - 1|^{-2\gamma}{\rm e}^{-N t^2 |z|^2} \mathrm{d} A (z) \right].
\end{gather*}
\end{Proposition}
We can now reformulate the condition of orthogonality for the polynomials $\pi_k (z)$ as a~Rie\-mann--Hilbert boundary value problem. To do this, we define the matrix
\begin{gather*}
Y (z) = \begin{bmatrix}
\pi_{k} (z) & \dfrac{\strut1}{\strut2 \pi {\rm i}} \displaystyle{\int_{\Sigma}} \dfrac{\strut \pi_{k} (z')}{\strut z' - z} w_{k} (z') \mathrm{d} z' \\
-{2 \pi {\rm i} } \Pi_{k-1} (z) & - \displaystyle{\int_{\Sigma}} \dfrac{\strut\Pi_{k-1} (z')}{\strut z' - z} w_{k} (z') \mathrm{d} z'
\end{bmatrix} ,
\end{gather*}
which is the unique solution of following Riemann--Hilbert problem for orthogonal polyno\-mials~\cite{FIK}.

\begin{figure}[t]\centering
\includegraphics[width=70mm]{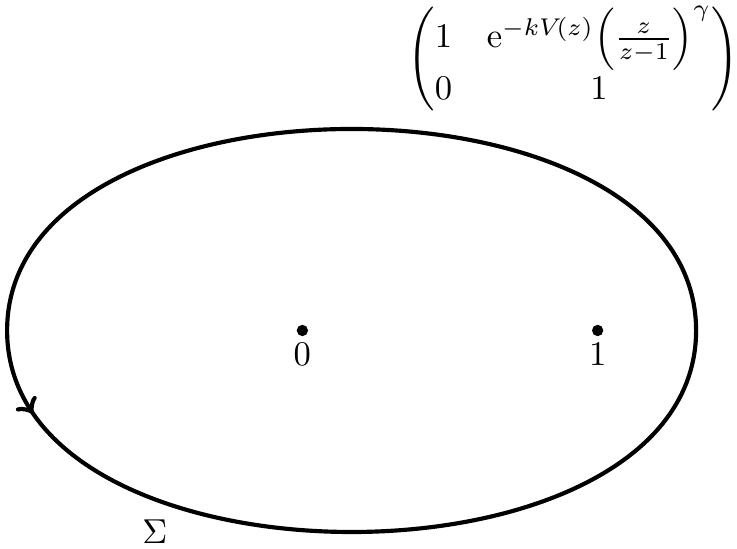}
\caption{Jump in the $Y (z)$ Riemann--Hilbert problem.}\label{fig: jumps Y}
\end{figure}

\begin{problem}\label{YRHP}The matrix $Y (z)$ is analytic in $ \mathbb{C} \setminus \Sigma$, where $\Sigma$ is the oriented curve in Fig.~{\rm \ref{fig: jumps Y}}. The limits $Y_{\pm} (z)$ exist and are continuous along~$ \Sigma $. Moreover
\begin{enumerate}\itemsep=0pt
\item[$1.$] Jump on $\Sigma$: The continuous boundary values $Y_{\pm} (z)$ are such that
\begin{gather*}
Y_{+} (z) = Y_{-} (z) \begin{pmatrix}1& w_k (z) \\ 0 & 1 \end{pmatrix} , \qquad z \in \Sigma.
\end{gather*}
This jump can be seen in Fig.~{\rm \ref{fig: jumps Y}}.

\item[$2.$] Behaviour at infinity: $Y (z)$ has the following behaviour as $z \rightarrow \infty$
\begin{gather*}
Y (z) = \left( \mathbf{1} + \mathcal{O} \left( \frac{1}{z} \right) \right) z^{k \sigma_{3}}.
\end{gather*}
\end{enumerate}
\end{problem}

\section{Riemann--Hilbert analysis}\label{Section: Riemann--Hilbert analysis}\label{section4}

\subsection{Transforming the Riemann--Hilbert problem}\label{Subsection: Transforming the Riemann--Hilbert Problem}

We will now begin by doing a simple transformation through the first undressing step and, after that, use the steepest descent method of Deift--Zhou~\cite{DeiftZhou} to proceed with further transformations and simplifications of our problem.

\subsubsection{First undressing step}\label{Subsection: First Undressing Step}

The first undressing step consists in a simplification of the Riemann--Hilbert Problem~\ref{YRHP} that is done by defining a new matrix-valued function~$\tilde{Y} (z)$ as
\begin{gather*}
\tilde{Y} (z) : = Y (z) \left( 1 - \frac{1}{z} \right)^{- \frac{\gamma}{2} \sigma_3} , \qquad z \in \mathbb \setminus ( \Sigma\cup [0,1] ).
\end{gather*}
Clearly the matrix $\tilde Y$ now satisfies a new Riemann--Hilbert problem whose simple formulation we leave to the reader. From the solution of Riemann--Hilbert problem we can recover the orthogonal polynomials through
\begin{gather*}
\pi_k (z) = \tilde{Y}_{1 1} (z) \left( 1 - \frac{1}{z} \right)^{\frac{\gamma}{2}}.
\end{gather*}

\subsubsection[Transformation $\tilde{Y} \rightarrow {R}$]{Transformation $\boldsymbol{\tilde{Y} \rightarrow {R}}$}\label{Subsection: Transformation tildeY to U}

We now define the so-called $g$-function; this is a scalar function that is analytic off a contour $\Gamma_r$ (see Fig.~\ref{fig:tikz_critical}), homotopically equivalent to $\Sigma$ in $\mathbb{C} \setminus [ 0, 1]$ and suitably chosen. Following~\cite{BGM} let
\begin{gather}\label{gint}
g (z) = \begin{cases}
 \dfrac{z}{z_0} + \ell , & z \in \inte ( \Gamma_r ), \\
 \log{z} , & z \in \exte ( \Gamma_r ) ,
\end{cases}\\
\ell := \log{r} - \dfrac{r}{z_0} , \qquad r > 0,\nonumber 
\end{gather}
where the logarithm is the principal determination and
\begin{gather}\label{eq: fam of contours}
\Gamma_{r} = \big\{ z \in \mathbb{C} \colon \Re ( \varphi (z; r)) = 0 ,\, |z| \leq z_{0} \big\} , \qquad 0 < r \leq z_{0} ,
\end{gather}
and
\begin{gather*}
\varphi (z; r) =
\begin{cases}
\log{z} - \dfrac{z}{z_0} - \log{r} + \dfrac{r}{z_0} , & z \in \inte ( \Gamma_r ), \vspace{1mm}\\
\dfrac{z}{z_0} - \log{z} + \log{r} - \dfrac{r}{z_0} , & z \in \exte ( \Gamma_r),
\end{cases}
\end{gather*}
or equivalently,
\begin{gather}\label{eq: phi function def}
\varphi (z; r) = \log{r} - \dfrac{r}{z_0} + V - 2 g (z).
\end{gather}

 Note that $g$ is such that
\begin{gather}
g_+ (z) + g_- (z) - \ell - V (z) \equiv 0 ,\qquad \forall\, z \in \Gamma_r \label{nullg}.
\end{gather}

In order to establish the transformation $\tilde{Y} \rightarrow {R}$ we will now deform the contour $\Sigma$ to the contour $\Gamma_r$. The choice of $r$ will be explained later. For the time being we define the matrix ${R}$ as
\begin{gather*}
{R} (z) = {\rm e}^{- k \frac{\ell}{2} \sigma_3} \tilde{Y} (z) {\rm e}^{- k g (z) \sigma_3} {\rm e}^{k \frac{\ell}{2} \sigma_3} , \qquad z \in \mathbb{C} \setminus ( \Gamma_r \cup [0,1] ).
\end{gather*}
The matrix ${R} (z)$ solves the following Riemann--Hilbert problem

\begin{problem} \label{URHP} The matrix ${R} (z)$ is analytic in $\mathbb{C} \setminus ( \Gamma_r \cup [0,1] )$ and admits non-tangential boundary values that satisfy:
\begin{enumerate}\itemsep=0pt
\item[$1.$] Jumps on $\Gamma_r$ and $[0,1]$ $($see Fig.~{\rm \ref{fig: jumps U})}:
\begin{gather}\label{eq: jump mat U}
{R}_{+} (z) = {R}_{-} (z)
\begin{cases}
 \begin{pmatrix} {\rm e}^{- k ( g_{+} - g_{-} )} & {\rm e}^{k ( g_{+} + g_{-} - \ell - V (z))} \\ 0 & {\rm e}^{ k ( g_{+} - g_{-})} \end{pmatrix} , & z \in \Gamma, \\
{\rm e}^{-\gamma\pi {\rm i}\sigma_3} , & z \in (0,1).
\end{cases}
\end{gather}
\begin{figure}[t]\centering
\includegraphics[width=80mm]{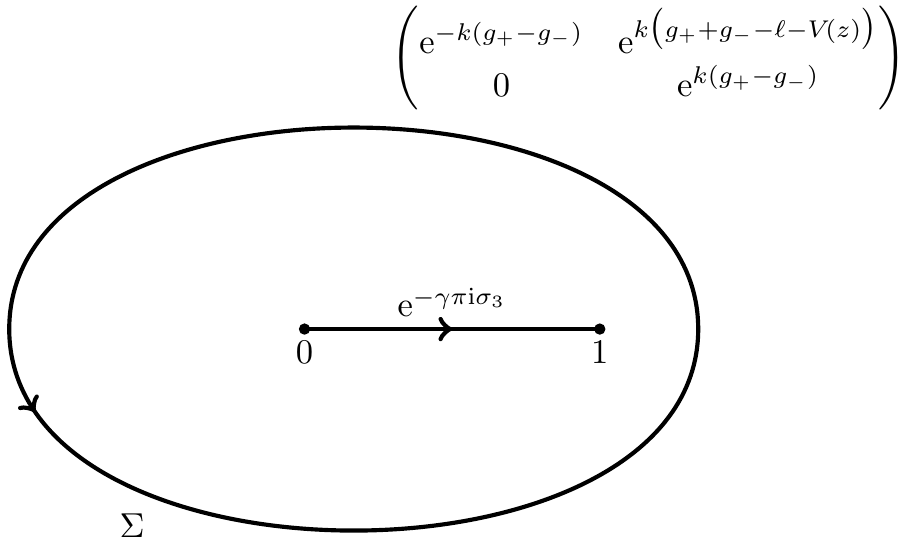}
\caption{Jumps for the Riemann--Hilbert problem of the matrix ${R}(z)$.}\label{fig: jumps U}
\end{figure}

\item[$2.$] Large $z$ boundary behaviour:
\begin{gather*}
{R} (z) = \left( \mathbf{1} + \mathcal{O} \left( \frac{1}{z} \right) \right) , \qquad z \rightarrow \infty.
\end{gather*}

\item[$3.$] Endpoint behaviour:
\begin{gather*}
{R} (z) = \mathcal{O} (1) z^{\frac{\gamma}{2}\sigma_3} , \qquad z \rightarrow 0 ,\\
{R} (z) = \mathcal{O} (1) (z-1)^{- \frac{\gamma}{2}\sigma_3} , \qquad z \to 1.
\end{gather*}
\end{enumerate}
\end{problem}
The orhtogonal polynomials can be recovered from this Riemann--Hilbert problem through
\begin{gather}\label{pi_U}
\pi_k (z) = {R}_{1 1} (z) {\rm e}^{k g (z)} \left( 1 - \frac{1}{z} \right)^{\frac{\gamma}{2}}.
\end{gather}

\subsubsection[Transformation ${R} \rightarrow T$]{Transformation $\boldsymbol{{R} \rightarrow T}$}\label{Subsection: Transformation U to T}

The jump matrix of ${R} (z)$ on $z \in \Gamma$ given by \eqref{eq: jump mat U} can be factorized in the following way
\begin{gather*}
 \begin{pmatrix} {\rm e}^{- k \left( g_{+} - g_{-} \right)} & {\rm e}^{k \left( g_{+} + g_{-} - \ell - V \right)} \\ 0 & {\rm e}^{ k \left( g_{+} - g_{-} \right)} \end{pmatrix} \\
 \qquad{} = \begin{pmatrix} 1 & 0 \\ {\rm e}^{k \left( \ell + V - 2 g_{-} \right)} & 1 \end{pmatrix} \begin{pmatrix} 0 & {\rm e}^{k \left( g_{+} + g_{-} - \ell - V \right)} \\ -{\rm e}^{k \left( g_{+} + g_{-} - \ell - V \right)} & 0 \end{pmatrix} \begin{pmatrix} 1 & 0 \\ {\rm e}^{k \left( \ell + V - 2 g_{+} \right)} & 1 \end{pmatrix} \\
\qquad{} = \begin{pmatrix} 1 & 0 \\ {\rm e}^{k \varphi (z)} & 1 \end{pmatrix} \begin{pmatrix} 0 & 1 \\ -1 & 0 \end{pmatrix} \begin{pmatrix} 1 & 0 \\ {\rm e}^{k \varphi (z)} & 1 \end{pmatrix} ,
\end{gather*}
where, in order to write the second line, we have used the definition of the $\varphi (z)$ function defined in \eqref{eq: phi function def} and the function $g (z)$ defined in~\eqref{gint} as well as the relation~\eqref{nullg}.

We will now consider three different loops $\Gamma_{i}$, $\Gamma_{r}$ and $\Gamma_{e}$, so that the space is split into four different domains $\Omega_{0}$, $\Omega_{1}$, $\Omega_{2}$ and $\Omega_{\infty}$, as it can be seen in Fig.~\ref{fig: jumps T}. $\Gamma_i$ is in the interior of $\Gamma_r$ and $\Gamma_e$ is in the exterior. Using this, we will now define a new matrix-valued function $T (z)$ in the following way
\begin{gather}\label{def_T}
T (z)= \begin{cases}
{R} (z) , & z \in \Omega_{0}\cup \Omega_{\infty}, \\
{R} (z) \begin{pmatrix} 1 & 0 \\ -{\rm e}^{k \varphi (z)} & 1 \end{pmatrix} , & z \in \Omega_{1},\vspace{1mm}\\
{R} (z) \begin{pmatrix} 1 & 0 \\ {\rm e}^{k \varphi(z)} & 1 \end{pmatrix} , & z \in \Omega_{2}.
\end{cases}
\end{gather}
This matrix $T (z)$ satisfies the following Riemann--Hilbert problem
\begin{figure}[t]\centering
\includegraphics[width=12.5cm]{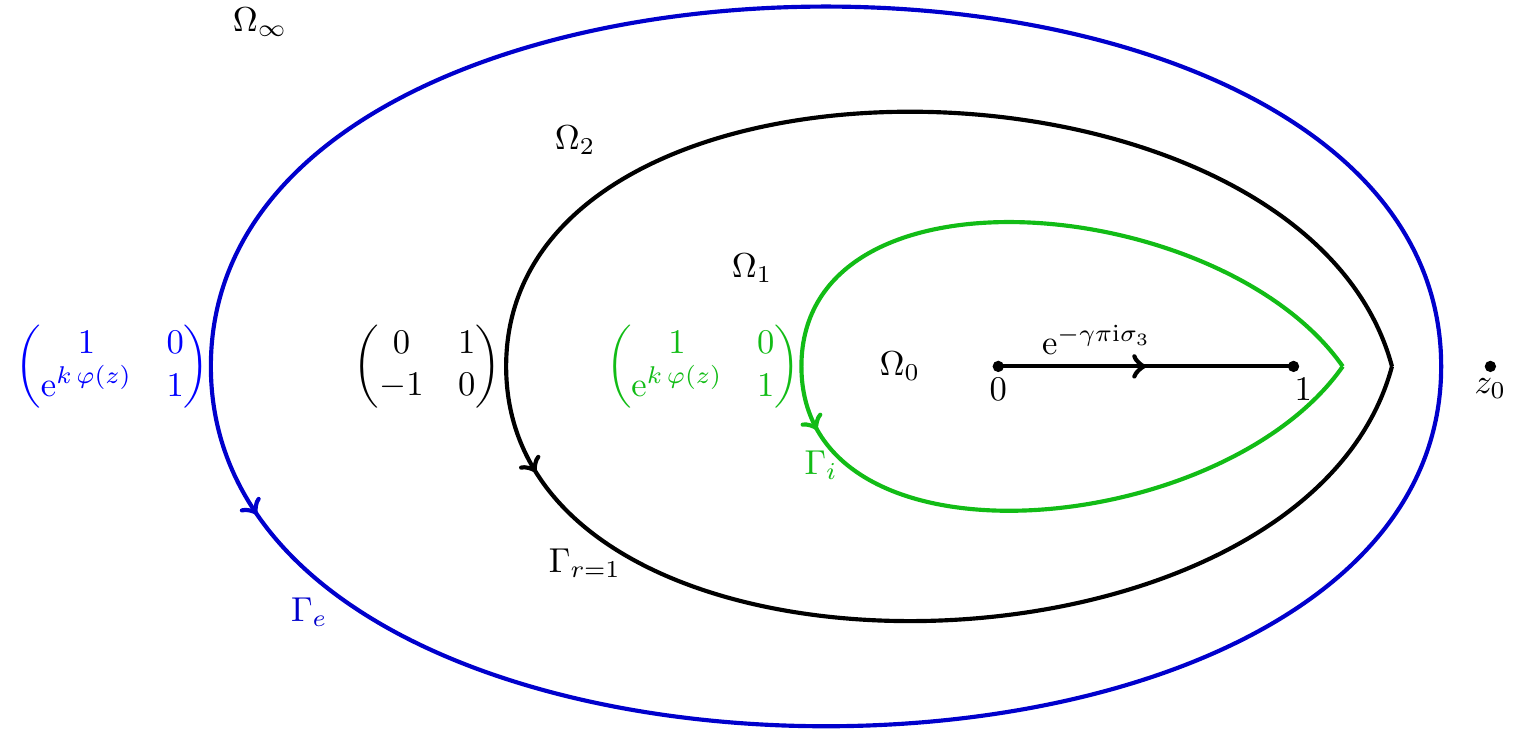}
\caption{Jumps in the $T (z)$ Riemann--Hilbert problem.}\label{fig: jumps T}
\end{figure}

\begin{problem}\label{TRHP}The matrix $T (z)$ is analytic in $ \mathbb{C} \setminus ( \Gamma_{i} \cup \Gamma_{r} \cup \Gamma_{e} \cup [0,1] ) $ and admits non-tangential boundary values. Moreover
\begin{enumerate}\itemsep=0pt
\item[$1.$] \textit{On $\Sigma_{T} = \Gamma_{i} \cup \Gamma_{r} \cup \Gamma_{e} \cup [0,1]$ the boundary values satisfy }
\begin{gather*}
T_{+} (z) = T_{-} (z) v_{T}, \qquad z \in \Sigma_{T},
\end{gather*}
where
\begin{gather}\label{eq: jump mat T}
v_{T} = \begin{cases}
\begin{pmatrix} 1 & 0 \\ {\rm e}^{k \varphi (z)} & 1 \end{pmatrix} , & z \in \Gamma_{i}, \vspace{1mm}\\
\begin{pmatrix} 0 & 1 \\ -1 & 0 \end{pmatrix}, & z \in \Gamma_{r}, \vspace{1mm}\\
\begin{pmatrix} 1 & 0 \\ {\rm e}^{k \varphi (z)} & 1 \end{pmatrix} , & z \in \Gamma_{e}, \\
{\rm e}^{-\gamma\pi {\rm i}\sigma_3}, & z \in (0,1).
\end{cases}
\end{gather}

\item[$2.$] \textit{Large $z$ boundary behaviour:}
\begin{gather*}
T (z) = \left( \mathbf{1} + \mathcal{O} \left( \frac{1}{z} \right) \right) , \qquad z \rightarrow \infty.
\end{gather*}

\item[$3.$] \textit{Endpoint behaviour:}
\begin{gather*}
T (z) = \mathcal{O} (1) z^{\frac{\gamma}{2}\sigma_3} , \qquad z \rightarrow 0 ,\\
T (z) = \mathcal{O} (1) (z-1)^{- \frac{\gamma}{2}\sigma_3} , \qquad z \to 1.
\end{gather*}
\end{enumerate}
\end{problem}

\subsection{Choice of contour}\label{Subsection: Choice of Contour}
Since our analysis will be performed in the critical case when $ z_0 \to 1$, the choice of the contour~$\Gamma_r$ with $0<r\leq z_0=1$ is forced (Fig.~\ref{fig:tikz_critical}).

\begin{figure}[t]\centering
\includegraphics[width=50mm]{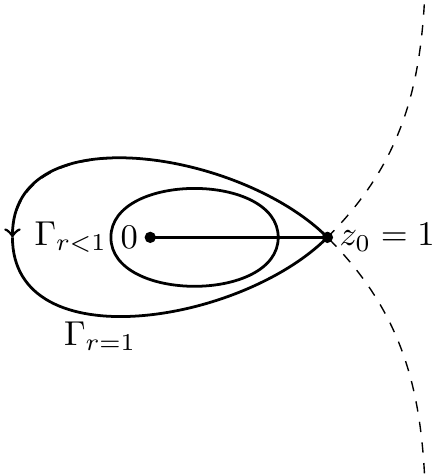}
\caption{The contours $\Gamma_r$ for $z_0=1$ and different values of $0<r\leq 1$.}\label{fig:tikz_critical}
\end{figure}

\begin{figure}[t]\centering
\includegraphics[width=80mm]{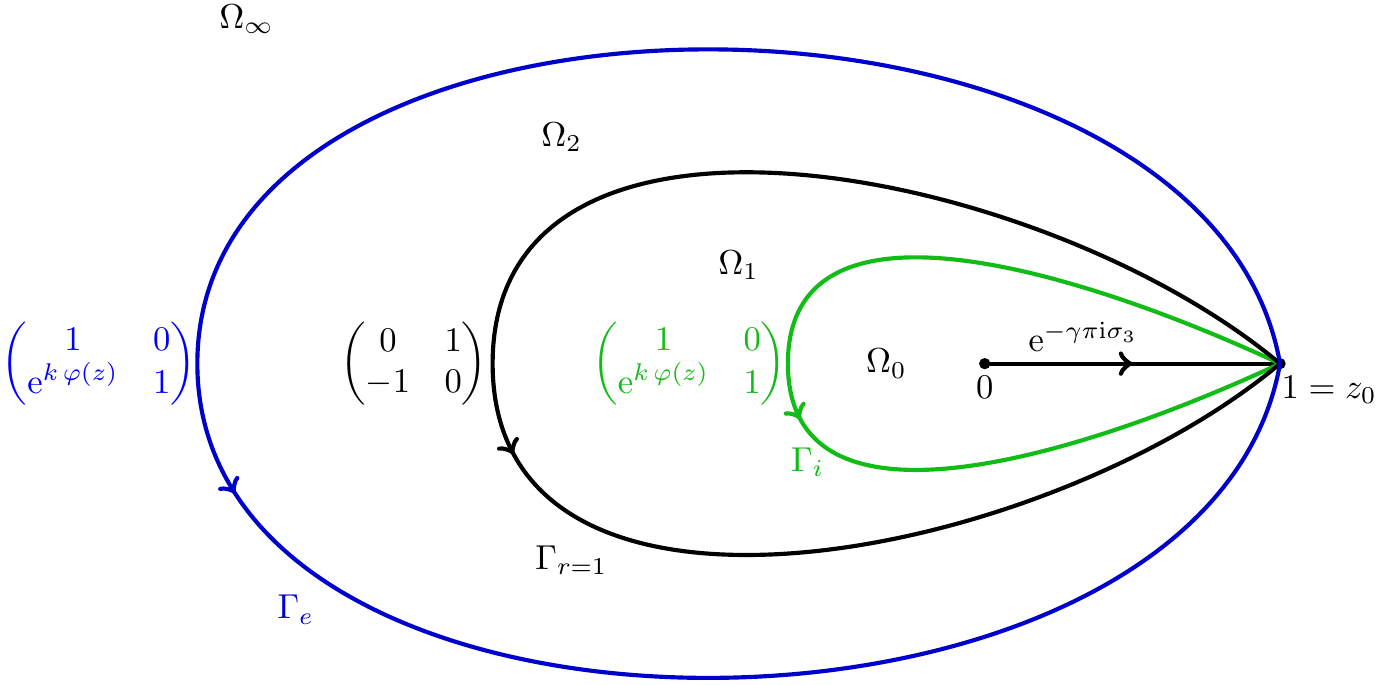}
\caption{Contours in the critical case $z_{0} = 1$.}\label{fig: critical contour0}
\end{figure}

The corresponding contours $\Gamma_i$ and $\Gamma_e$ have to be chosen in such a way that the jump matrix~$v_T$ defined in~\eqref{eq: jump mat T} converges (as $k\to\infty$ and hence also as $n\to \infty$) exponentially fast to a~constant. We make the choice of the deformed contours $\Gamma_i$ and $\Gamma_e$ as in Fig.~\ref{fig: critical contour0}. This is the correct choice because the sign of $\Re \varphi (z)$ on the contours~$\Gamma_i$ and~$\Gamma_e$ is negative (see Fig.~\ref{fig: critical contour} where the region $\Re \varphi (z) < 0$ is plotted).

Therefore, the jump matrices on $\Gamma_i$ and $\Gamma_e$ converge exponentially fast to the identity matrix in any $L^p$ norm except $p=\infty$ because $\varphi ( z = 1) = 0$.

\begin{figure}[t]\centering
\includegraphics[width=10cm]{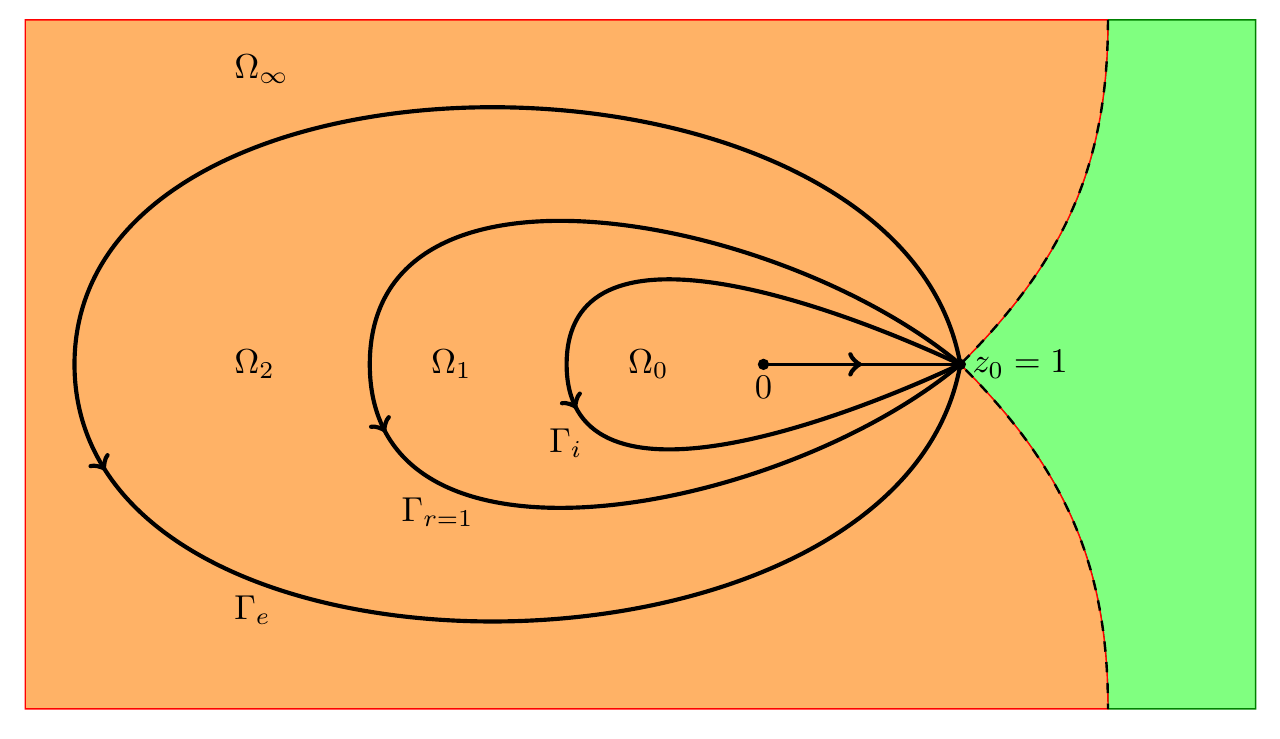}
\caption{Contours $\Gamma$ chosen for the critical case $z_{0} = r = 1$. The orange region corresponds to the part of the complex plane where $\Re{( \varphi (z))} < 0$ and the green contour is the region where $\Re{ ( \varphi (z))} > 0$.}\label{fig: critical contour}
\end{figure}

With this choice of contours we arrive to the following theorem (whose proof is a simple inspection and hence omitted)
\begin{Theorem} The jump matrix $v_{T}$ converges exponentially fast as $k \to \infty$ to the constant jump matrix
\begin{gather}\label{eq:v_infty}
v_{\infty} =
\begin{cases}
 \begin{pmatrix} 1 & 0 \\ 0 & 1 \end{pmatrix} , & z \in ( \Gamma_{e} \cup \Gamma_{i} ) \setminus \mathbb{D},\vspace{1mm} \\
\begin{pmatrix} 0 & 1 \\ -1 & 0 \end{pmatrix}, & z \in \Gamma_{1}, \\
{\rm e}^{-\gamma\pi {\rm i}\sigma_3}, & z \in (0,1) ,
\end{cases}
\end{gather}
where $ \mathbb{D} $ is a small {disk} surrounding the point $z = 1$.
\end{Theorem}

\subsection[Approximate solutions to $T (z)$]{Approximate solutions to $\boldsymbol{T (z)}$}\label{Subsection: Approximate solutions to T(z)}

We are now ready to approximate the matrix $T (z)$ with two solutions, one outside a neighbourhood of $z = 1$ and one inside. We call exterior parametrix the Riemann--Hilbert problem solved by the matrix $M (z)$ with jump $v_{\infty} (z)$. We call local parametrix the solution $P (z)$ of the Riemann--Hilbert problem obtained within a neighbourhood of $z = 1$. These two solutions are approximations of the exact solution, $T (z)$, in the limit $k \to \infty$. In order to obtain the asymptotics of the orthogonal polynomials $\pi_{k} (z)$, we need sub-leading corrections to the matrices ${N} (z)$ and $P (z)$ and this will be accomplished by evaluating perturbatively the error mat\-rix~$E (z)$, which is defined as
\begin{gather*}
E (z) = \begin{cases}
 T (z) {N} (z)^{-1} , & z \in \mathbb{C} \setminus \mathbb{D}, \\
T (z) P (z)^{-1}, & z \in \mathbb{D}.
\end{cases}
\end{gather*}

We will first construct the matrix ${N} (z)$ and then the matrix $P (z)$.

\subsubsection{Exterior parametrix}\label{Subsection: Exterior Parametrix}

The exterior parametrix, ${N} (z)$, is the piecewise analytic $2 \times 2$ matrix
\begin{gather*}
{N} (z) =\left( 1 - \frac{1}{z} \right)^{\frac{\gamma}{2} \sigma_{3}} \begin{pmatrix} 0 & 1 \\ -1 & 0 \end{pmatrix}^{\chi_L},
\end{gather*}
where $\chi_L$ is the characteristic function that is one on the left of the contour $\hat{\Gamma}_{r=1}$ and is zero otherwise. By design, it satisfies ${N} (\infty) = \mathbf 1$ and
\begin{gather*}
{N}_{+} (z) = {N}_{-} (z) v_{\infty}, \qquad z \in {\Gamma_1 \cup [0,1]},
\end{gather*}
where $v_{\infty}$ has been defined in~\eqref{eq:v_infty}.

\subsubsection{Local parametrix and double scaling limit}\label{Subsection: Local Parametrix and double scaling limit}

The local parametrix near the point $z = 1$ is the solution, $P (z)$, of a matrix Riemann--Hilbert problem that has the same jumps as $T (z)$ and that matches ${N} (z)$ on the boundary of a disc centered at $z = 1$.
\begin{problem}\label{RHP_P} The matrix $P (z)$ is analytic in $\mathbb{D} \setminus \Sigma_{P}$, $\Sigma_{P} = \Gamma_e \cup \Gamma_{r=1} \cup \Gamma_i \cup [0,1]$ and admits non-tangential boundary values. Moreover
\begin{enumerate}\itemsep=0pt
\item[$1.$] Jumps on $\Sigma_{P}$:
\begin{gather*}
P_{+} (z) = P_{-} (z) v_{P} (z), \qquad z \in \Sigma_{P}\cap \mathbb{D},
\end{gather*}
where{\samepage
\begin{gather*}
v_{P} (z) =
\begin{cases}
 \begin{pmatrix} 1 & 0 \\ {\rm e}^{k \varphi (z)} & 1 \end{pmatrix} , & z \in \Gamma_{i}\cap \mathbb{D},\vspace{1mm}\\
 \begin{pmatrix} 0 & 1 \\ -1 & 0 \end{pmatrix}, & z \in \Gamma_{r}\cap \mathbb{D},\vspace{1mm}\\
 \begin{pmatrix} 1 & 0 \\ {\rm e}^{k \varphi (z)} & 1 \end{pmatrix} , & z \in \Gamma_{e}\cap\mathbb{D}, \\
{\rm e}^{-\gamma\pi {\rm i}\sigma_3}, & z \in (0,1) \cap \mathbb{D},
\end{cases}
\end{gather*}
which are the jumps that can be seen in Fig.~{\rm \ref{fig: Improved Approx}}.}

\begin{figure}[t]\centering
\includegraphics[width=10.49cm]{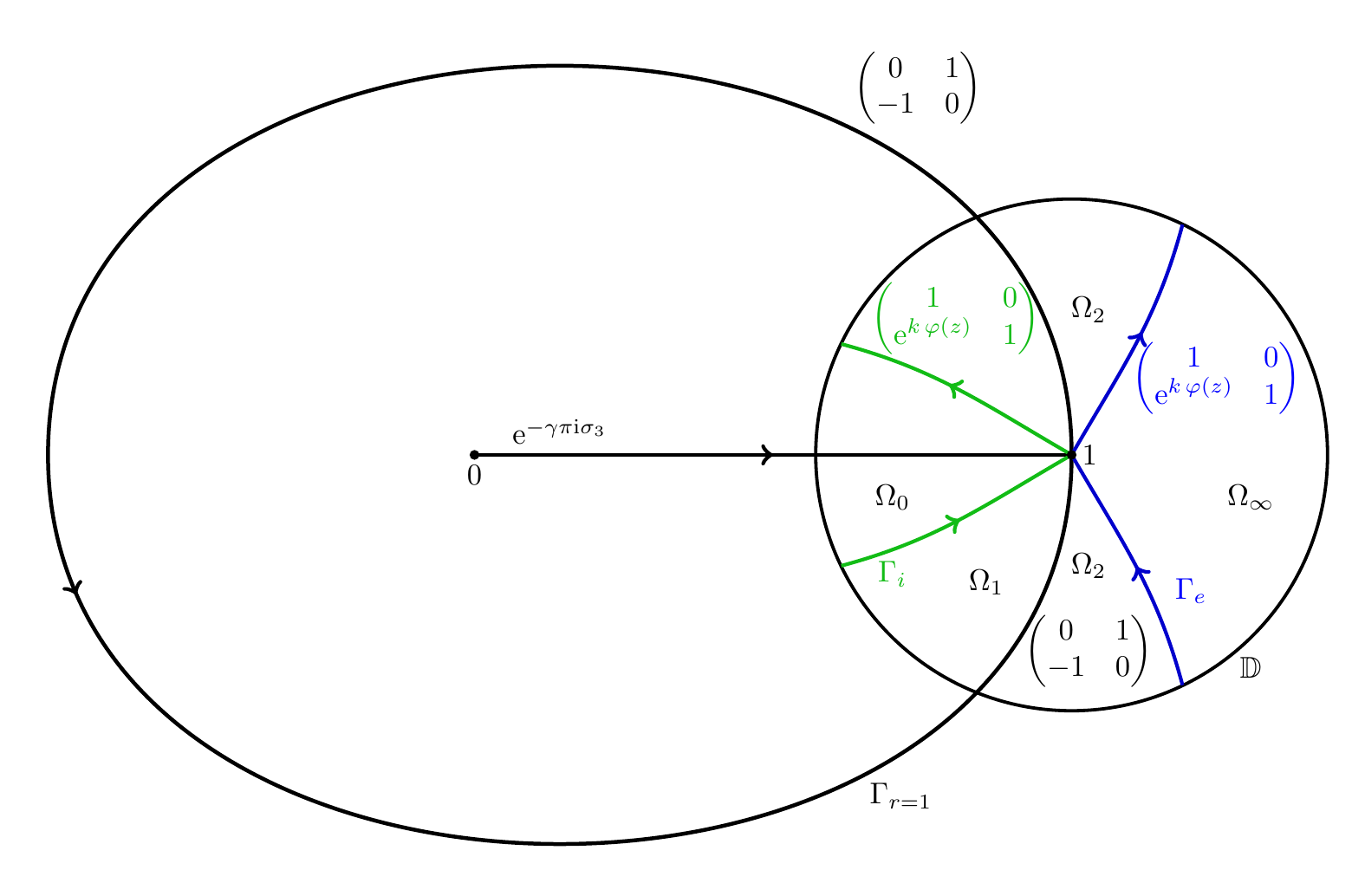}
\caption{Jumps of the Riemann--Hilbert problem for the matrix $P (z)$.}\label{fig: Improved Approx}
\end{figure}

\item[$2.$] Behaviour at the boundary $\partial \mathbb{D}$:
\begin{gather}\label{P_infty}
P (z) = {N} (z) ( \mathbf{1} + o (1)), \qquad \text{as} \quad k \to \infty \quad \text{and} \quad z\in \partial \mathbb{D}.
\end{gather}
\end{enumerate}
\end{problem}

\subsubsection{The double scaling limit}\label{Subsection: The Double Scaling limit for k varphi(z)}

Let $\phi(z;z_0)$ be defined by
\begin{gather*}
\phi (z; z_0) = \left( \frac{z - 1}{z_{0}} - \log{z} \right).
\end{gather*}
We note that $\phi(z;z_0)$ is the analytic continuation of $\varphi(z;z_0)$ \eqref{eq: phi function def} (setting $r=1$) from the exterior of $\Gamma= \Gamma_1$. For $z_0\sim 1$ we see that it has a single critical point in the neighbourhood of $z=1$; then
\begin{Proposition}\label{propnormal}
There exists a jointly analytic function $\zeta (z; z_0)$ which is univalent in a fixed neighbourhood of $z=1$, with $z_0$ in a neighbourhood of $z_0=1$, and an analytic function $A (z_0)$ near $z_0=1$ such that
\begin{gather}\label{11429}
\phi (z; z_0) = \frac 1 2 \zeta^2 (z; z_0) + A (z_0) \zeta (z; z_0) , \qquad \zeta (1;z_0) \equiv 0,\qquad A (1)=0.
\end{gather}
\end{Proposition}
\begin{proof}
The function $\phi(z;z_0)$ has a critical point at $z = z_0$ with critical value $\phi_{\rm cr} =\frac {z_0-1}{z_0} - \ln z_0$. Since we are interested in the values $z_0\simeq 1$ we observe that $\phi_{\rm cr} = -\frac 1 2 (z_0-1)^2 (1 + \mathcal O (z_0-1))$. Hence we can write
\begin{gather*}
\phi_{\rm cr} = -\frac{ A(z_0)^2}{2}
\end{gather*}
with $A(z_0)$ an analytic series near $z_0=1$. Moreover
\begin{gather*}
\phi(z;z_0)- \phi_{\rm cr} = \frac 1 {2z_0^2} (z-z_0)^2 (1 + \mathcal O(z-z_0)).
\end{gather*}
Define
\begin{gather*}
\frac {\zeta(z;z_0)}{\sqrt{2}} := \sqrt{\phi(z;z_0) +\frac{ A^2(z_0)}2 } - \frac{A(z_0)}{\sqrt{2}}.
\end{gather*}
This is a conformal map in a neighbourhood of $z=z_0$ of uniform radius of convergence for $z_0 \simeq 1$; moreover, from $\phi(1;z_0)=0$ follows that $\zeta(1;z_0)\equiv 0$. The equation \eqref{11429} is trivially satisfied.
\end{proof}

\begin{Remark}
A simple computation yields the following first few terms for the expansion ($\delta = z_0-1$)
\begin{gather*}
\zeta = \le(1 + \frac \delta 3 + \frac {\delta^2}4 \ri) (z-1) -\frac {(z-1)^2}3 \le(1 - \frac \delta 4 + \frac {11 \delta^2}{60}\ri) + \mathcal O \big( (z-1)^3, \delta^3 \big),\\
A (\delta) = -\delta + \frac 2 3 \delta^3 + \mathcal O \big(\delta^3\big).
\end{gather*}
From this it follows
\begin{gather*}
\sqrt{k} \zeta(z;z_0) = \le(z-1 - \frac{(z-1)^2}3 \ri) \sqrt{k} + \mathcal O(1),
\end{gather*}
which manifests the fact that $\sqrt{k} \zeta$ homothetically expands the image of the disk $\mathbb D$.
\end{Remark}

\begin{Definition}[double-scaling limit]
The double-scaling limit is defined by taking $k \rightarrow \infty$ and $A \rightarrow 0$ (or equivalently $z_{0} \rightarrow 1$) so that
\begin{gather*}
\lim_{k \to \infty, \, z_0 \to 1} \sqrt{k} A = \mathcal{S}
\end{gather*}
with $\mathcal{S}$ in compact subsets of the complex plane. Equivalently, it can be defined as
\begin{gather}\label{az0}
A \sim \frac{\mathcal{S}}{\sqrt{k}} \qquad \hbox{or} \qquad z_0 \sim{1 - \frac{\mathcal S}{\sqrt k} }.
\end{gather}
\end{Definition}

\subsubsection{Model problem and the Painlev\'e IV equation}\label{Subsection: Model problem and the Painleve IV equation}

We will now recall the Riemann--Hilbert Problem~\ref{PIV_0} for the function $\Psi (\lambda; s)$ associated to the Painlev\'e IV equation. For convenience, we replace the contour $\Gamma_1$ with $\hat{\Gamma}_i$ and the contour $\Gamma_{\infty}$ with $\hat{\Gamma}_e$ in order to be consistent with the notation we are using in this section.

Given that $\Theta_\infty = \Theta_0 = \frac \gamma 2$, the equations \eqref{isomono} and \eqref{def_H} reduce to a somewhat simpler form
\begin{gather}
U' = U (Y-s), \qquad Z = \dfrac{1}{2}\big( s Y + \gamma- Y'-Y^2\big),\qquad H =\left( s + \dfrac{\gamma }{Y} - Y\right)Z - \dfrac{Z^{2}}{Y},\label{isomono1}
\end{gather}
where $Y = Y (s)$ solves the Painlev\'e IV equation \eqref{P4} in the form
\begin{gather*}
Y'' = \frac{1}{2} \frac{(Y')^{2}}{Y} + \frac{3}{2} Y^{3} - 2 s Y^{2} + \left( 1 + \frac{s^{2}}{2} - \gamma \right) Y - \frac{\gamma^2}{2Y}.
\end{gather*}

For our purposes, we need to modify the Riemann--Hilbert problem for $\Psi (\lambda)$ to be able to match the Riemann--Hilbert problem for $P (z)$. Let us introduce the contour $\hat \Gamma_{r=1}$ on the $\lambda$ plane, which is a contour between $\hat{\Gamma}_i$ and $\hat{\Gamma}_e$ and passing through $\lambda=0$. Then, let us define
\begin{gather*}\widehat{\Psi} (\lambda) = \Psi (\lambda) {\rm e}^{\theta \sigma_3}\begin{pmatrix} 0 & 1 \\- 1 & 0 \end{pmatrix}^{\chi_L},
\end{gather*}
where $\theta$ is given by \eqref{eq: theta def} and $\chi_L$ is the characteristic function that is one on the left of the con\-tour~$\hat{\Gamma}_{r=1}$ and is zero otherwise. It is now straightforward to check that the Riemann--Hilbert problem for $\widehat{\Psi} (\lambda)$ is given by

\begin{problem}\label{Psihat}The matrix $\widehat{\Psi} (\lambda) \textrm{is analytic in } \mathbb{C} \setminus \Sigma_{\widehat{\Psi}}$, $\Sigma_{\widehat{\Psi}} = \hat{\Gamma}_e\cup\hat{\Gamma}_i\cup\hat{\Gamma}_{r=1}\cup\mathbb{R}^-$ and admits non-tangential boundary values. Moreover
\begin{enumerate}\itemsep=0pt
\item[$1.$] Jumps on $\Sigma_{\widehat{\Psi}}$:
\begin{gather*}
\widehat{\Psi}_{+} (\lambda) = \widehat{\Psi}_{-} (\lambda) v_{\widehat{\Psi}} (\lambda), \qquad \zeta \in \Sigma_{\widehat{\Psi}},\end{gather*}
where
\begin{gather}\label{eq: jump mat F}
v_{\widehat{\Psi}} (\lambda) =
\begin{cases}
 \begin{pmatrix} 1 & 0 \\ {\rm e}^{-2\theta(\lambda)} & 1 \end{pmatrix} , & z \in \hat{\Gamma}_{i},\vspace{1mm}\\
\begin{pmatrix} 0 & 1 \\ -1 & 0 \end{pmatrix}, & z \in \hat{\Gamma}_{r=1},\vspace{1mm}\\
\begin{pmatrix} 1 & 0 \\ {\rm e}^{2\theta(\lambda)} & 1 \end{pmatrix} , & z \in \hat{\Gamma}_{e},\\
 {\rm e}^{-\gamma\pi {\rm i}\sigma_3}, & z \in \mathbb{R}^-.
\end{cases}
\end{gather}

\item[$2.$] Behaviour for $\lambda\to\infty$
\begin{gather}\label{psihat_infty}
\widehat{\Psi} (\lambda) = \left( \mathbf{1} + \dfrac{\Psi_1}{\lambda} + \dfrac{\Psi_2}{\lambda^2} + \mathcal{O} \big( \lambda^{-3} \big) \right) \lambda^{\frac{\gamma}{2} \sigma_{3}} \begin{pmatrix} 0 & 1 \\ -1 & 0 \end{pmatrix}^{\chi_L}.
\end{gather}
\end{enumerate}
\end{problem}

Comparing the jump matrices for $P (z)$ and $\widehat{\Psi} (\lambda)$, we are now ready to obtain the local parametrix $P (z)$ defined by the Riemann--Hilbert Problem~\ref{RHP_P}, which is given by
\begin{gather}\label{P_parametrix}
P (z) = {N} (z) \begin{pmatrix}
0&-1\\1&0
\end{pmatrix}^{\chi_L} \big( \sqrt{k} \zeta (z) \big)^{-\frac{\gamma}{2} \sigma_3} \widehat{\Psi} \big( \sqrt{k} \zeta (z); \sqrt{k} A \big) ,
\end{gather}
where $\zeta(z) = \zeta(z;z_0)$ as defined in \eqref{propnormal}. We observe that the product of the first three terms of \eqref{P_parametrix} is holomorphic in the neighbourhood of $z=1$ and therefore does not change the Riemann--Hilbert problem. Furthermore, the first three terms have been inserted in order to have the behaviour \eqref{P_infty} for $z \in \partial \mathbb{D}$ in the limit $k \to \infty$ and $A \to 0$ in such a way that
\begin{gather*}
\lim_{k\to\infty} \sqrt{k} A=\mathcal{S}.
\end{gather*}
In the analysis above and below, we assume that $\mathcal{S}$ belongs to compact sets where the solution of the Painlev\'e IV equation does not have poles. Using \eqref{psihat_infty}, \eqref{Psi1} and \eqref{Psi2}, we obtain in the limit $k \to \infty$
\begin{gather}
P (z) = {N} (z)
\begin{pmatrix}
0&-1\\1&0
\end{pmatrix}^{\chi_L} \big( \sqrt{k} \zeta (z) \big)^{-\frac{\gamma}{2}\sigma_3} \nonumber\\
\hphantom{P (z)=}{}\times \left( \mathbf{1} + \dfrac{1}{\sqrt{k} \zeta} \begin{bmatrix} H & \dfrac{Z}{U} \vspace{1mm}\\ U & -H \end{bmatrix} + \mathcal{O} \big( k^{-1} \big) \right) \big( \sqrt{k} \zeta \big)^{\frac{\gamma}{2} \sigma_{3}} \begin{pmatrix} 0 & 1 \\- 1 & 0 \end{pmatrix}^{\chi_L}\nonumber\\
\hphantom{P (z)}{} = {N} (z) \!\left( \mathbf{1} + \begin{pmatrix}
0&-1\\1&0
\end{pmatrix}^{\chi_L}\! \dfrac{1}{\sqrt{k} \zeta} \!\begin{bmatrix} H & \big( \sqrt{k} \zeta \big)^{-\gamma} \dfrac{Z}{U} \vspace{1mm}\\ \big( \sqrt{k} \zeta \big)^{\gamma} U & -H \end{bmatrix}\! \begin{pmatrix} 0 & 1 \\- 1 & 0 \end{pmatrix}^{\chi_L}\! + \mathcal{O} \big( k^{-1+\gamma/2} \big) \right)\!\nonumber \\
\hphantom{P (z)}{} = {N} (z) \big( \mathbf{1} + \mathcal{O} \big( k^{\frac{\gamma-1}{2}} \big) \big).\label{1231}
\end{gather}
From the above expansion we can see that the subleading terms are not uniformly small with respect to $\gamma$ as $k\to\infty$ since $\gamma \in [0,1)$. Note that the source of the non-uniformity in the error analysis is the element $(2,1)$ in~\eqref{1231}. For this reason we need to introduce an improved parametrix in the next section.

\subsubsection{Improved parametrix}\label{Subsection: Improved Parametrix}

To construct an improved parametrix we define
\begin{gather}\label{374}
\widetilde{{N}} (z) = \left( \mathbf{1} + \frac{B \sigma_-}{z-1} \right) {N} (z)
,\qquad \sigma_-:= \left[
\begin{matrix}
0&0\\1&0
\end{matrix}\right] ,\qquad \sigma_+:= \left[
\begin{matrix}
0&1\\0&0
\end{matrix}\right] ,
\end{gather}
where $B$ is to be determined. In the same way we define
\begin{gather}\label{Ptilde}
\widetilde{P} (z) = \widetilde{{N}} (z) \begin{pmatrix}
0&-1\\1&0
\end{pmatrix}^{\chi_L} \big( \sqrt{k} \zeta (z) \big)^{-\frac{\gamma}{2}\sigma_3}\left(\mathbf{1} -\dfrac{U\sigma_-}{\sqrt{k} \zeta (z)} \right) \widehat{\Psi} \big( \sqrt{k} \zeta (z); \sqrt{k} A \big) ,
\end{gather}
where $U$ is the $(2 , 1)$ entry of the subleading term of the expansion of $\widehat{\Psi} (\lambda)$ for $\lambda\to\infty$. Now we have in the limit $k \to \infty$
\begin{gather*}
\begin{split}
& \widetilde{P} (z) = \widetilde{{N}} (z) \left( \mathbf{1} + \begin{bmatrix}
0&-1\\1&0
\end{bmatrix}^{\chi_L}\dfrac{1}{\sqrt{k} \zeta} \begin{bmatrix} H & (\sqrt{k} \zeta)^{-\gamma}\dfrac{Z}{U} \vspace{1mm}\\0 & -H \end{bmatrix} \begin{bmatrix} 0 & 1 \\- 1 & 0 \end{bmatrix}^{\chi_L} + \mathcal{O} \big( k^{-1+\gamma/2} \big) \right)\\
& \hphantom{\widetilde{P} (z)}{} = \widetilde{{N}} (z) \big( \mathbf{1} + \mathcal{O} \big( k^{-\frac{1}{2}} \big) \big).\end{split}
 \end{gather*}
The improved parametrics $\widetilde{{N}} (z)$ and $\widetilde{P} (z)$ have the same jump discontinuities as before, but~$\widetilde{P} (z)$ might have poles at $z=1$. The constant~$B$ in~\eqref{374} is then determined by the requirement that~$\widetilde{P} (z)$ is bounded at $z=1$. This gives the constant~$B$ as
\begin{gather*}
B = U k^{\frac{\gamma-1}{2}}.
\end{gather*}
We are now ready to compute the error matrix $E (z)$.

\subsubsection{Error matrix}\label{Subsection: Error matrix}

The error matrix $E (z)$ is defined as
\begin{gather}\label{def_E}
E (z) = \begin{cases}
 T (z) \widetilde{{N}} (z)^{-1} , & z \in \mathbb{C} \setminus \mathbb{D}, \\
 T (z) \widetilde{P} (z)^{-1}, & z \in \mathbb{D},
\end{cases}
\end{gather}
where the boundary of $\mathbb{D}$ is oriented clockwise. The matrix has also jumps on $\wh \Gamma_i \cup \wh \Gamma_e\setminus \mathbb D$ which are exponentially close to the identity matrix and can be ignored for the purposes of the analysis. Then, the matrix $E (z)$ satisfies the Riemann--Hilbert problem
\begin{gather*}
E_{+} (z) = E_{-} (z) v_{E} (z) ,\qquad z \in \partial \mathbb{D} ,\qquad v_{E} (z) = \widetilde{{N}} (z) \widetilde{P}^{-1} (z).\label{v_E}
\end{gather*}
In the double scaling limit $k\to\infty$, after using \eqref{eq: jump mat F}, \eqref{psihat_infty}, \eqref{Psi1}, \eqref{Psi2} and the last equation in \eqref{P_infty}, the jump matrix $v_{E} (z)$ takes the form
\begin{gather*}
 v_{E} (z) = \left(\mathbf{1} +\dfrac{U k^{\frac{\gamma-1}{2}}}{z-1}\sigma_-\right)\left(\dfrac{z-1}{\sqrt{k} z \zeta}\right)^{\frac{\gamma}{2}\sigma_3}\\
\hphantom{v_{E} (z) =}{} \times \left( \mathbf{1} - \dfrac{\begin{pmatrix} H & \dfrac{Z}{U} \vspace{1mm}\\ 0 & -H \end{pmatrix} }{\sqrt{k} \zeta}
 + \dfrac{\begin{bmatrix}
\dfrac{H^2 + s H - Z}{2} & \dfrac {H Z Y + H Y + Z Y^2}{U Y}\vspace{1mm}\\
s U - U Y & \dfrac{H^2 - s H + Z}{2}
\end{bmatrix}}{k \zeta^2} + \mathcal{O} \big( k^{-\frac{3}{2}} \big) \right)\\
\hphantom{v_{E} (z) =}{} \times
\left(\dfrac{z-1}{\sqrt{k} z \zeta}\right)^{-\frac{\gamma}{2}\sigma_3} \left( \mathbf{1}-\dfrac{Uk^{\frac{\gamma-1}{2}}}{z-1}\sigma_-\right) = \mathbf{1} + \dfrac{v_E^{(1)}}{\sqrt{k}} + \dfrac{v_E^{(2)}}{k^{\frac{1}{2} + \frac{\gamma}{2}}}+\dfrac{v_E^{(3)}}{k^{1-\frac{\gamma}{2}}} + \mathcal{O} \big( k^{-1} \big) ,
\end{gather*}
where the expansion is for $z\in \pa \mathbb D$ (while $v_E$ is exponentially close to the identity on $ \Gamma_E\setminus \mathbb D$), and
\begin{gather*}
v_E^{(1)} = - \dfrac{H}{\zeta}\sigma_3 ,\qquad v_E^{(2)} = - \left(\dfrac{z-1}{z \zeta}\right)^{\gamma}\left(\dfrac{Z}{U \zeta}\right) \sigma_+ ,\\
v_E^{(3)} = \left(\left(\dfrac{z-1}{z \zeta}\right)^{-\gamma}\dfrac{ (s - Y) U}{\zeta^2} - 2\dfrac{H U}{(z-1) \zeta} \right) \sigma_- ,
\end{gather*}
where $\sigma_+=\left(\begin{smallmatrix}0&1\\0&0\end{smallmatrix}\right)$. By the standard theory of small norm Riemann--Hilbert problems (see for example \cite[Chapter~7]{DeiftBook}), one has a similar expansion for~$E (z)$, namely
\begin{gather*}
E (z) = \mathbf{1} + \dfrac{E^{(1)}}{\sqrt{k}}+\dfrac{E^{(2)}}{k^{\frac{1}{2} + \frac{\gamma}{2}}} + \dfrac{E^{(3)}}{k^{1 - \frac{\gamma}{2}}} + \mathcal{O} \big( k^{-1} \big) ,
\end{gather*}
so that
\begin{gather*}
 E^{(j)}_{+} (z) = E^{(j)}_{-} (z) + v_{E}^{(j)} (z) , \qquad z \in \partial \mathbb{D} , \qquad j=1,2,3.
\end{gather*}
By solving the corresponding Riemann--Hilbert problem, we obtain, using the Plemelj--Sokhtski formula
\begin{gather*}
E^{(j)} (z) = \dfrac{1}{2 \pi {\rm i}}\int_{\partial\mathbb{D}} \dfrac{v_{E}^{(i)} (\xi)}{\xi - z} \mathrm{d} \xi,\qquad j = 1, 2, 3,
\end{gather*}
which gives
\begin{gather}
E^{(1)} (z) = - \dfrac{\mbox{Res}_{\xi=1} v_{E}^{(1)} (\xi)}{z-1} = \frac {H \sigma_3}{z-1} , \qquad z \in \mathbb{C} \setminus \mathbb{D} , \nonumber\\
E^{(2)} (z) = - \dfrac{\mbox{Res}_{\xi=1} v_{E}^{(2)} (\xi)}{z-1} = \frac {(Z/U) \sigma_+}{z-1} , \qquad z \in \mathbb{C} \setminus \mathbb{D} , \nonumber\\
E^{(3)} (z) = - \dfrac{\mbox{Res}_{\xi=1} v_{E}^{(3)} (\xi)}{(z-1)} - \dfrac{\mbox{Res}_{\xi=1} (\xi-1) v_{E}^{(2)} (\xi)}{(z-1)^2}\nonumber\\
\hphantom{E^{(3)} (z)}{} = \left( \dfrac{2}{3} \dfrac{\gamma (2 H + \gamma - s) \gamma + H + \gamma-s}{z-1} +\dfrac{U (2 H + Y - s)}{z-1} \right) \sigma_- , \qquad z \in \mathbb{C} \setminus \mathbb{D},\label{Eout}
\end{gather}
and
\begin{gather}
 E^{(1)} (z) = v_{E}^{(1)} - \dfrac{\mbox{Res}_{\xi = 1} v_{E}^{(1)} (\xi)}{z-1} = v_E^{(1)} + \frac{H\sigma_3}{z-1} , \qquad z \in \mathbb{D} ,\nonumber \\
E^{(2)} (z) = v_{E}^{(2)} - \dfrac{\mbox{Res}_{\xi=1} v_{E}^{(2)} (\xi)}{z-1} = v_{E}^{(2)} + \frac {(Z/U) \sigma_+}{z-1} ,
\qquad z \in \mathbb{D} , \label{Ein}\\
 E^{(3)} (z) = v_E^{(3)} + \left(\dfrac{2}{3}\dfrac{\gamma ( 2 H + \gamma - s) \gamma + H + \gamma - s}{z - 1} + \dfrac{U ( 2 H + Y - s)}{z-1} \right) \sigma_{-} , \qquad z \in \mathbb{D}.
\nonumber
\end{gather}

\subsection[Asymptotics for the polynomials $\pi_{k} (z)$ and proof of Theorems~\ref{theorem1} and~\ref{theorem2}]{Asymptotics for the polynomials $\boldsymbol{\pi_{k} (z)}$ and proof of Theorems~\ref{theorem1} and~\ref{theorem2}} \label{Subsection: Asymptotics for the polynomials pi_k ( z ) and proof of Theorem1 and Theorem2}

We are now ready to determine the asymptotic expansions for the orthogonal polynomials
\begin{gather*}
\pi_{k} (z) = {\rm e}^{k g (z)} \left( 1 - \frac{1}{z} \right)^{\frac{\gamma}{2}} [U (z)]_{1 1}.
\end{gather*}
Using \eqref{pi_U}, \eqref{def_T} and \eqref{def_E} we have
\begin{gather*}
\pi_{k} (z) = {\rm e}^{k g (z)} \left( 1 - \frac{1}{z} \right)^{\frac{\gamma}{2}}
\begin{cases}
\big[ E (z) \widetilde{{N}} (z) \big]_{1 1} , & z \in ( \Omega_{\infty} \cup \Omega_{0} ) \setminus \mathbb{D},\vspace{1mm} \\
\left[E (z) \widetilde{{N}} (z) \begin{pmatrix} 1 & 0 \\ {\rm e}^{k \varphi (z)} & 1 \end{pmatrix} \right]_{1 1} , & z \in \Omega_{1} \setminus \mathbb{D}, \vspace{1mm}\\
\left[E (z) \widetilde{{N}} (z) \begin{pmatrix} 1 & 0 \\ -{\rm e}^{k \varphi (z)} & 1 \end{pmatrix}\right]_{1 1} , & z \in \Omega_{2} \setminus \mathbb{D}, \vspace{1mm}\\
 \big[E (z) \widetilde{P} (z) \big]_{1 1} , & z \in ( \Omega_{0} \cup \Omega_{\infty} ) \cap \mathbb{D}, \vspace{1mm}\\
\left[E (z) \widetilde{P} (z) \begin{pmatrix} 1 & 0 \\ {\rm e}^{k \varphi (z)} & 1 \end{pmatrix} \right]_{1 1} , & z \in \Omega_{1} \cap \mathbb{D}, \vspace{1mm}\\
\left[E (z) \widetilde{P} (z) \begin{pmatrix} 1 & 0 \\ -{\rm e}^{k \varphi (z)} & 1 \end{pmatrix} \right]_{1 1} , & z \in \Omega_{2} \cap \mathbb{D}.
\end{cases}
\end{gather*}
We want to analyze each distinct region. Using \eqref{Eout}, \eqref{Ein}, \eqref{Ptilde} and \eqref{374} we obtain the following expressions:

{\bf The region $\boldsymbol{\Omega_{\infty} \setminus \mathbb{D}}$}
\begin{gather}
\pi_{k} (z) = {\rm e}^{k g (z)} \left( 1 - \frac{1}{z} \right)^{\gamma} \left( 1 + \dfrac{H}{\sqrt{k} (z-1)} + \mathcal{O} \left( \frac{1}{k} \right) \right)\nonumber\\
\hphantom{\pi_{k} (z)}{}
 = z^k \left( 1 - \frac{1}{z} \right)^{\gamma} \left( 1 + \dfrac{H}{\sqrt{k} (z-1)} + \mathcal{O} \left( \frac{1}{k} \right) \right),\label{exp_Omega_infty}
\end{gather}
with $H$ defined in \eqref{isomono1}.

{\bf The region $\boldsymbol{\Omega_{0} \setminus \mathbb{D}}$}{\samepage
\begin{gather*}
\pi_{k} (z) = {\rm e}^{k g (z)}\left( - \dfrac{Z}{U (z-1) k^{\frac{1}{2} + \gamma}} + \mathcal{O} \left( \frac{1}{k} \right) \right) ,
\end{gather*}
with $Z$ and $U$ defined in \eqref{isomono1}.}

{\bf The region $\boldsymbol{\Omega_{1} \setminus \mathbb{D}}$}
\begin{gather}
\pi_{k} (z) = {\rm e}^{k g (z)} \left( \frac{z - 1}{z} \right)^{\gamma} \nonumber\\
\hphantom{\pi_{k} (z) =}{}\times \left( {\rm e}^{k \varphi (z)} \left( 1 + \dfrac{H}{\sqrt{k} (z-1)} \right) - \dfrac{Z}{U (z-1) k^{\frac{1 + \gamma}{2}}} \left( \frac{z - 1}{z} \right)^{-\gamma} + \mathcal{O} \left( \frac{1}{k} \right) \right),\label{Omega1}
\end{gather}
with $H$, $U$ and $Z$ defined in (\ref{isomono1}). In a similar way we can obtain the expansion in the region~$\Omega_{2} \setminus \mathbb{D}$.

{\bf The region $\boldsymbol{\Omega_{2}\setminus \mathbb{D}}$}
\begin{gather}\label{Omega2}
\pi_{k} (z) = {\rm e}^{k g (z)} \left( \frac{z - 1}{z} \right)^{\gamma}
\left( 1 + \dfrac{H}{\sqrt{k} (z-1)} - \dfrac{Z {\rm e}^{k \varphi (z)}}{U (z-1) k^{\frac{1 + \gamma}{2}}} \left( \frac{z-1}{z} \right)^{-\gamma} + \mathcal{O} \left( \frac{1}{k} \right) \right).
\end{gather}

{\bf The region $\boldsymbol{\mathbb{D}}$.} In the region $( \Omega_{0} \cup \Omega_{\infty}) \cap \mathbb{D}$ we have{\samepage
\begin{gather*}
\pi_k (z) = {\rm e}^{k g (z)} \left( \frac{z - 1}{z} \right)^{\gamma}\left( \dfrac{\widehat{\Psi}_{1 1} \big( \sqrt{k} \zeta (z) ; \sqrt{k} A \big)}{k^{\frac{\gamma}{4}} \zeta (z)^{\frac{\gamma}{2}}} + \mathcal{O} \left( \frac{1}{k^{\frac{1}{2} + \frac{\gamma}{4} } } \right) \right),
\end{gather*}
where $\widehat{\Psi}_{1 1}$ is the $1 1$ entry of the Painlev\'e isomonodromic problem~\eqref{Psihat}.}

In the region $\Omega_{1} \cap \mathbb{D}$ and $\Omega_{2} \cap \mathbb{D}$ we have
\begin{gather*}
\pi_{k} (z) = {\rm e}^{k g (z)} \left( \frac{z - 1}{z} \right)^{\gamma} \\
\hphantom{\pi_{k} (z) =}{}\times \left( \dfrac{\widehat{\Psi}_{1 1} \big( \sqrt{k} \zeta (z) ; \sqrt{k} A \big) \pm {\rm e}^{k \varphi (z)} \widehat{\Psi}_{1 2} \big( \sqrt{k} \zeta (z) ; \sqrt{k} A \big)}{k^{\frac{\gamma}{4}} \zeta (z)^{\frac{\gamma}{2}}} + \mathcal{O} \left( \frac{1}{k^{\frac{1}{2} + \frac{\gamma}{4} } } \right) \right),
\end{gather*}
where $\pm$ refers to the region $\Omega_{1} $ and $\Omega_{2}$, respectively, and $\widehat{\Psi}_{1 2}$ is the $1 2$ entry of the solution Painlev\'e isomonodromic problem \eqref{Psihat}. Making the change of variables $
z = 1 - \frac{\lambda^d}{t_{c}}$, the proof of Theorem~\ref{theorem2} follows in a straightforward way from the above expansions. With these expansions, we are able to locate the zeros of the orthogonal polynomials.
\begin{figure}[t]\centering
\includegraphics[width=0.29\textwidth]{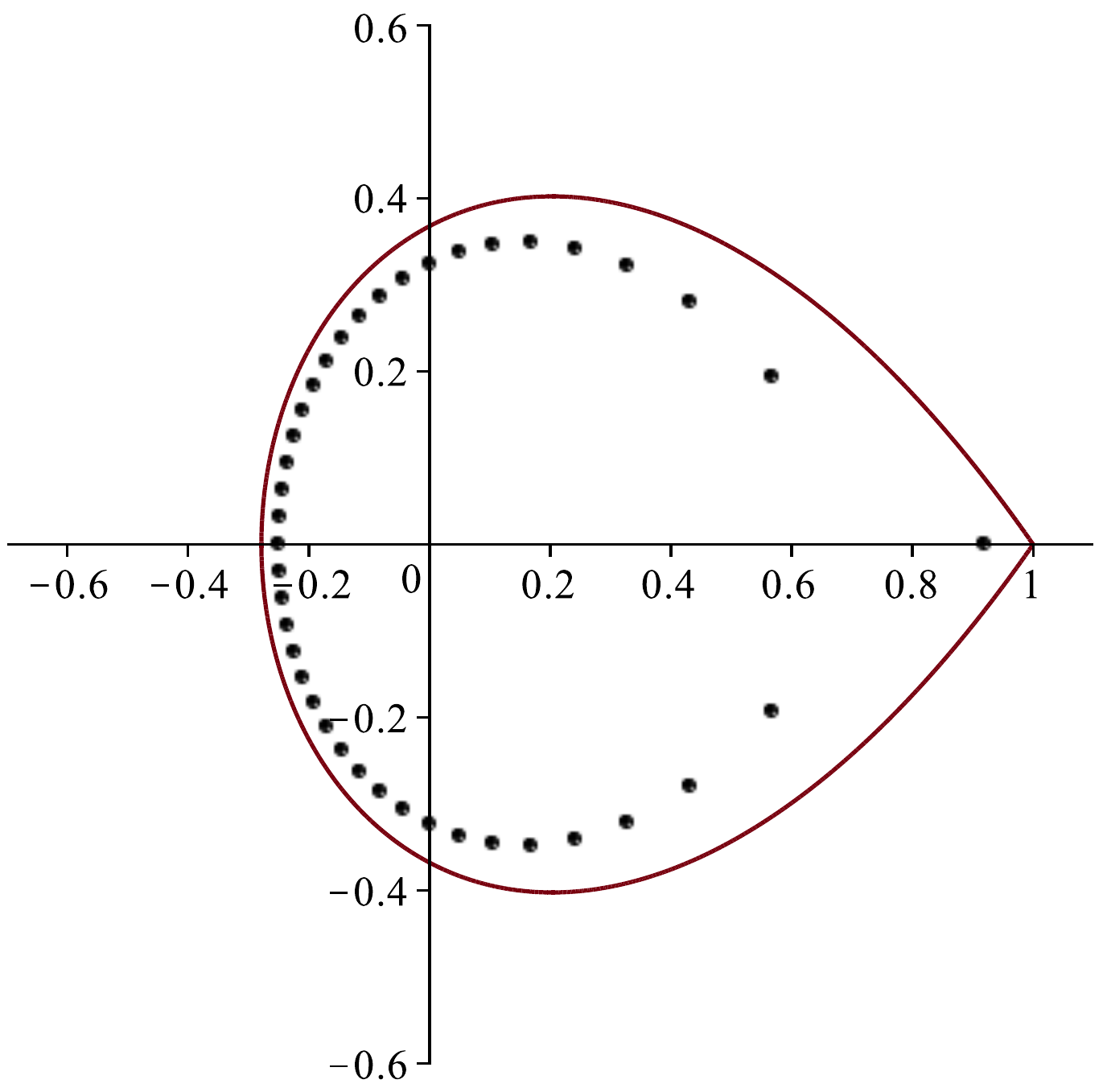}\quad
\includegraphics[width=0.29\textwidth]{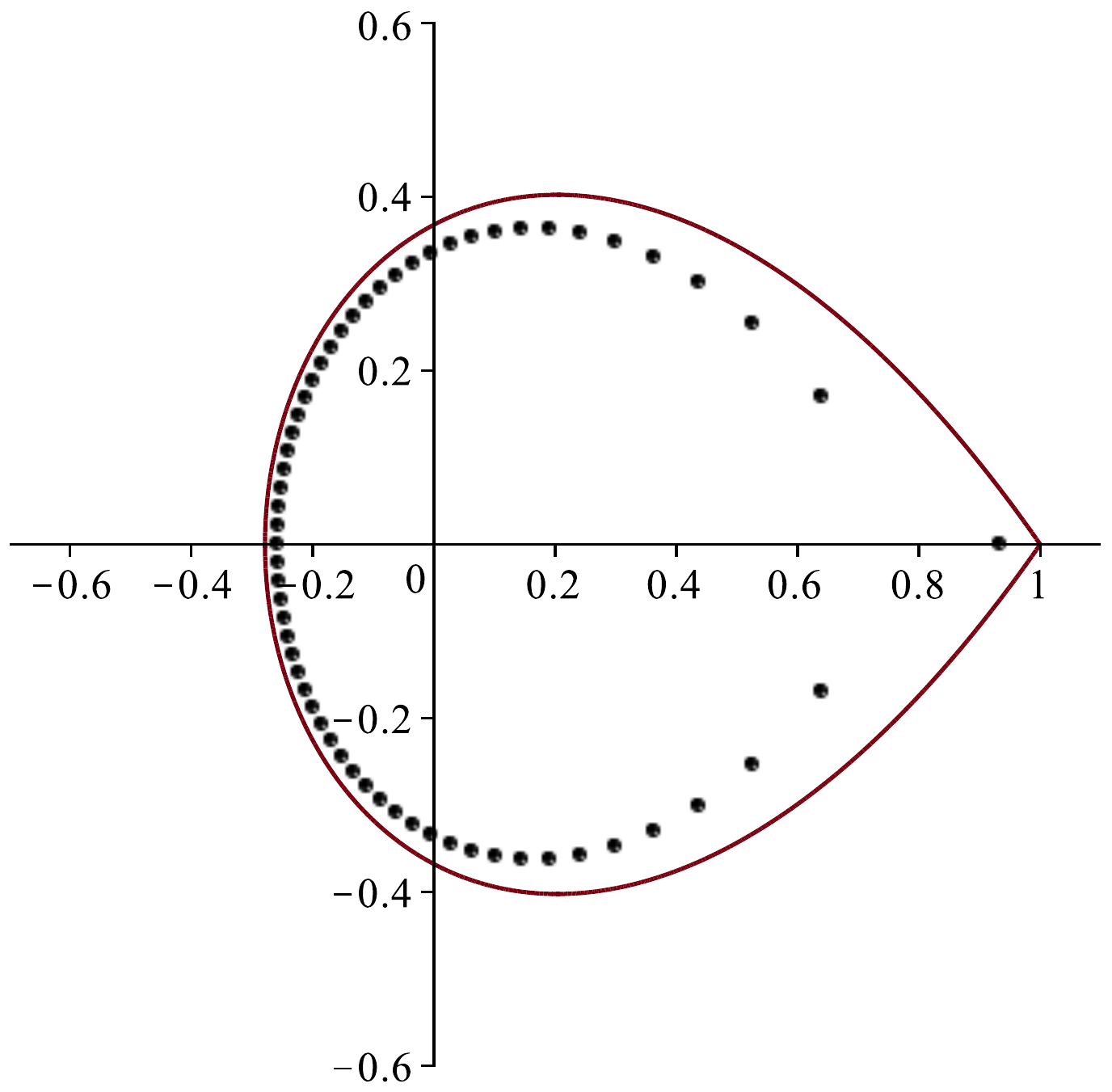}\quad
\includegraphics[width=0.29\textwidth]{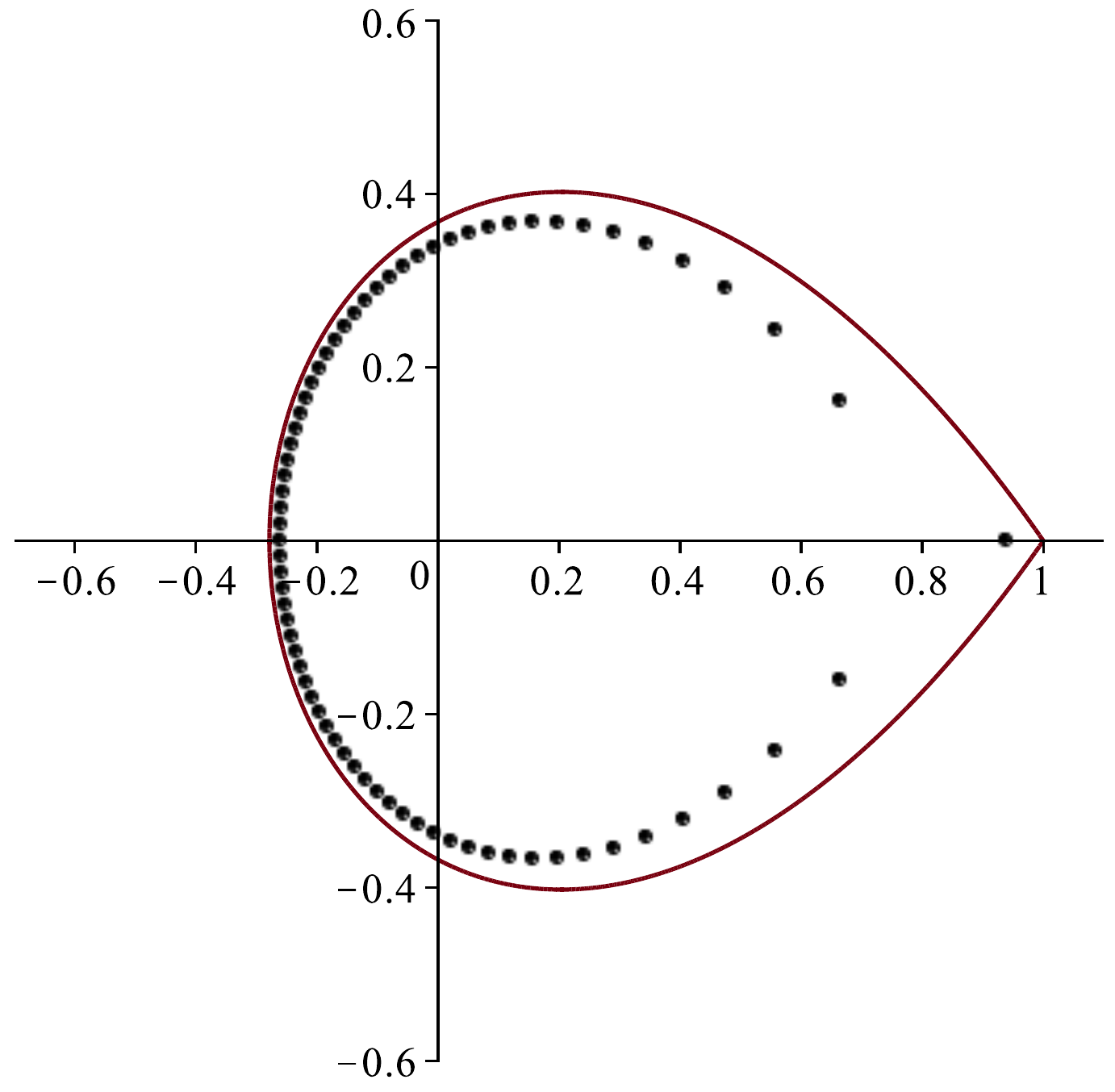}
\caption{Zeros of the polynomials $\pi_k (z)$ for $k = 40$, $60$, $70$. The values of the parameters are $z_0 = 1$ and $t = 2$, $d = 3$ and $\ell = 0$. The plot of the support of the limiting measure (Szeg\"o curve) of the zeros of the orthogonal polynomials $\pi_k (z)$ is in red.}\label{fig: zero_pik}
\end{figure}

\begin{Proposition}\label{propo_zeros_pre}The support of the counting measure of the zeros of the polynomials $\pi_k (z)$ outside an arbitrary small disk $\mathbb{D}$ surrounding the point $z = 1$ tends uniformly to the curve $\Gamma_{r=1}$ defined in \eqref{eq: fam of contours} for $z_0 = 1$. The zeros are within a distance $ o( 1 / k)$ from the curve defined by
\begin{gather}\label{phi_deform}
\log |z| - \frac{|z - 1|}{|z_0|} = - \dfrac{1 + \gamma}{2} \dfrac{\log k}{k} + \dfrac{1}{k} \log \left( \left| \dfrac{z}{z-1} \right|^{\gamma} \left| \dfrac{Z (\mathcal{S})}{(z-1) U (\mathcal{S})} \right| \right) ,
\end{gather}
where we recall from \eqref{az0} that $z_0 = \frac{\sqrt{k}}{\sqrt{k} + \mathcal{S}}$ and that the function $Z = Z (\mathcal{S})$ and $U = U (\mathcal{S})$ are related to the Painlev\'e~IV equation via~\eqref{isomono1}. The curves in~\eqref{phi_deform} approach $\Gamma_{r=1}$ at the rate $\mathcal{O} ( \log k / k)$ and lie in $\inte ( \Gamma_{r=1})$. The normalized counting measure of the zeros of $\pi_{k} (z)$ converges to the probability measure~$\nu$ defined in~\eqref{nur}.
\end{Proposition}
\begin{proof}Observing the asymptotic expansion \eqref{exp_Omega_infty} of $\pi_{k} (z)$ in $\Omega_\infty \setminus {\mathbb{D}}$, it is clear that $\pi_{k} (z)$ does not have any zeros in that region, since $z = 0$ and $z = 1$ do not belong to $\Omega_\infty \setminus {\mathbb{D}}$. The same reasoning applies to the region $\Omega_0 \setminus {\mathbb{D}}$, where there are no zeros of $\pi_{k} (z)$ for $k$ sufficiently large.

From the relations \eqref{Omega1} and \eqref{Omega2}, one has that in $\Omega_1 \cup \Omega_2$ using the explicit expression of~$g (z)$ defined in~\eqref{gint}
\begin{gather*}
\pi_{k} (z) = z^{k} \left( \frac{z - 1}{z} \right)^{\gamma} \left( 1 + \dfrac{H}{\sqrt{k} (z-1)} - \dfrac{{\rm e}^{\pm k \varphi (z)} Z}{U (z-1) k^{\frac{1+\gamma}{2}}} \left( \frac{z - 1}{z} \right)^{-\gamma} + \mathcal{O} \left( \frac{1}{k} \right) \right) ,
\end{gather*}
where $\pm$ refers to $\Omega_2$ and $\Omega_1$, respectively. The zeros of $\pi_{k} (z)$ may only lie asymptotically where the expression
\begin{gather*}
1 + \dfrac{H}{\sqrt{k} (z-1)} = \dfrac{{\rm e}^{\pm k \varphi (z)} Z}{U (z-1) k^{\frac{1 + \gamma}{2}}} \left( \frac{z - 1}{z} \right)^{-\gamma} , \qquad z \in \Omega_1 \cup \Omega_2.
\end{gather*}
Since $\Omega_2 \cup \Omega_1 \subset \{ \Re (\varphi) \leq 0 \}$, it follows that the zeros of $\pi_{k} (z)$ may lie only in the region $\Omega_1$ and such that $\Re \varphi (z) = \mathcal{O} ( \log k / k )$ (where $z_0$ is the value given by the double scaling~\eqref{az0}). Taking the logarithm of the modulus of the above equality, we obtain
\begin{gather*}
\Re \varphi (z) = - \dfrac{1 + \gamma}{2} \dfrac{\log k}{k} + \dfrac{1}{k} \log \left( \left|\dfrac{z}{z - 1} \right|^{\gamma} \left| \dfrac{Z (\mathcal{S})}{(z-1) U (\mathcal{S})} \right| \right) + \dfrac{1}{k^{\frac{3}{2}}} \Re \left(\dfrac{H (\mathcal{S})}{z - 1} \right) + \mathcal{O} \left( \frac{1}{k^2}\right).
\end{gather*}
Namely, the zeros of the polynomials $\pi_{k} (z)$ lie on the curve given by \eqref{phi_deform} with an error of order $\mathcal{O} \big(1/k^2\big)$. Such curves converge to the curve $\Gamma_{r=1}$ defined in \eqref{eq: fam of contours} with $z_{0} = 1$ at a rate $\mathcal{O} (\log k/k)$.
\end{proof}

The proof of Proposition~\ref{propo_zeros_pre_0} follows immediately from the proof of Proposition~\ref{propo_zeros_pre}. The rest of the proof of Theorem~\ref{theorem1} follows the steps obtained in \cite{BGM} and for this reason we omit it.

\subsection*{Acknowledgements}
The authors wish to thank the anonymous referees for their many suggestions for improving this manuscript. T.G.~acknowledges the support of the H2020-MSCA-RISE-2017 PROJECT No.~778010 IPADEGAN. M.B.~acknowledges the support by the Natural Sciences and Engineering Research Council of Canada (NSERC) grant RGPIN-2016-06660 and the FQRNT grant ``Matrices Al\'eatoires, Processus Stochastiques et Syst\`emes Int\'egrables" (2013--PR--166790).

\pdfbookmark[1]{References}{ref}
\LastPageEnding

\end{document}